\documentclass[prodmode]{acmsmall-ec16-ARXIV-EDIT}  

\usepackage[group-separator={,}]{siunitx}
\usepackage{amsmath,amsfonts,amssymb}
\usepackage{graphicx}
\usepackage{url}
\usepackage{mathtools}
\usepackage[textsize=tiny,textwidth=1cm,shadow]{todonotes}
\usepackage{thmtools, thm-restate}
\usepackage{enumitem}
\usepackage{pdflscape}

\usepackage{etoolbox}
\newtoggle{isFullVersion}
\toggletrue{isFullVersion}

\usepackage{algorithm}
\usepackage{algpseudocode}
\usepackage{algorithmicx}
\let\oldReturn\Return
\renewcommand{\Return}{\State\oldReturn}
\algtext*{EndWhile}
\algtext*{EndIf}
\algtext*{EndForAll}
\algtext*{EndFor}
\algtext*{EndFunction}

\usepackage{pgf}
\usepackage{tikz}
\usetikzlibrary{shadows,arrows,decorations,decorations.shapes,backgrounds,shapes,snakes,automata,fit,petri,shapes.multipart,calc,positioning,shapes.geometric}
  \tikzstyle{whitesq}=[rectangle,draw=black,fill=white,thin,inner sep=1.5pt,minimum size=6mm]
  \tikzstyle{whitecirc}=[circle,draw=black,fill=white,thin,inner sep=1.5pt,minimum size=6mm]
  \tikzstyle{fake}=[circle,draw=black,fill=white,thin,inner sep=1.5pt,minimum size=6mm,opacity=0.0]

  \newcommand{\squishlist}{
 \begin{list}{$\bullet$}
  { \setlength{\itemsep}{0pt}
     \setlength{\parsep}{3pt}
     \setlength{\topsep}{3pt}
     \setlength{\partopsep}{0pt}
     \setlength{\leftmargin}{1.5em}
     \setlength{\labelwidth}{1em}
     \setlength{\labelsep}{0.5em} } }
\doi{2940716.2940759}

\newcommand{\zCF}{Z_{\mathit{CF}}}        
\newcommand{\zPIEF}{Z_{\mathit{PIEF}}}    
\newcommand{\zPICEF}{Z_{\mathit{PICEF}}}  
\newcommand{\zHPIEF}{Z_{\mathit{HPIEF}}}  
\newcommand{\OPT}{\mathit{OPT}}

\newcommand{\PIEFRR}{PIEFR\textsuperscript{2}} 
\newcommand{\KRR}{\K^\text{red\textsuperscript{2}}} 
\newcommand{\cycles}{\mathcal{C}}
\newcommand{\chains}{\mathcal{C}'}
\newcommand{\K}{\mathcal{K}}

\tikzset{patient/.style={fill=white,draw=black,text=black,circle,inner sep=.1cm}}
\tikzset{ndd/.style={fill=white,draw=black,text=black,inner sep=.15cm}}
\tikzset{edge/.style = {->,> = latex'}}
\tikzset{thickedge/.style = {->,> = latex', line width=2.5pt}}

\usepackage{xcolor}

\newcommand{\algABS}{\textsc{BnP-DFS}}
\newcommand{\algPDS}{\textsc{BnP-Poly}}
\newcommand{\algRMA}{\textsc{CG-TSP}}
\newcommand{\algKlim}{\textsc{BnP-DCD}}
\newcommand{\algPICEF}{\textsc{PICEF}}
\newcommand{\algHPIEF}{\textsc{HPIEF}}
\newcommand{\algPICEFBnP}{\textsc{BnP-PICEF}}

\begin{document}

\iftoggle{isFullVersion}{%
\markboth{J.P.\ Dickerson, D.F.\ Manlove, B.\ Plaut, T.\ Sandholm and J.\ Trimble}{Position-Indexed Formulations for Kidney Exchange}
\pagestyle{headings}
}

\title{Position-Indexed Formulations for Kidney Exchange}

\author{%
  JOHN P. DICKERSON
  \affil{Carnegie Mellon University}%
  DAVID F. MANLOVE
  \affil{University of Glasgow}%
  BENJAMIN PLAUT
  \affil{Carnegie Mellon University}%
  TUOMAS SANDHOLM
  \affil{Carnegie Mellon University}%
  JAMES TRIMBLE
  \affil{University of Glasgow}%
}

\begin{abstract}
A kidney exchange is an organized barter market where patients in need of a kidney swap willing but incompatible donors.  Determining an optimal set of exchanges is theoretically and empirically hard.  Traditionally, exchanges took place in cycles, with each participating patient-donor pair both giving and receiving a kidney.  The recent introduction of chains, where a donor without a paired patient triggers a sequence of donations without requiring a kidney in return, increased the efficacy of fielded kidney exchanges---while also dramatically raising the empirical computational hardness of clearing the market in practice.  While chains can be quite long, unbounded-length chains are not desirable: planned donations can fail before transplant for a variety of reasons, and the failure of a single donation causes the rest of that chain to fail, so parallel shorter chains are better in practice.

In this paper, we address the tractable clearing of kidney exchanges with short cycles and chains that are long but bounded.  This corresponds to the practice at most modern fielded kidney exchanges.  We introduce three new integer programming formulations, two of which are compact.  Furthermore, one of these models has a linear programming relaxation that is exactly as tight as the previous \emph{tightest} formulation (which was not compact) for instances in which each donor has a paired patient. On real data from the UNOS nationwide exchange in the United States and the NLDKSS nationwide exchange in the United Kingdom, as well as on generated realistic large-scale data, we show that our new models are competitive with all existing solvers---in many cases outperforming all other solvers by orders of magnitude.

Finally, we note that our \emph{position-indexed chain-edge formulation} can be modified in a straightforward way to take post-match edge failure into account, under the restriction that edges have equal probabilities of failure.  Post-match edge failure is a primary source of inefficiency in presently-fielded kidney exchanges.  We show how to implement such failure-aware matching in our model, and also extend the state-of-the-art general branch-and-price-based non-compact formulation for the failure-aware problem to run its pricing problem in polynomial time.

\end{abstract}



\keywords{Kidney exchange; matching markets; stochastic matching; integer programming; branch and price}

\acmformat{John P. Dickerson, David F. Manlove, Benjamin Plaut, Tuomas Sandholm, and James Trimble, 2016. Position-Indexed Formulations for Kidney Exchange.}

\begin{bottomstuff}
Work supported by NSF grants IIS-1320620 and IIS-1546752, ARO grant W911NF-16-1-0061, Facebook and Siebel fellowships, and by XSEDE through the Pittsburgh Supercomputing Center (Dickerson / Plaut / Sandholm), and by EPSRC grants EP/K010042/1, EP/K503903/1 and EP/N508792/1 (Manlove / Trimble).

Authors' addresses: J.P. Dickerson, B. Plaut, and T. Sandholm, Computer Science Department, Carnegie Mellon University; email: \url{{dickerson,sandholm}@cs.cmu.edu}, \url{bplaut@stanford.edu}; D.F. Manlove and J. Trimble, School of Computing Science, University of Glasgow; email: \url{david.manlove@glasgow.ac.uk}, \url{james.trimble@yahoo.co.uk}.
\end{bottomstuff}

\maketitle

\section{Introduction}\label{sec:intro}

Transplantation is the most effective treatment for kidney failure, yet
transplant waiting lists have grown rapidly in many countries. In the United
States alone, the kidney transplant waiting list grew from \num{58000} people in
\num{2004} to \num{99000} people in \num{2014}~\cite{Hart16:Kidney}. In order to increase the
supply of kidneys for transplant, \emph{kidney exchange schemes} now operate in
several countries, including the United States, United Kingdom, Netherlands, and South Korea.

A kidney exchange~\cite{Rapaport86:Case,Roth04:Kidney} is a
centrally-administered barter exchange market for kidneys.  If a patient with end-stage
renal disease has a friend or family member who is willing to donate a kidney
to him, but unable to do so due to blood- or tissue-type incompatibility, the patient
may enter the exchange in the hope of exchanging his donor with the donor of
another participating patient. Three-way exchanges are also possible, but most
schemes do not carry out four-way or longer exchanges due to the requirement
that transplants be carried out simultaneously, and the risk that one of the
participants may need to withdraw, in which case the cycle does not go to execution and the pairs go back into the kidney exchange pool.

The scheme administrator carries out periodic ``match runs'' in which
exchanges are selected in order to maximize the number of planned transplants,
or a similar goal---perhaps prioritizing pediatric patients, or
to those who have been waiting the longest.

A more recent development has been the introduction of \emph{non-directed
donors (NDDs)} to kidney exchange
schemes~\cite{Montgomery06:Domino,Roth06:Utilizing}. An NDD enters the scheme
with the intention of donating a kidney, but \emph{without} a paired recipient.
Such a donor initiates a \emph{chain}, in which the paired donor of each
patient who receives a kidney donates a kidney to the next patient. In contrast
to cyclic exchanges, chains can be carried out non-simultaneously, since patients
later in the chain are not waiting to be ``repaid'' for a donation that has
already been made. In practice, it is desirable to impose a cap on chain length, since there is an increasing chance that the final exchanges planned
for a chain will not proceed as the chain length is increased (due to various reasons such as pre-transplant crossmatch incompatibility, death of a recipient before transplant, the recipient receiving a deceased-donor kidney, and so on).

In many fielded kidney exchanges, an optimal solution is found by using an integer programming (IP) solver to find
a set of disjoint cyclic exchanges and chains that maximizes some scoring function. This approach has been tractable so
far;~\citeN{Manlove15:Paired} report that each instance up to October \num{2014} in
the United Kingdom's National Living Donor Kidney Sharing Schemes (NLDKSS)---one of the largest kidney
exchange schemes---could be solved in under a second, with similar results using state-of-the-art solvers at other large exchanges in the US~\cite{Anderson15:Finding,Plaut16:Fast} and the Netherlands~\cite{Glorie14:Kidney}. However, there is an
urgent need for faster kidney exchange algorithms, for three reasons:
\begin{itemize}[topsep=0pt] 
\item Schemes have recently increased chain-length caps, and we expect further increases as more schemes evolve towards using \emph{nonsimultaneous extended altruistic donor (NEAD) chains}~\cite{Rees09:Nonsimultaneous}, which can extend across dozens of transplants.
\item Opportunities exist for cross-border schemes, which will greatly increase the size of the problem to be solved; indeed, collaborations have already occurred between, for example, the USA and Greece, and between Portugal and Spain.
\item The run time for kidney exchange algorithms depends on the problem instance, and is difficult to predict. It is desirable to have improved algorithms to insure against the possibility that future instances will be intractable for current solvers.
  \end{itemize}

With that motivation, this paper presents new scalable integer-programming-based approaches to optimally clearing large kidney exchange schemes, including two models which can comfortably handle chain caps greater than \num{10}.

\subsection{Prior research}
Substantial prior research helped grow kidney exchange from a thought experiment to its present increasingly-ubiquitous state.  We briefly overview the literature from economics, computer science, and operations research.

\paragraph{Theoretical basis for kidney exchange} 
Roth et al.~\citeyear{Roth04:Kidney,Roth05:Pairwise,Roth07:Efficient} set the groundwork for large-scale organized kidney exchange.  These papers explored what efficient matchings in a steady-state kidney exchange would look like; extensions by~\citeN{Ashlagi12:Need},~\citeN{Ashlagi14:Free}, and~\citeN{Ding15:Non-asymptotic} address shortcomings in those theoretical models that appeared as kidney exchange became reality.  Game-theoretic models of kidney exchange, where transplant centers are viewed as agents with a private type consisting of their internal pools, were presented and explored by~\citeN{Ashlagi14:Free},~\citeN{Toulis15:Design}, and~\citeN{Ashlagi15:Mix}.  Various forms of dynamism, like uncertainty over the possible future vertices in the pool~\cite{Unver10:Dynamic,Ashlagi13:Kidney,Akbarpour14:Dynamic} or uncertainty over the existence of particular edges~\cite{Blum13:Harnessing,Dickerson13:Failure,Blum15:Ignorance} have been explored from both an economic and algorithmic efficiency point of view.

\paragraph{Practical approaches to the kidney exchange clearing problem}
The two fundamental IP models for kidney exchange are the \emph{cycle formulation}, which includes one binary decision variable for each feasible cycle or chain, and the \emph{edge
formulation}, which includes one decision variable for each compatible pair of agents~\cite{Abraham07:Clearing,Roth07:Efficient}. In the cycle formulation, the number of constraints is sublinear in the input size, but the number of variables is exponential. In the edge formulation, the number of variables is linear but the number of constraints is exponential.  Optimally solving these models has been an ongoing challenge for the past decade.

\citeN{Constantino13:New} introduced the first two compact IP
formulations for kidney exchange, where \emph{compact} means that the counts of variables and
constraints are polynomial in the size of the input. Their \textit{extended
edge formulation} was shown empirically to be effective in finding the optimal
solution where the cycle-length cap is greater than $3$, particularly on dense graphs. However, each of the
compact formulations introduced by Constantino et al. has a weaker
linear program (LP) relaxation than the cycle formulation, even in the absence of NDDs.

The EE-MTZ model~\cite{MakHau15:Kidney}, another compact formulation, uses the variables and constraints of the extended edge formulation to model
cycles and a variant of the Miller-Tucker-Zemlin model for the traveling
salesman problem to model chains.  The same paper introduces the
exponentially-sized SPLIT-MTZ model, which adds redundant constraints to the
edge formulation in order to tighten the LP relaxation.

A number of kidney exchange algorithms use the cycle formulation with 
\emph{branch and price}~\cite{Barnhart98:Branchprice} to avoid the
need to hold an exponential number of variables in memory
~\cite{Abraham07:Clearing,Dickerson13:Failure,Glorie14:Kidney,Klimentova14:New,Plaut16:Fast}.
These have been the fastest algorithms to date for the kidney exchange problem.

An alternative approach to avoiding the cost of keeping the entire model in
memory has been constraint generation, using variants of the edge
formulation~\cite{Abraham07:Clearing}.
\citeN{Anderson15:Finding} describe an approach based on an algorithm for the
prize-collecting traveling salesman problem. This algorithm is particularly effective for solving instances where the cycle-length cap 
is $3$ and there is no cap on the length of chains, but it is outperformed by branch-and-price-based approaches if a finite chain-length 
cap is used~\cite{Plaut16:Fast}.

Alternative objectives for the kidney exchange problem include maximizing the
\emph{expected} number of transplants subject to post-match arc and vertex
failures~\cite{Dickerson13:Failure,Pedroso14:Maximizing,Alvelos15:Compact}.  Some fielded exchanges use lexicographic optimisation of a hierarchy of objectives~\cite{Glorie14:Kidney,Manlove15:Paired}; we note that our models can be augmented to support this class of objective function.


\subsection{Our contribution}
This paper introduces three integer program formulations for the
kidney exchange problem, two of which are compact.  Model size (i.e., memory footprint) often constrains today's kidney exchange solvers; critically, our models are typically much smaller than the state of the art \emph{while managing to maintain tight linear program relaxations (LPRs)}---which in practice is quite important to proving optimality quickly.

In Section~\ref{sec:pief}, we introduce the \emph{position-indexed edge formulation (PIEF)}, a model for the kidney exchange problem with only cycles that is substantially smaller than, yet has an LPR equivalent to, the model with the tightest LPR for the cycles-only version of the problem~\cite{Abraham07:Clearing,Roth07:Efficient}.  Section~\ref{sec:pief} presents the \emph{position-indexed chain-edge formulation (PICEF)} which compactly brings chains into the model via a polynomial number of decision variables; the number of cycle decision variables is exponential in just the maximum \emph{cycle} length (which is typically only $3$ or $4$ in fielded exchanges).  To address that latter exponential reliance on the cycle length, we also present a branch-and-price-based implementation of 
\iftoggle{isFullVersion}{%
  PICEF.  Finally, in Appendix~\ref{sec:hybrid} we present the \emph{hybrid position-indexed edge formulation (HPIEF)}, which combines PIEF and PICEF to yield a compact formulation.
}{%
  PICEF.\footnote{In Appendix C of the full version of the paper~\cite{Dickerson16:PositionArxiv}, we present the \emph{hybrid position-indexed edge formulation (HPIEF)}, which combines PIEF and PICEF to yield a compact formulation, as well as additional theoretical results and proofs.}
}

Throughout, we prove new results regarding the tightness of the LPRs of our models relative to the current state of the art.  The tightness of these relaxations hints that our formulations will be competitive in practice; toward that end, we provide extensive experimental evidence that they are.
In particular, we show that at least one of PICEF and HPIEF is faster than the best solver from all those provably-optimal solvers contributed in earlier papers that we evaluated in $96.41$\% of instances considered, with the speed-ups being most evident for larger instance sizes and larger chain caps.
In Section~\ref{sec:failure}, we use real and generated data from two nationwide kidney exchange programs---one in the UK, and one in the US---to compare our formulations against other competitive solvers~\cite{Abraham07:Clearing,Klimentova14:New,Anderson15:Finding,Plaut16:Fast}.  Our new formulations are on par or faster than all other solvers, outperforming all other solvers by orders of magnitude on many problem instances.


Finally, we demonstrate that our PICEF formulation can be adapted straightforwardly to the so-called \emph{failure-aware} model where arcs probabilistically fail after algorithmic match but before transplantation, and we show how to adapt a recent (non-compact) branch-and-price-based formulation due to~\citeN{Dickerson13:Failure} to perform pricing in polynomial time (Section~\ref{sec:failure}).

\section{Preliminaries}\label{sec:prelims}
Given a pool consisting of patient-donor pairs and non-directed donors (NDDs), we model the kidney exchange problem using a directed graph $D=(V,A)$.  The set of vertices $V = \{1, \dots, |V|\}$ is partitioned into $N = \{1, \dots, |N|\}$ representing the NDDs
and $P = \{|N| + 1, \dots, |N| + |P|\}$ representing the patient-donor pairs.

For each $i \in N$ and each $j \in P$, $A$ contains the arc $(i,j)$ if and only
if NDD $i$ is compatible with patient $j$. For each $i \in P$ and each $j \in P
\setminus \{i\}$, $A$ contains the arc $(i,j)$ if and only if the paired donor\footnote{In this paper, we assume that each patient has a single paired donor.  In practice, a patient may have multiple donors; this can be modeled by regarding such a patient $i$'s vertex in $D$ as representing $i$ and all of her paired donors; an arc $(i,j)$ in $D$ then represents compatibility between at least one of $i$'s donors and patient $j$.  Our model can also handle patients with \emph{no} paired donor by drawing a vertex with no out edges.} of
  patient $i$ is compatible with patient $j$. Each arc $(i,j) \in A$ has a
  weight $w_{ij} \in \mathbb{R}^+$ representing the priority that the scheme
  administrator gives to that transplant.  If the objective is to maximize the
  number of transplants, each arc has unit weight; most fielded exchanges use weights to encode various prioritization schemes and other value judgments about potential transplants.

Since NDDs do not have paired patients, each vertex representing an NDD has no
incoming arcs. Moreover, we assume that no patient is compatible with her own
paired donor, and therefore that the digraph has no loops.  (The IP models
introduced in this paper can trivially be adapted to the case where loops
exist by adding a binary variable for each loop and modifying the
objective and capacity constraints accordingly.)

We use the term \emph{chain} to refer to a path in the digraph initiated by an
NDD vertex, and \emph{cycle} to refer to a cycle in the directed graph (which
must involve only vertices in $P$, since vertices in $N$ have no incoming arcs).
The weight $w_c$ of each cycle or chain $c$ is defined as the sum of its arc
weights.

Given a maximum cycle size $K$ and a maximum chain length $L$, the \emph{kidney exchange problem} is an optimization problem where the objective is to select a vertex-disjoint packing in $D$ of cycles of size up to $K$ and chains of length up to $L$ that has maximum weight.  The problem is NP-hard for realistic parameterizations of $K$ and $L$~\cite{Abraham07:Clearing,Biro09:Maximum}.  In practice, $K$ is kept low due to the logistical constraint of scheduling all transplants simultaneously.  At both the United Network for Organ Sharing (UNOS) US-wide exchange and the UK's NLDKSS, $K=3$.  The chain cap $L$ is typically longer due to chains being executed non-simultaneously; yet, typically $L \neq \infty$ due to potential matches failing before transplantation.  This paper addresses the realistic setting of small cycle cap $K$ and large---but finite---chain cap $L$.


Figure~\ref{fig:problem_example} shows a problem instance with $|N|=2$ and
$|P|=5$. If each arc has unit weight and $K=L=3$, then the bold arcs show an
optimal solution: cycle $( (3,4), (4,5), (5,3))$ and chain
$((1,7), (7,6))$, with a total weight of $5$.

\begin{figure}
\centering
\begin{tikzpicture}[scale=0.8,transform shape]
  \newcount \c

  \foreach \n in {1, 2} {
    \c=\n
    \multiply\c by -180
    \advance\c by 90
    \node[whitesq] (N\n) at (\the\c:.9) {\n};
  }
  \begin{scope}[xshift=3cm]
    \foreach \n in {3, ..., 7} {
      \c=\n
      \multiply\c by -72
      \advance\c by 36 
      \node[patient] (N\n) at (\the\c:1.1) {\n};
    }
  \end{scope}

  \foreach \m/\n in {3/4, 4/5, 5/3}
     \draw[-latex,line width=2] (N\m) -- (N\n);

  \foreach \m/\n in {3/7, 1/3, 6/5, 2/4}
     \draw[-latex] (N\m) -- (N\n);

  \foreach \m/\n in {1/7, 7/6}
     \draw[-latex, line width=2] (N\m) -- (N\n);
\end{tikzpicture}
\caption{The directed graph for a kidney-exchange instance with $|N|=2$ and $|P|=5$}
\label{fig:problem_example}
\end{figure}
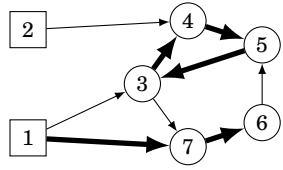

\section{PIEF: Position-Indexed Edge Formulation}\label{sec:pief}  
\subsection{Description of the model}
In this section, we present our first of three new IP formulations, the \textit{position-indexed
edge formulation} (PIEF).  PIEF is a natural extension of the extended edge
formulation (EEF) of~\citeN{Constantino13:New}. For this formulation, we assume that the problem
instance contains no 
\iftoggle{isFullVersion}{%
  NDDs; HPIEF (the hybrid PIEF) in Appendix~\ref{sec:hybrid} is a compact generalisation of this formulation which can be used for instances with NDDs.
}{%
  NDDS.\footnote{In Appendix C of the full version of this paper, we present the hybrid PIEF (HPIEF), a compact generalisation of PIEF which can be used for instances with NDDs.}
}

The PIEF, like the EEF, uses copies of the underlying compatibility digraph $D$. For each vertex $l \in V$, let
$D^l=(V^l,A^l)$ be the subgraph of $D$ induced by $\{i \in V: i \geq l\}$.  The
PIEF ensures that at most one cycle is selected in each copy, and that a cycle
selected in graph copy $D^l$ must contain vertex $l$.

The first directed graph in Figure~\ref{fig:pief_example} is an instance with
four patients which we will use as an example in this section.
The figure shows graph copies $D^1$ ($=D$), $D^2$, and $D^3$.
The remaining graph copy, $D^4=(\{4\}, \{\})$, contains no arcs and
is not shown.

\begin{figure}
\centering
\begin{tikzpicture}[scale=0.8,transform shape]
  \newcount \c
  \newcount \lpos
  \newcount \lminusone
 
  \foreach \l in {1, ..., 3} {
    \lpos = \l
    \lminusone = \l
    \advance \lminusone by -1
    \multiply\lpos by 4
    \begin{scope}[xshift=\the\lpos cm]
      \ifthenelse {\lengthtest{2 pt > \l pt}} {
        \node[] at (-2,0) {$D=D^\l=$};
      }{
        \node[] at (-1.6,0) {$D^\l=$};
      }

      \foreach \n in {1, ..., 4} {
        \ifthenelse {\lengthtest{\n pt > \lminusone pt}} {
          \c=\n
          \multiply\c by -90
          \advance\c by -135
          \node[whitecirc] (N\l\n) at (\the\c:1.1) {\n};
        }{}
      }
    \end{scope}

    \foreach \m/\n in {1/2, 2/1, 2/3, 3/4, 4/1, 4/3} {
      \ifthenelse {\lengthtest{\m pt > \lminusone pt} \AND \lengthtest{\n pt > \lminusone pt}} {
        \draw[-latex, bend left=16] (N\l\m) to node {} (N\l\n);
      }{}
    }

    \ifthenelse {\lengthtest{2 pt > \l pt}} {
    }{}
    \ifthenelse {\lengthtest{3 pt > \l pt}} {
      \draw[-latex] (N\l4) to node {} (N\l2);
    }{}
    \ifthenelse {\lengthtest{4 pt > \l pt}} {
    }{}

  }

\end{tikzpicture}
\caption{A kidney exchange instance $D$ with $|N|=0$ and $|P|=4$, along with graph copies $D^1$ ($=D$), $D^2$, and $D^3$. $D^4=(\{4\}, \{\})$ is not shown.}
\label{fig:pief_example}
\end{figure}
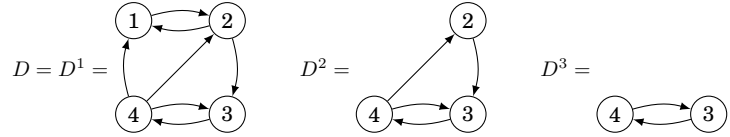

The main innovation of the PIEF formulation is the use of arc \emph{positions}
to index variables; the position of an arc in a cycle is defined as follows.
Let $c = (a_1, \dots, a_{|c|})$ be a cycle represented as a list of arcs in
$A$. Further, assume that we use the unique representation of $c$ such that
$a_1$ leaves the lowest-numbered vertex involved in the cycle.  For $1 \leq i
\leq |c|$, we say that $a_i$ has position $i$.

We define $\K(i,j,l)$, the set of positions at which arc $(i,j)$ is permitted
to be selected in a cycle in graph copy $D^l$. For $i, j, l \in V$ such that $(i, j)
\in A^l$, let

\vspace{-3mm}
\[
\K(i,j,l) =
\begin{cases}
    \{1\} & i=l \\
    \{2, \dots, K-1\} & i,j>l \\
    \{2, \dots, K\} & j=l.
\end{cases}
\]

Thus, an arc may be selected at position \num{1} in graph copy $l$ if and only if it leaves vertex $l$, and
any arc selected at position $K$ in graph copy $l$ must enter $l$.

Now, create a set of binary decision variables as follows. For $i, j, l \in P$ such that $(i,
j) \in A^l$, create variable $x_{ijk}^l$ for each $k\in \K(i,j,l)$. Variable
$x_{ijk}^l$ takes the value \num{1} if and only if arc $(i,j)$ is selected at position
$k$ of a cycle in graph copy $D^l$. Returning to our example instance and letting $K=3$,
we give $x_{342}^1$ as an example of a variable in the model; this represents the arc $(3,4)$ being
used in position 2 of a cycle in graph copy $1$. In full, the set of variables created for
this instance is
$x_{121}^1$,
$x_{212}^1$,
$x_{213}^1$,
$x_{232}^1$,
$x_{342}^1$,
$x_{412}^1$,
$x_{413}^1$,
$x_{422}^1$,
$x_{432}^1$ (in graph copy $D^1$),
$x_{231}^2$,
$x_{342}^2$,
$x_{422}^2$,
$x_{423}^2$,
$x_{432}^2$
(in graph copy $D^2$),
$x_{341}^3$,
$x_{432}^3$, and
$x_{433}^3$ (in graph copy $D^3$).


The following integer program finds the optimal cycle packing.

\begin{subequations}
\begin{align}
\max \qquad \sum_{l\in V} \sum_{(i,j)\in A^l} \sum_{k\in \K(i,j,l)} w_{ij}x_{ijk}^l & \label{eq:pief_obj}\\
\text{s.t.} \qquad \sum_{l\in V} \sum_{j:(j,i)\in A^l} \sum_{k\in \K(j,i,l)} x_{jik}^l \leq 1 & \qquad i \in V \label{eq:pief_a} \\
\sum_{\substack{j: (j,i) \in A^l \wedge \\ k \in \K(j,i,l)}} x_{jik}^l =
    \sum_{\substack{j: (i,j) \in A^l \wedge \\ k+1 \in \K(i,j,l)}}x_{i,j,k+1}^l &
    \qquad 
    \begin{aligned}
               & l \in V,\\
               & i \in \{l+1, \dots, n\},\\ 
               & k \in \{1, \dots, K-1\}
    \end{aligned}
    \label{eq:pief_b} \\
x_{ijk}^l \in \{0,1\} & \qquad
               l \in V,
               (i,j) \in A^l,
               k \in \K(i,j,l)
    \label{eq:pief_c}
\end{align}
\end{subequations}

The objective (\ref{eq:pief_obj}) is to maximize the weighted sum of selected
arcs. Constraint (\ref{eq:pief_a}) is the capacity constraint for vertices: for
each vertex $i \in V$, there must be at most one selected arc whose target is
$i$. Constraint (\ref{eq:pief_b}) is the flow conservation constraint. For each
graph copy $D^l$, each vertex $i$ in $D^l$ except the lowest-numbered vertex, and each arc position $k<K$, the number of selected arcs at position $k$ with target $i$ is equal to the number of selected arcs at position $k+1$ with source $i$. Constraint (\ref{eq:pief_c}) ensures that no fractional solutions are 
\iftoggle{isFullVersion}{%
  selected.  Theorem \ref{thm:PIEFcorrect} in Section \ref{sec:PIEFcorrect} establishes the correctness of the PIEF model.
}{%
  selected.\footnote{Theorem~A.3 in Appendix~A.1 of the full paper establishes the correctness of the PIEF model.}
}

We note that PIEF is not the first IP model for a directed-graph program to use
position-indexed variables;~\citeN{Vajda61:Mathematical} uses
position-indexed variables for subtour elimination in a model for the
travelling salesman problem (TSP). Vajda's model is substantially different
from the PIEF; most notably, graph copies are not required for Vajda's model
because any TSP solution contains exactly one cycle.

\subsection{Further reducing the size of the basic PIEF model}\label{sec:pief-reduce}
We now present methods for reducing the size of the PIEF model while maintaining provable optimality.  These reductions are performed as a polynomial-time preprocess (prior to solving the NP-hard kidney exchange clearing problem), and thus may result in practical run time improvements.

\subsubsection{Basic reduced PIEF}\label{sec:pief-basic-reduction}

In typical problem instances, there are many $(i,j,k,l)$ tuples such that $x_{ijk}^l$ takes the value zero in
any assignment that satisfies constraints (\ref{eq:pief_a})-(\ref{eq:pief_c}).
For example, suppose that $K=4$ and that Figure~\ref{fig:reduced_pief} is graph
copy $D^1$. Arc $(6,7)$ cannot be chosen at position \num{3} of a cycle, since the
arc only appears in one cycle and it is at position \num{2} of that cycle.  Hence, if
we could eliminate variable $x_{673}^1$ from the integer program it would not change the
optimal solution.  Similarly, all variables for the arc $(3,4)$ within this
graph copy could be eliminated, since this arc does not appear in any cycle of
length less than \num{5}.

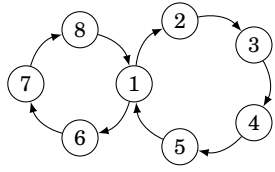
\begin{figure}
\centering
\begin{tikzpicture}[scale=0.8,transform shape]
  \newcount \c
 
  \foreach \n in {6, ..., 8} {
    \c=\n
    \multiply\c by -90
    \advance\c by 90
    \node[whitecirc] (N\n) at (\the\c:.9) {\n};
  }
  \begin{scope}[xshift=2cm]
    \foreach \n in {1, ..., 5} {
      \c=\n
      \multiply\c by -72
      \advance\c by -108 
      \node[whitecirc] (N\n) at (\the\c:1.1) {\n};
    }
  \end{scope}

  \foreach \m/\n in {1/2, 2/3, 3/4, 4/5, 5/1, 1/6, 6/7, 7/8, 8/1}
  \draw[-latex, bend left] (N\m) to node {} (N\n);
\end{tikzpicture}
\caption{A graph copy where the reduced PIEF decreases the number of variables
in the integer program} \label{fig:reduced_pief}
\end{figure}

Following the approach used by~\citeN{Constantino13:New} for the extended edge
formulation, we eliminate variables as follows. For
$i,j \in V^l$, let
$d_{ij}^l$ be the length of the shortest path in terms of arcs from $i$ to $j$
in $D^l$. For $(i,j) \in A^l$, let
\[
\K^\text{red}(i,j,l) = \{k : 1 \leq k \leq K \wedge d_{li}^l<k \wedge d_{jl}^l \leq (K-k) \}.
\]

For any $k \notin \K^\text{red}(i,j,l)$, no cycle in graph copy $D^l$ of length
less than or equal to $K$ contains $(i,j)$ at position $k$, since either there
is no sufficiently short path from $l$ to $i$ or there is no sufficiently short
path from $j$ to $l$.  Note that $\K^\text{red}(i,j,l) \subseteq \K(i,j,l)$. We
can substitute $\K^{\text{red}}(i,j,l)$ for $\K(i,j,l)$ in
(\ref{eq:pief_obj})-(\ref{eq:pief_c}), yielding a smaller integer
program---\emph{PIEF-reduced (PIEFR)}---with the same optimal solution.


%


\subsubsection{Elimination of variables at position $1$ and $K$: the \PIEFRR{} formulation}\label{sec:pief-redundancy}

In the PIEFR model with $K \geq 3$, variables at position \num{1} are redundant, since $x_{lj1}^l=1$
if and only if $x_{ji2}^l=1$ for some $i$.  Similarly, variables at position
$K$ are redundant; $x_{jlK}^l=1$ if and only if $x_{ij(K-1)}^l=1$ for some $i$.
We can eliminate variables at positions 1 and $K$ from PIEFR as follows. Define
a modified weight function $w'$: for all $i, j, l \in P$ such that $(i,j) \in
A^l$ and all $k \in \{2, \dots, K-1\}$, let

\[
w'(i,j,k,l) =
\begin{cases}
    w_{ij} + w_{li} & k=2 \\
    w_{ij} + w_{jl} & k=K-1 \\
    w_{ij} & \text{otherwise}. 
\end{cases}
\]

With this weight function, a selected arc $(i,j)$ at position \num{2} of a cycle in
$D^l$ contributes to the objective value its own weight plus the weight of the
implicitly selected arc $(l,i)$. An arc $(i,j)$ at position $K-1$ of a cycle in
$D^l$ contributes its own weight plus the weight of of the implicitly selected
arc $(j,l)$.

For $i, j, l \in P$ such that $(i, j) \in A^l$, define the restricted set of
permitted arc positions:

\[
\KRR(i,j,l) = \K^\text{red}(i,j,l) \setminus \{1, K\}.
\]

For $i, j, l \in P$ such that $(i, j) \in A^l$, and for each $k\in
\KRR(i,j,l)$, create a binary variable $x_{ijk}^l$ . The following IP, denoted \PIEFRR{} (PIEF reduced twice), is solved.

\begin{subequations}
\begin{align}
\max \qquad \sum_{l\in V} \sum_{(i,j)\in A^l} \sum_{k\in \KRR(i,j,l)} w'(i,j,k,l) x_{ijk}^l &
            \quad\text{subject to} \label{eq:pief_prime_obj}\\
\sum_{l\in V} \left(
                 \sum_{j:(j,i)\in A^l} \quad \sum_{\mathclap{\quad k\in \KRR(j,i,l)}} x_{jik}^l + \quad
                 \sum_{\mathclap{\substack{j:(i,j)\in A^l \wedge \\ 2 \in \KRR(i,j,l)}}} x_{ij2}^l
              \right)
              + \quad
       \sum_{\mathclap{\substack{h,j:j\not=i \wedge \\ (h,j)\in A^i \wedge \\ K-1 \in \KRR(h,j,i)}}} x_{hj(K-1)}^i
    \leq 1 & \quad i \in V \label{eq:pief_prime_a} \\
\sum_{\substack{j: (j,i) \in A^l \wedge \\ k \in \KRR(j,i,l)}} x_{jik}^l =
    \sum_{\substack{j: (i,j) \in A^l \wedge \\ k+1 \in \KRR(i,j,l)}}x_{i,j,k+1}^l &
    \quad 
    \begin{aligned}
               & l \in V,\\
               & i \in \{l+1, \dots, n\},\\ 
               & k \in \{2, \dots, K-2\}
    \end{aligned}
    \label{eq:pief_prime_b} \\
x_{ijk}^l \in \{0,1\} & \quad
    \begin{aligned}
          & l \in V,
          (i,j) \in A^l, \\
          & k \in \KRR(i,j,l)
    \end{aligned}
    \label{eq:pief_prime_c}
\end{align}
\end{subequations}

The constraints of \PIEFRR\ differ from those of PIEFR
(\ref{eq:pief_a}-\ref{eq:pief_c}) in the following two respects.  First, the
second term in parentheses in the \PIEFRR\ capacity constraint for vertex $i$
(\ref{eq:pief_prime_a}) ensures that any selected arc $(i,j)$ at position \num{2} of
a cycle in $D^l$ counts towards the capacity for $i$, since it is implicit that
the arc $(l,i)$ is also chosen. The final term on the left-hand side of
constraint (\ref{eq:pief_prime_a}) serves the same function for selected arcs
at position $K-1$, since an arc at position $K$ is implicitly chosen.  Second,
the flow conservation constraint (\ref{eq:pief_prime_b}) is not required for $k
\in \{1,K-1\}$, since arcs at positions \num{1} and $K$ are not modelled explicitly in
\PIEFRR.

\subsubsection{Vertex-ordering heuristic}\label{sec:pief-heuristics}

We can reduce the number of variables in the reduced PIEF model by carefully choosing the order of vertex labels in the digraph $D$. We have found that relabelling the 
vertices in descending order of total degree is an effective heuristic to this end. To estimate the effect of this
ordering heuristic on model size, we generated the \PIEFRR\ model for each of the
ten PrefLib instances~\cite{Mattei13:Preflib} with \num{256} vertices
and no NDDs. The heuristic reduced the variable count by a
mean of \num{38} percent, and reduced the constraint count by a mean of \num{60} percent.

\subsection{PIEF has a tight LPR}\label{sec:pief-lpr}
We now compare the LPR bound of PIEF to those of other popular IP models.  The tightness of an LPR is typically viewed as a proxy for how well an IP model will perform in practice, due to the important role the relaxation plays in modern branch-and-bound-based tree search.  In this section, we compare the LPR of PIEF against the IP formulation with the tightest LPR, the \emph{cycle formulation}. While the number of decision variables in the cycle formulation model is exponential in the cycle cap $K$, PIEF maintains an LPR that is just as tight, but has far fewer variables if $K>3$.

\subsubsection{Cycle formulation}\label{sec:pief-cycle-formulation}
We begin by reviewing the cycle formulation due to~\citeN{Abraham07:Clearing} and \citeN{Roth07:Efficient}, and note that the formulation is equivalent to that due to~\citeN{Anderson15:Finding} if chains are disallowed.
In the cycle formulation, a binary variable $z_c$ is used for each feasible
cycle or chain $c$ to represent whether it is selected in the solution. For
each vertex $v$, there is a constraint to ensure that $v$ is in at most
selected cycle or chain. Let $\cycles$ be the set of cycles in $D$ of length at
most $K$ and let $\chains$ be the set of NDD-initiated chains in $D$ of length
at most $L$. The model to be solved is as follows.

\vspace{-6mm}
\begin{subequations}
\begin{align}
\max\ \sum_{c \in \cycles \cup \chains} w_cz_c & \label{eq:cf_objective}\\
\text{s.t. } \sum_{c \in \cycles \cup \chains : i \text{ appears in } c} z_c \leq 1 & \qquad i: 1 \leq i \leq |V| \label{eq:cf_a} \\
z_c \in \{0, 1\} & \qquad c \in \cycles \cup \chains \label{eq:cf_b}
\end{align}
\end{subequations}

Constraints (\ref{eq:cf_a}) ensure that the selected cycles and chains are
vertex disjoint.

\subsubsection{LPR of PIEF}

We now show that the LPR of PIEF is exactly as tight as that of the cycle formulation.
Formally, if $A$ and $B$ are two IP formulations for the kidney exchange problem, we write that $A$ \emph{weakly dominates} $B$, denoted $Z_A \preceq Z_B$, if for every problem instance, the LPR objective value under $A$ is no greater than the LPR objective value under $B$.  Further, we say that $A$ \emph{strictly dominates} $B$, denoted $Z_A \prec Z_B$, if $Z_A \preceq Z_B$ and for some problem instance, the LPR objective value under $A$ is strictly smaller than the LPR objective value under $B$.  Finally, we write $Z_A = Z_B$ if $Z_A \preceq Z_B$ and $Z_B\preceq Z_A$.

\begin{restatable}{theorem}{thmpiefvscf}%
\label{thm:pief-vs-cf}
  $\zCF{} = \zPIEF{}$ (without chains).
\end{restatable}

\section{PICEF: Position-Indexed Chain-Edge Formulation}\label{sec:picef}
\subsection{Description of the model}
Our second new IP formulation, PICEF, uses a variant of PIEF for chains, and---like the cycle formulation---uses one binary variable for each cycle.
The idea of using variables for arcs in chains and a variable for
each cycle was introduced in the PC-TSP-based
algorithm of~\citeN{Anderson15:Finding}. The innovation in our IP model is the use of
position indices on arc variables, which results in polynomial counts of
constraints and arc-variables; this is in contrast to the exponential number of
constraints in the PC-TSP-based model.

Unlike PIEF, PICEF does not require copies of $D$. Intuitively, this is
because a chain is a simpler structure than a cycle, with no requirement for
a final arc back to the initial vertex.

We define $\K'(i,j)$, the set of possible positions at which arc $(i,j)$ may occur in a chain in the digraph $D$. For $i, j \in V$ such that $(i, j) \in A$,
\[
\K'(i,j) =
\begin{cases}
    \{1\} & i \in N \\
    \{2, \dots, L\} & i \in P
\end{cases}.
\]

Thus, any arc leaving an NDD can only be in position \num{1} of a chain, and any
arc leaving a patient vertex may be in any position up to the cycle-length cap
$L$, except \num{1}.

For each $(i,j) \in A$ and each $k \in \K'(i,j)$, create variable $y_{ijk}$,
which takes value 1 if and only if arc $(i,j)$ is selected at position $k$ of
some chain. For each cycle $c$ in $D$ of length up to $K$, define a binary variable $z_c$ to
indicate whether $c$ is used in a packing. 

For example, consider the instance in Figure~\ref{fig:picef_example}, in which
$|N|=2$ and $|P|=4$. Suppose that $K=3$ and $L=4$, and suppose further that
each arc has unit weight.  The IP model includes variables $y_{131}$,
$y_{141}$, and $y_{241}$, corresponding to arcs leaving NDDs. For each $k \in
{2, 3, 4}$, the model includes variables $y_{34k}$, $y_{45k}$, $y_{56k}$,
$y_{64k}$, and $y_{65k}$, corresponding to arcs between donor-patient pairs at
position $k$ of a chain.  Finally, the model includes $z_c$ variables for the
cycles $((4,5), (5,6), (6,4))$ and $((5,6), (6,5))$.

\vspace{-3mm}
  \begin{figure}[ht!bp]
    \centering
    \begin{tikzpicture}[scale=0.8,transform shape]

  \node (n1) at (0, 1.6) [whitesq] {1};
  \node (n2) at (0, 0) [whitesq] {2};

  \node (n3) at (1.4, .8) [whitecirc] {3};
  \node (n4) at (2.8, .8) [whitecirc] {4};
  \node (n5) at (4.2, 1.6) [whitecirc] {5};
  \node (n6) at (4.2, 0) [whitecirc] {6};

  \draw [-latex] (n1) to node {} (n3);
  \draw [-latex, bend left=10] (n1) to node {} (n4);
  \draw [-latex, bend right=10] (n2.east) to node {} (n4);
  \draw [-latex] (n3) to node {} (n4);
  \draw [-latex, bend left=20] (n4) to node {} (n5);
  \draw [-latex, bend left=20] (n5) to node {} (n6);
  \draw [-latex, bend left=20] (n6) to node {} (n4);
  \draw [-latex, bend left=20] (n6) to node {} (n5);

\end{tikzpicture}
    \caption{An instance with $|N|=2$ and $|P|=4$}
    \label{fig:picef_example}
  \end{figure}
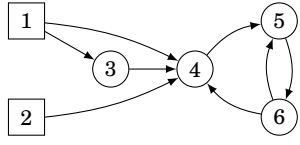
  \vspace{-4mm}
  
The following IP is solved to find a maximum-weight packing of cycles
and chains.

\vspace{-3mm}

\begin{subequations}
\begin{align}
\max \qquad \sum_{(i,j)\in A} \sum_{k\in \K'(i,j)} w_{ij}y_{ijk} 
        + \sum_{c \in \cycles} w_c z_c & \label{eq:picef_obj}\\
\text{s.t.} \qquad \sum_{j:(j,i)\in A} \sum_{k\in \K'(j,i)} y_{jik} +
        \sum_{\mathclap{c \in \cycles : i \text{ appears in } c}} z_c \leq 1
        & \qquad i \in P \label{eq:picef_a} \\
\sum_{j:(i,j)\in A} y_{ij1} \leq 1
        & \qquad i \in N \label{eq:picef_b} \\
\sum_{\substack{j: (j,i) \in A \wedge \\ k \in \K'(j,i)}} y_{jik} \geq
    \sum_{j: (i,j) \in A}y_{i,j,k+1} &
    \qquad 
    \begin{aligned}
               & i \in P,\\ 
               & k \in \{1, \dots, K-1\}
    \end{aligned}
    \label{eq:picef_c} \\
y_{ijk} \in \{0,1\} & \qquad
               (i,j) \in A, k \in \K'(i,j)
    \label{eq:picef_d} \\
z_c \in \{0, 1\} & \qquad c \in \cycles \label{eq:picef_e}
\end{align}
\end{subequations}

Inequality~(\ref{eq:picef_a}) is the capacity constraint for patients: each
patient vertex is involved in at most one chosen cycle or incoming arc of a
chain.  Inequality~(\ref{eq:picef_b}) is the capacity constraint for NDDs: each
NDD vertex is involved in at most one chosen outgoing arc.  The flow
inequality~(\ref{eq:picef_c}) ensures that patient-donor pair vertex $i$ has an
outgoing arc at position $k+1$ of a selected chain only if $i$ has an incoming arc at
position $k$; we use an inequality rather than an equality since the final vertex of a chain will have an incoming arc
but no outgoing 
\iftoggle{isFullVersion}{%
  arc.  Theorem \ref{thm:PICEFcorrect} in Section \ref{sec:PICEFvalid} establishes the correctness of the PICEF model.%
}{%
  arc.\footnote{Theorem~B.4 in Appendix~B.1 of the full paper establishes the correctness of the PICEF model.}%
}
  
We now give an example of each of the inequalities
(\ref{eq:picef_a}--\ref{eq:picef_c}) for the instance in
Figure~\ref{fig:picef_example}.  For $i=4$, the capacity
constraint~(\ref{eq:picef_a}) ensures that $y_{141} + y_{241} + y_{342} +
y_{343} + y_{344} + z_{((4,5),(5,6),(6,4))} \leq 1$. For $i=1$, the NDD
capacity constraint~(\ref{eq:picef_b}) ensures that $y_{131} + y_{141} \leq 1$.
For $i=5$ and $k=2$, the chain flow constraint~(\ref{eq:picef_c}) ensures that
$z_{452} + z_{652} \geq z_{563}$; that is, the outgoing arc $(5,6)$ can only be
selected at position $3$ of a chain if an incoming arc to vertex $5$ is selected at
position $2$ of a chain.

In our example in Figure~\ref{fig:picef_example}, the optimal objective value
is $4$.  One satisfying assignment that gives this objective value is $y_{131} =
y_{342} = z_{((5,6), (6,5))} = 1$, with all other variables
equal to zero.

\subsection{Practical implementation of the PICEF model}\label{sec:picef-implementation}

We now discuss methods for the practical implementation of PICEF, first by reducing the number of decision variables via a polynomial-time preprocess, and second by tackling the large number of decision variables for cyclic exchanges via a branch-and-price-based transformation of the model.

\subsubsection{Reduced PICEF}\label{sec:picef-reduced}

We can reduce the PICEF model using a similar approach to the PIEF reduction in
Subsection~\ref{sec:pief-basic-reduction}. For $i \in P$, let $d(i)$ be length of the shortest
path in terms of arcs from some $j \in N$ to $i$. Since any outgoing arc from $i$
cannot appear at position less than $d(i) + 1$ in a chain, we can replace $\K'$
in PICEF with $\K^\text{red}$, defined as follows:
\vspace{-3mm}
\[
\K^\text{red}(i,j) =
\begin{cases}
    \{1\} & i \in N \\
    \{d(i)+1, \dots, L\} & i \in P.
\end{cases}
\]

\subsubsection{A branch and price implementation of PICEF}\label{sec:picef-bnp}
We now discuss a method for scaling PICEF to graphs with high cycle caps, or large graphs with many cycles; this method maintains the full set of arc decision variables, but only incrementally considers those corresponding to cycles. 

Formally, for $V = P \cup N$, the number of cycles of length at most $K$ is $O(|P|^K)$, making explicit representation and enumeration of all cycles infeasible for large enough instances. With one decision variable per cycle,~\citeN{Abraham07:Clearing} could not even write the full integer program in memory for instances as small as \num{1000} pairs.

Branch and price is a method where only a subset of the decision variables are kept in memory, and columns (in the case of PICEF, only those corresponding to cycle variables) are slowly added until correctness is proven at each node in a branch-and-bound search tree. If necessary, superfluous columns can also be removed from the model, in order to prevent its size from exceeding memory.

The following process occurs at each node in the search tree: first, the LP relaxation of the current model (which may contain only a small number of cycles) is solved. The next step is to generate \emph{positive price cycles}: cycles that have the potential to improve the objective value if included in the model.

The price of a cycle $c$ is given by $\sum_{(i,j) \in c} (w_{ij} - \delta_i)$, where $\delta_i$ is the dual value of vertex $i$. While there exist any positive price cycles at a node in the search tree, optimality of the reduced LP has not yet been proved at that node. The \emph{pricing problem} is to bring at least one new positive price cycle into the model, or prove that none exist. Multiple methods exist for solving the pricing problem in kidney exchange~\cite{Abraham07:Clearing,Glorie14:Kidney,MakHau15:Kidney,Plaut16:Fast}; in our experimental section, we use the cycle pricer of~\citeN{Glorie14:Kidney} with the bugfix of~\citeN{Plaut16:Fast}.

Once no more positive price cycles exist, the reduced LP at a specific node is guaranteed to be optimal. However, it may not be integral: in this case, branching occurs, as in standard brand-and-bound tree search. In our experiments, we explore branches in depth-first-order unless optimality is proven at all nodes in the search tree.

\subsection{The LPR of PICEF is not as tight}\label{sec:picef-lpr}

As an analogue to Section~\ref{sec:pief}, we now compare the LPR of PICEF against the cycle formulation LPR.  Unlike in the PIEF case, where Theorem~\ref{thm:pief-vs-cf} showed an equivalence between the two models' relaxations, we show that PICEF's relaxation can be looser than that of the cycle formulation.  Theorem~\ref{thm:picef-vs-cf} gives a simple construction showing this, while Theorem~\ref{thm:picef-vs-cf-large} presents a family of graphs on which PICEF's LPR is arbitrarily worse than that of the cycle formulation.
\iftoggle{isFullVersion}{%
  The proofs of both of these results are contained in Section \ref{sec:PICEFLPRproofs}.
}{}

\begin{restatable}{theorem}{thmpicefvscf}%
  \label{thm:picef-vs-cf}
  $\zCF{} \prec \zPICEF{}$ (with chains).
\end{restatable}

Indeed, Theorem~\ref{thm:picef-vs-cf-large} shows that the ratio between the optimum objective value for the relaxations of
PICEF and the cycle formulation can be made arbitrarily large.

\begin{restatable}{theorem}{thmpicefvscflarge}%
  \label{thm:picef-vs-cf-large}
   Let $z \in \mathbb{R}^+$ be given. There exists a problem instance for which
   $Z_\text{PICEF} / Z_\text{CF} > z$, where $Z_\text{PICEF}$ is the objective
   value of the LPR of PICEF and $Z_\text{CF}$ is the objective value of the LPR
   of the cycle formulation.
\end{restatable}

While the results of Theorems~\ref{thm:picef-vs-cf} and~\ref{thm:picef-vs-cf-large} may be disheartening, in the following section, we give experimental evidence that PICEF (as well as its branch-and-price-based interpretation) perform extremely competitively on real and generated data.


\section{Experimental Comparison of State-of-the-Art Kidney Exchange Solvers}\label{sec:experiments}
In this section, we compare implementations of our new models against  existing state of the art kidney exchange solvers.  To ensure a fair comparison, we received code from the author of each solver that is not introduced in this paper.
We compare run times of the following state-of-the-art solvers:
\begin{itemize}[topsep=0pt]
  \item \algABS{}, the original branch-and-price-based cycle formulation solver due to~\citeN{Abraham07:Clearing};
  \item \algPDS{}, a branch-and-price-based cycle formulation solver with pricing due to~\citeN{Glorie14:Kidney} and~\citeN{Plaut16:Fast};\footnote{Recently,~\citeN{Plaut16:Hardness} showed a correctness bug in both implementations of the \algPDS{}-style solvers due to~\citeN{Glorie14:Kidney} and~\citeN{Plaut16:Fast}; for posterity, we still include these run times.  Furthermore, we note that the objective values returned by \algPDS{} always equaled that of the other provably-correct solvers on all of our test instances.}
  \item \algRMA{}, a recent IP formulation based on a model for the prize-collecting traveling salesman problem, with constraint generation~\cite{Anderson15:Finding};
  \item \algPICEF{}, the model from Section~\ref{sec:picef} of this paper;
  \item \algPICEFBnP{}, a branch and price version of the PICEF model, as presented in Section~\ref{sec:picef-bnp} of this paper;
  \item \algHPIEF{}, the Hybrid PIEF model from
    \iftoggle{isFullVersion}{%
      Appendix~\ref{sec:hybrid}
    }{%
      Appendix~C of the full version
    }%
    of this paper (which reduces to PIEF for $L=0$); and
  \item \algKlim{}, a branch-and-price algorithm using the Disaggregated Cycle Decomposition model, which is related to both the cycle formulation and the extended edge formulation \cite{Klimentova14:New}. 
\end{itemize}

A cycle-length cap of \num{3} and a time limit of \num{3600} seconds was imposed on each run. When a timeout occurred, we counted the run-time as \num{3600} seconds.

We test on two types of data: real and generated.  Section~\ref{sec:experiments-real} shows run time results on \emph{real} match runs, including \num{286} runs from the UNOS US-wide exchange, which now contains \num{143} transplant centers, and \num{17} runs from the NLDKSS UK-wide exchange, which uses \num{24} transplant centers.  Section~\ref{sec:experiments-generated} increases the size and varies other traits of the compatibility graphs via a realistic generator seeded by the real UNOS data.  We find that \algPICEF{} and \algHPIEF{} substantially outperform all other models.

\subsection{Real match runs from the UK- and US-wide exchanges}\label{sec:experiments-real}

We now present results on real match run data from two fielded nationwide kidney exchanges: The United Network for Organ Sharing (UNOS) US-wide kidney exchange where the decisions are made by algorithms and software from Prof.\ Sandholm's group, and the UK kidney exchange (NLDKSS) where the decisions are made by algorithms and software from Dr.\ Manlove's group.\footnote{Due to privacy constraints on sharing real healthcare data, the UNOS and NLDKSS experimental runs were necessarily performed on different computers---one in the US and one in the UK.  All runs \emph{within} a figure were performed on the same machine, so relative comparisons of solvers within a figure are accurate.} 
The UNOS instances used include all the match runs starting from the beginning of the exchange in October $2010$ to January $2016$. The exchange has grown significantly during that time and chains have been incorporated. The match cadence has increased from once a month to twice a week; that keeps the number of altruists relatively small. 
On average, these instances have $|N|=2$, 
$|P|=231$, and $|A|=5021$.
The NLDKSS instances cover the $17$ quarterly match runs during the period January $2012$--January $2016$. On average, these instances have $|N|=7$, $|P|=201$, and $|A|=3272$.


Figure~\ref{fig:real-data} shows mean run times across all match runs for both exchanges;
\iftoggle{isFullVersion}{%
  Appendix~\ref{app:tables}%
}{%
  Appendix~E in the full version of the paper%
}
gives additional statistics like minimum and maximum run times, as well as their standard deviations. 
Immediately obvious is that the non-compact formulations---\algABS{} and \algRMA{}---tend to scale poorly compared to our newer formulations.  Interestingly, \algPICEFBnP{} tends to perform worse than the base \algPICEF{} and \algHPIEF{}; we hypothesize that this is because branch-and-price-based methods are necessarily more ``heavyweight'' than standard IP techniques, and the small size of presently-fielded kidney exchange pools may not yet warrant this more advanced technique.  Perhaps most critically, both \algPICEF{} and \algHPIEF{} clear real match runs in both exchanges within seconds.

In the NLDKSS results, the wide fluctuation in mean run time as the chain cap is varied can be explained by the small sample size of available NLDKSS instances, and the fact that the algorithms other than HPIEF and PICEF occasionally timed out at one hour. By contrast, each of the HPIEF and PICEF runs on NLDKSS instances took less than five seconds to complete.
We also note that the LP relaxation of PICEF and HPIEF are very tight in practice; the LPR bound equaled the IP optimum for $614$ of the $663$ runs carried out on NLDKSS data.

\begin{figure*}[htb!]
  \vspace{-2mm}
  \centering
  \begin{minipage}{0.49\linewidth}
    \centering
    \includegraphics[width=1.0\linewidth]{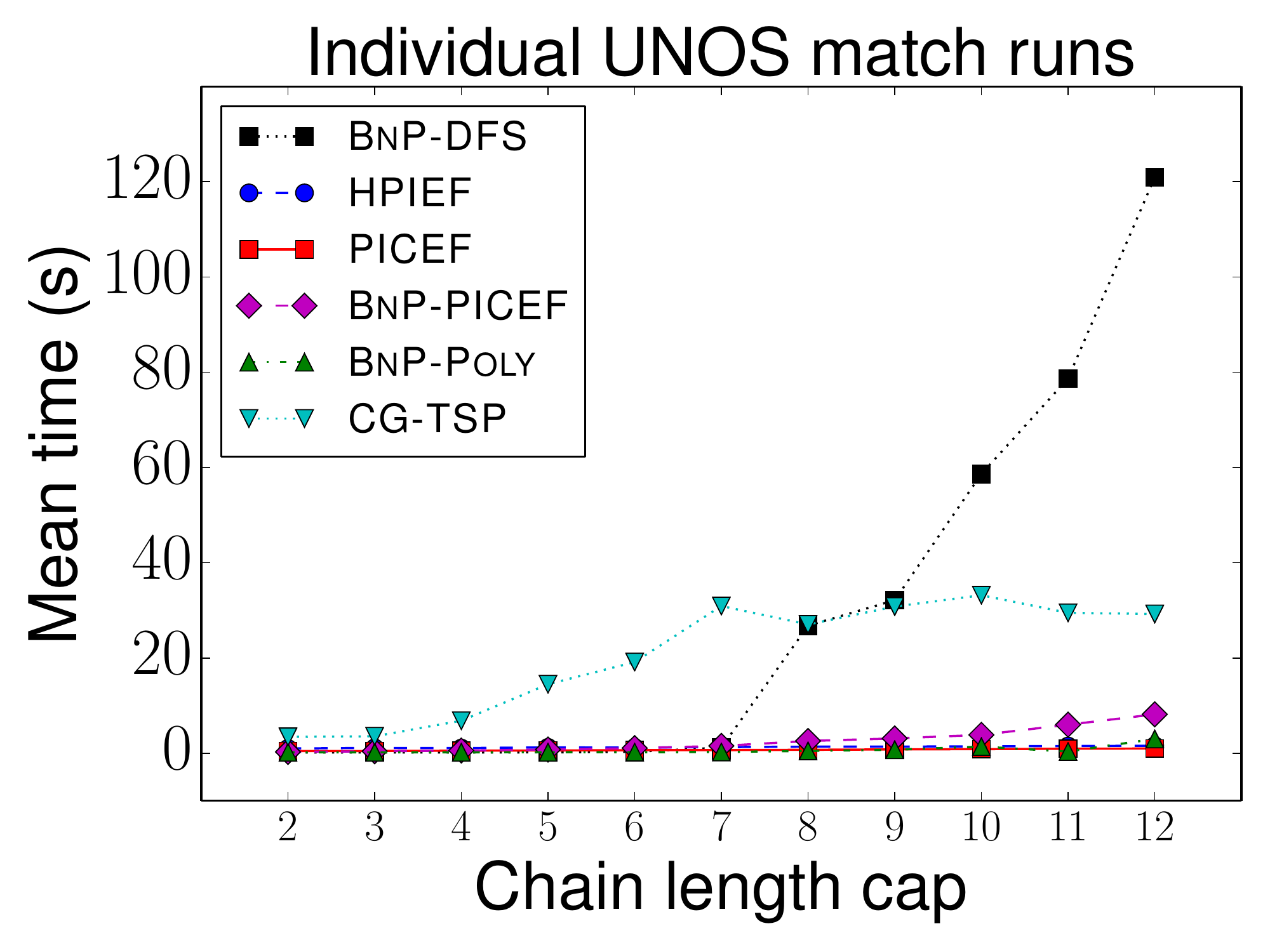}
  \end{minipage}%
  \hspace*{-0.1cm}
  \begin{minipage}{0.49\linewidth}
    \centering
    \includegraphics[width=1.0\linewidth]{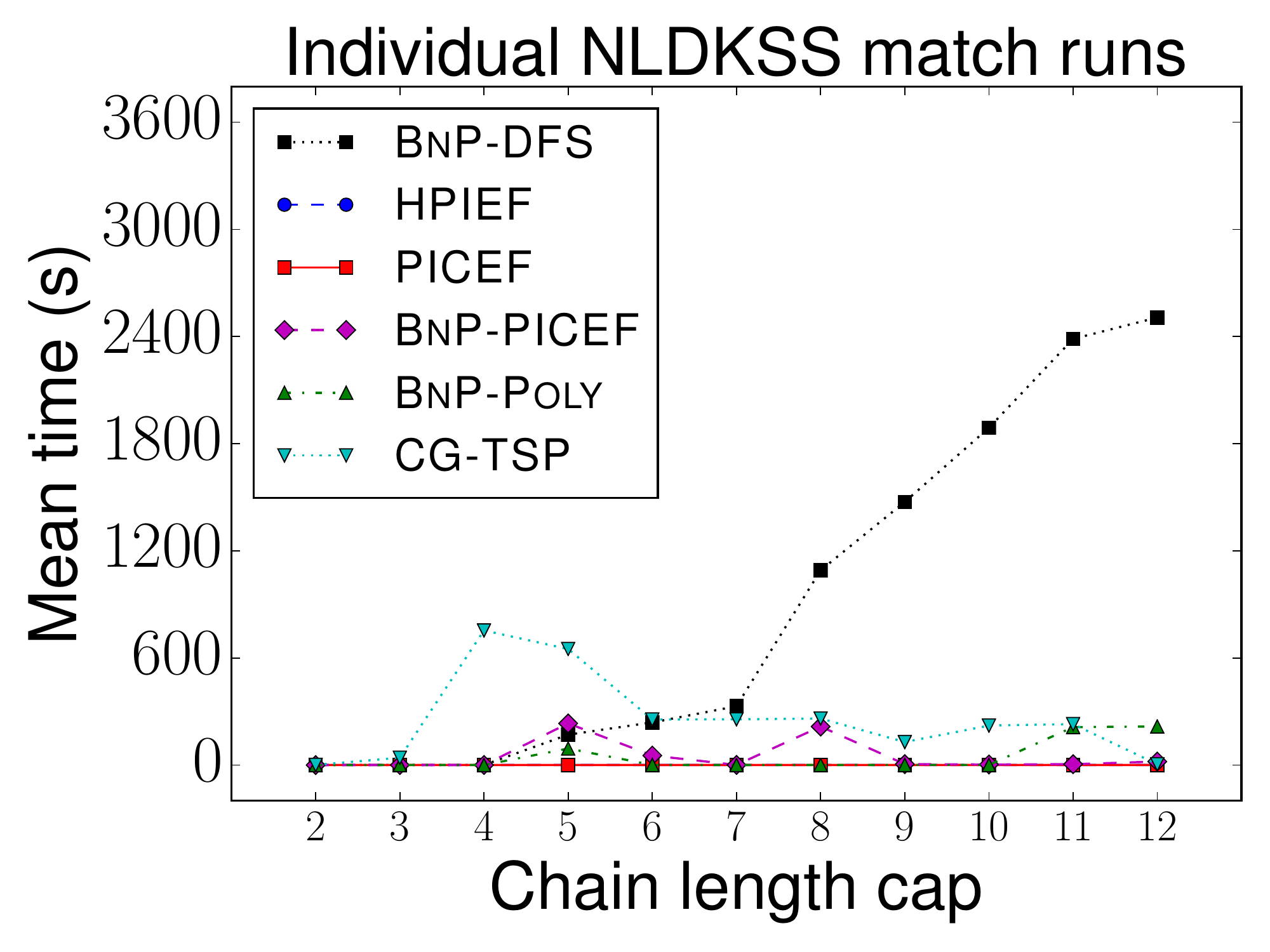}
  \end{minipage}%
  \caption{Mean run times for various solvers on \num{286} real match runs from the UNOS exchange (left), and \num{17} real match runs from the UK NLDKSS exchange (right).}
  \label{fig:real-data}
  \vspace{-3mm}
\end{figure*}

We remark that the \algKlim{} model due to \citeN{Klimentova14:New} was run on all NLDKSS instances where the chain cap $L$ was equal to $0$.  Larger values of $L$ could not be tested since the current implementation of the model in our possession does not accept NDDs in the input.  However for the case that $L=0$ the \algKlim{} model was the fastest for all NLDKSS instances.

Finally we note that the solver of~\citeN{Glorie14:Kidney} was executed on the NLDKSS instances with a chain cap of $L$, for $0\leq L\leq 4$.  It was found that on average the execution time was $8.9$ times slower than the fastest solver from among all the others executed on these instances as detailed at the beginning of Section~\ref{sec:experiments}.  PICEF was the fastest solver on $40$\% of occasions.

\subsection{Scaling experiments on realistic generated UNOS kidney exchange graphs}\label{sec:experiments-generated}

As motivated earlier in the paper, it is expected that kidney exchange pools will grow in size as (a) the idea of kidney exchange becomes more commonplace, and barriers to entry continue to drop, as well as (b) organized large-scale international exchanges manifest.  Toward that end, in this section, we test solvers on generated compatibility graphs from a realistic simulator seeded by all historical UNOS data; the generator samples patient-donor pairs and NDDs with replacement, and draws arcs in the compatibility graph in accordance with UNOS' internal arc creation rules.

Figure~\ref{fig:increasing-v-increasing-a} gives results for increasing numbers of patient-donor pairs (each column), as well as increasing numbers of non-directed donors as a percentage of the number of patient-donor pairs (each row).  As expected, as the number of patient-donor pairs increases, so too do run times for all solvers.  Still, in each of the experiments, for each chain cap, both \algPICEF{} and \algHPIEF{} are on par or (typically) much faster---sometimes by orders of magnitude compared to other solvers.
\iftoggle{isFullVersion}{%
  Appendix~\ref{app:tables}%
}{%
  Appendix~E in the full version of the paper%
}
gives these results in tabular form, including other statistics---minimum and maximum run times, as well as their standard deviations---that were not possible to show in Figure~\ref{fig:increasing-v-increasing-a}.

\begin{figure*}[htb!]
  \centering
  \begin{minipage}{0.33\linewidth}
    \centering
    \includegraphics[width=1.0\linewidth]{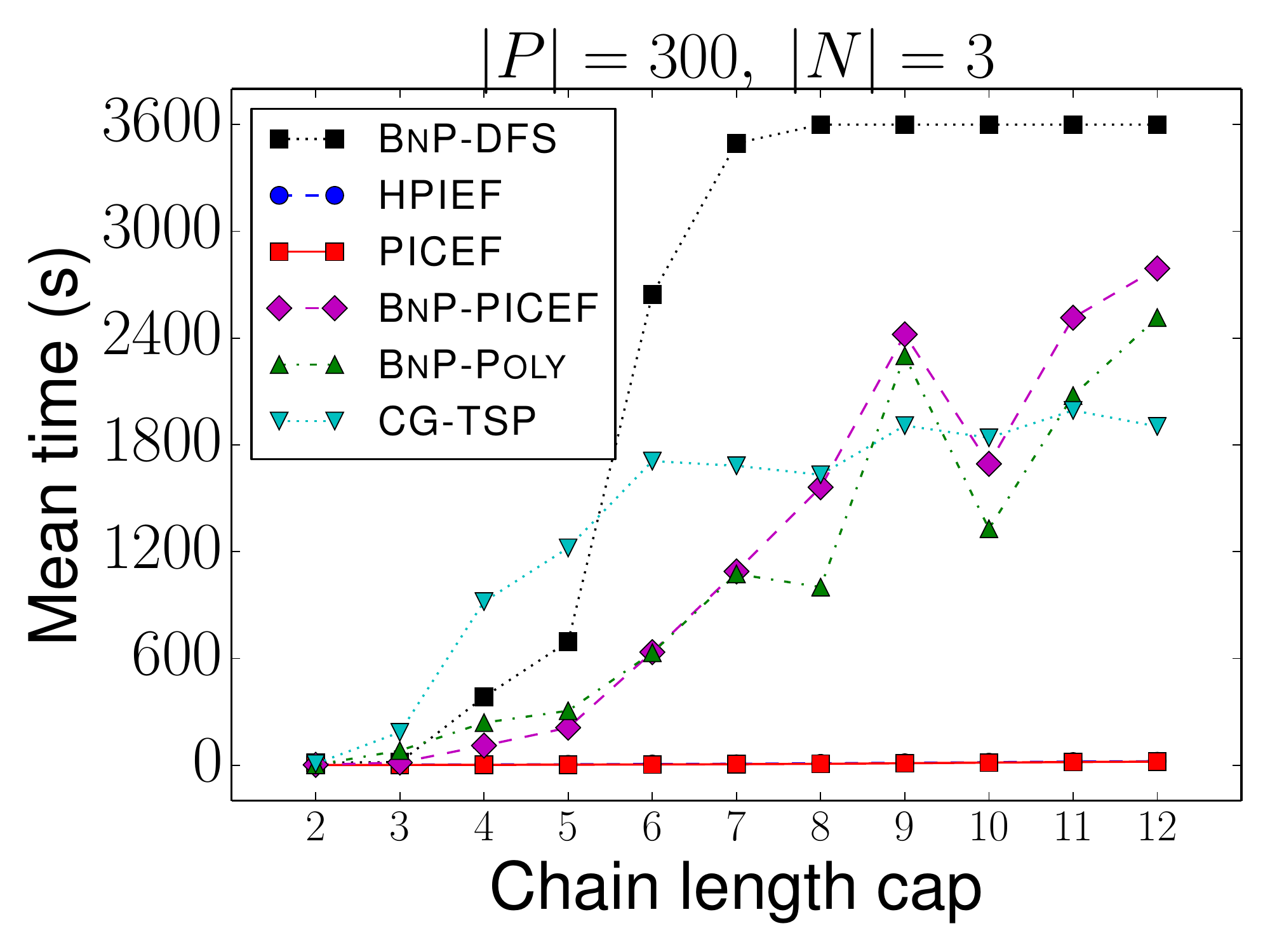}
  \end{minipage}%
  \hspace*{-0.1cm}
  \begin{minipage}{0.33\linewidth}
    \centering
    \includegraphics[width=1.0\linewidth]{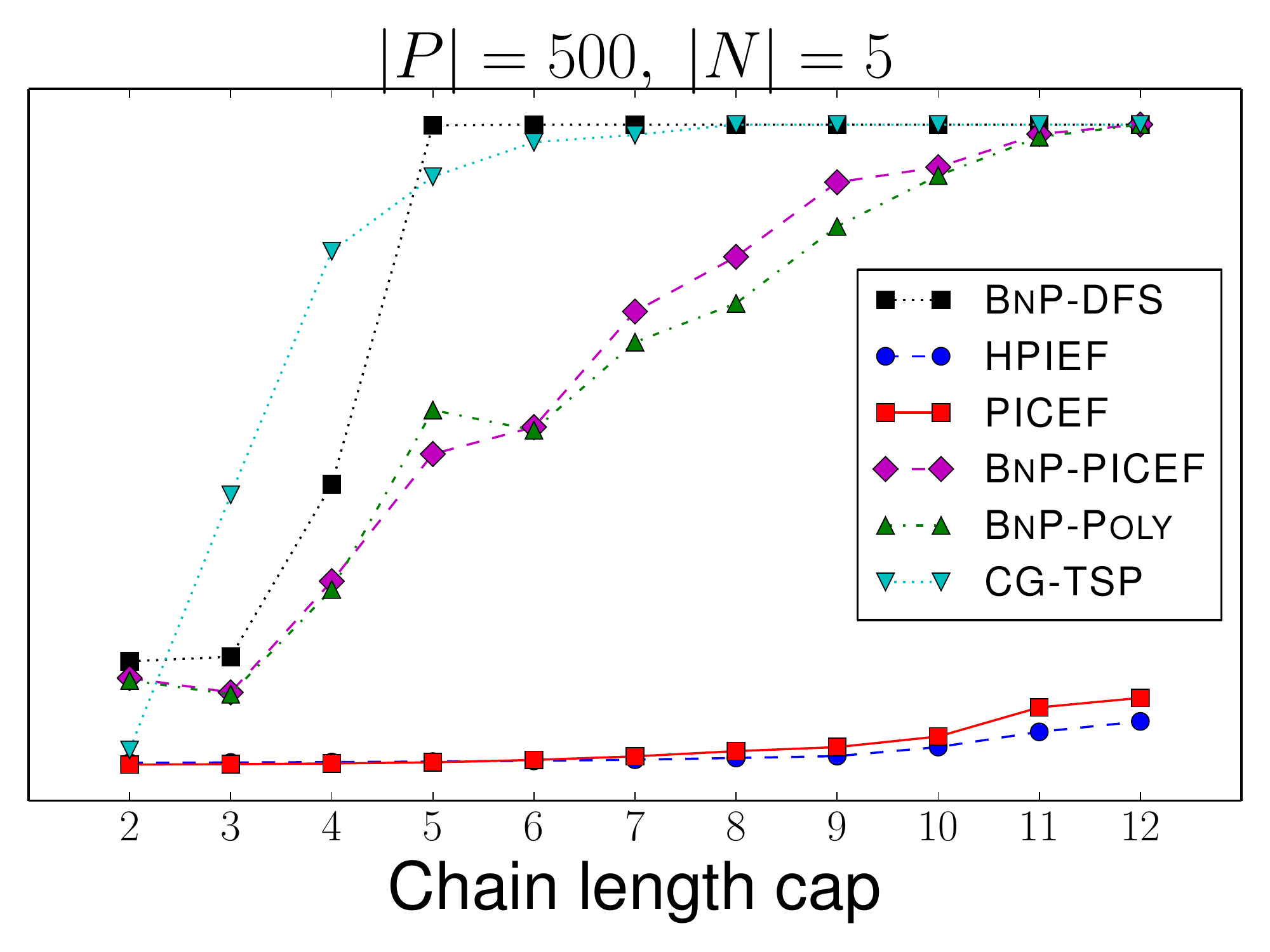}
  \end{minipage}%
  \hspace*{-0.1cm}
  \begin{minipage}{0.33\linewidth}
    \centering
    \includegraphics[width=1.0\linewidth]{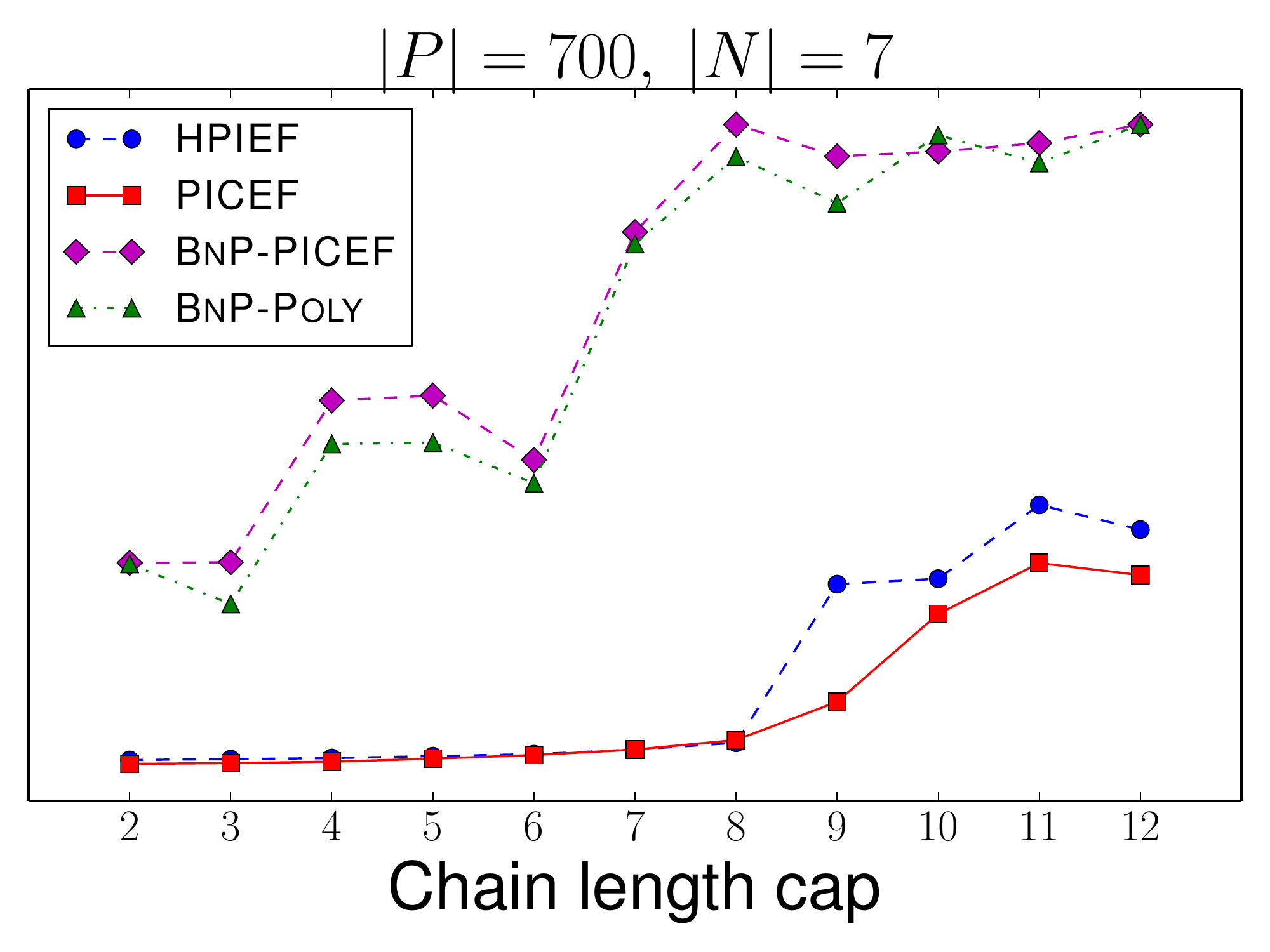}
  \end{minipage}%

  \begin{minipage}{0.33\linewidth}
    \centering
    \includegraphics[width=1.0\linewidth]{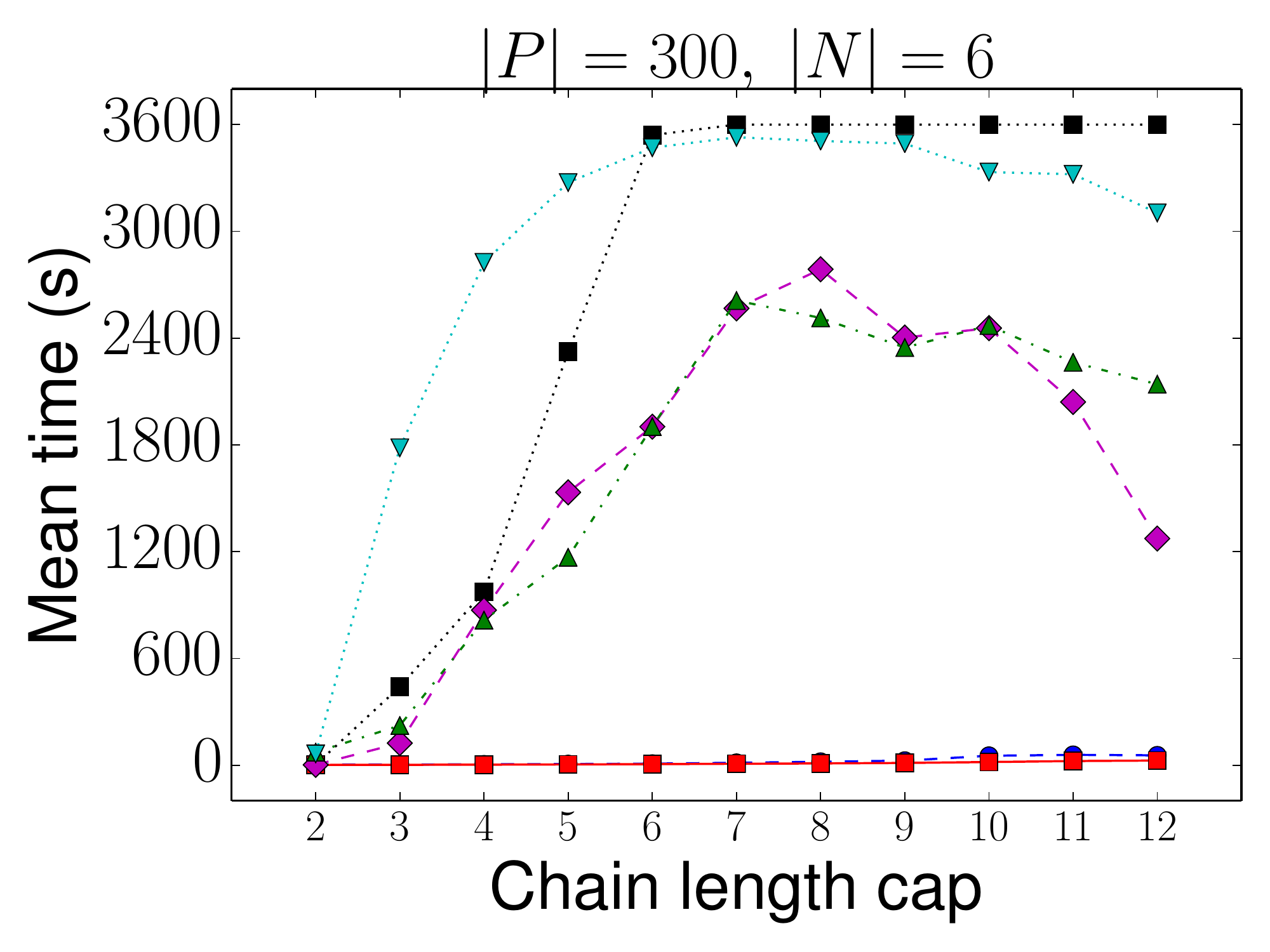}
  \end{minipage}%
  \hspace*{-0.1cm}
  \begin{minipage}{0.33\linewidth}
    \centering
    \includegraphics[width=1.0\linewidth]{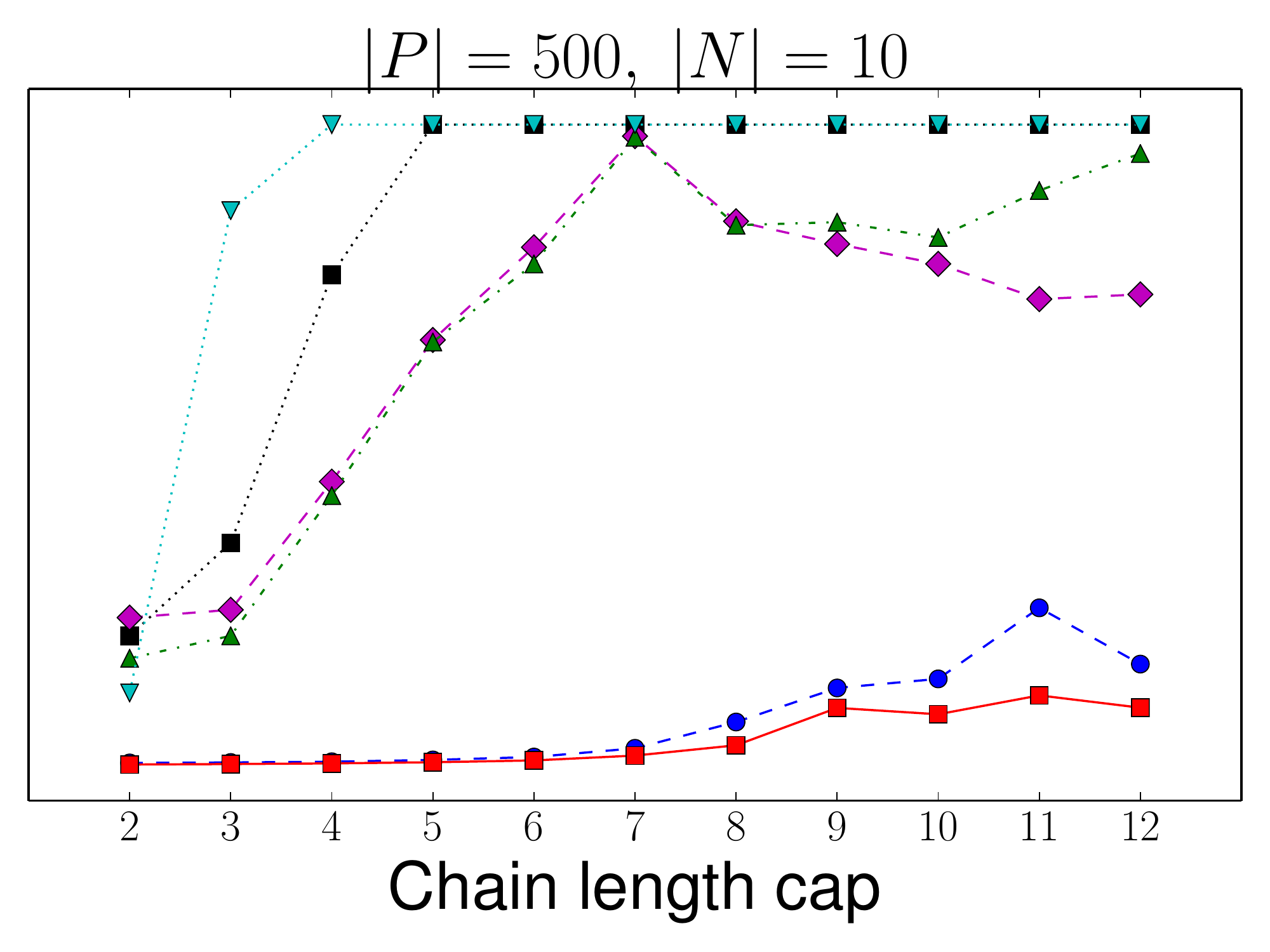}
  \end{minipage}%
  \hspace*{-0.1cm}
  \begin{minipage}{0.33\linewidth}
    \centering
    \includegraphics[width=1.0\linewidth]{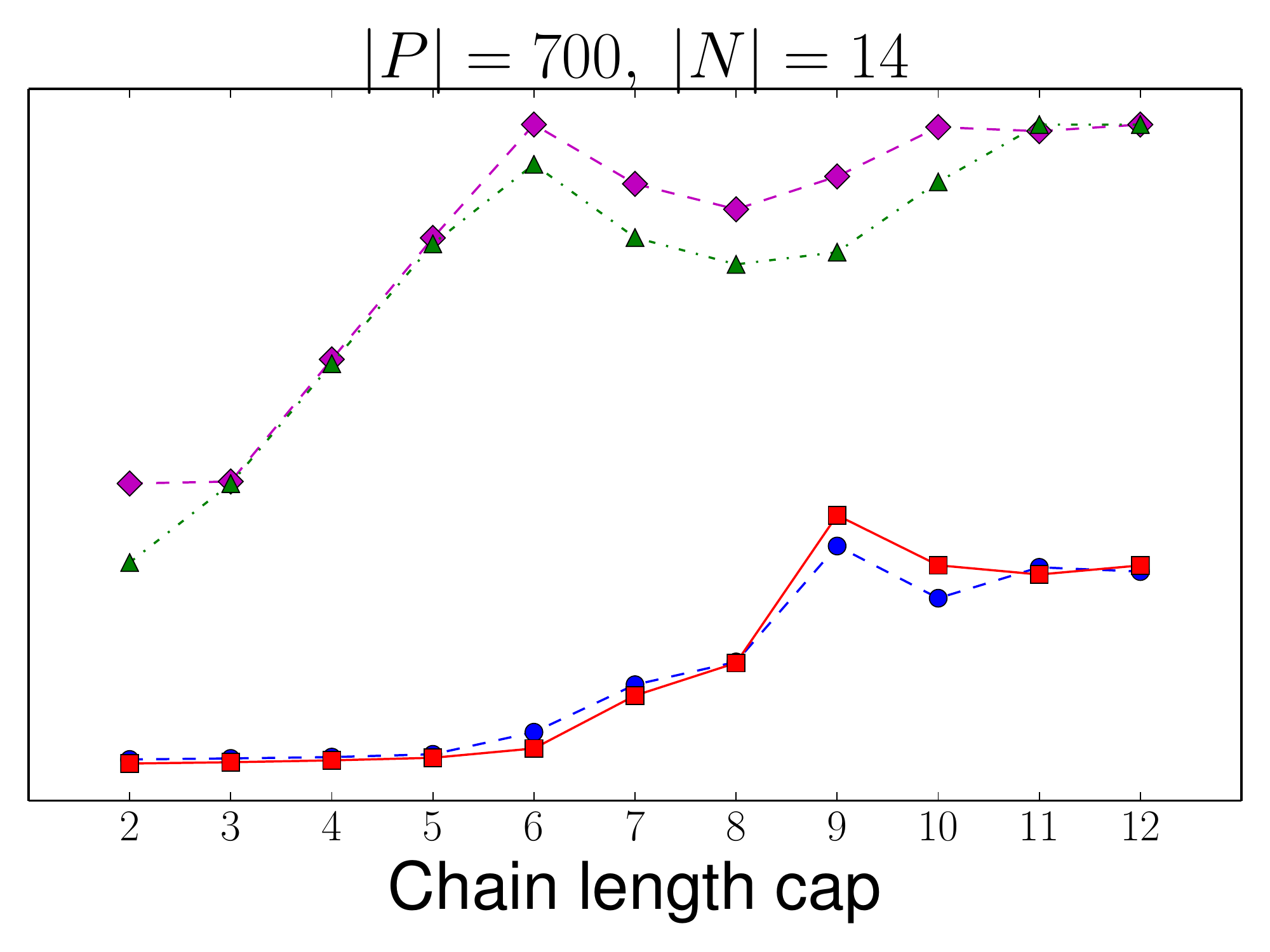}
  \end{minipage}%

  \begin{minipage}{0.33\linewidth}
    \centering
    \includegraphics[width=1.0\linewidth]{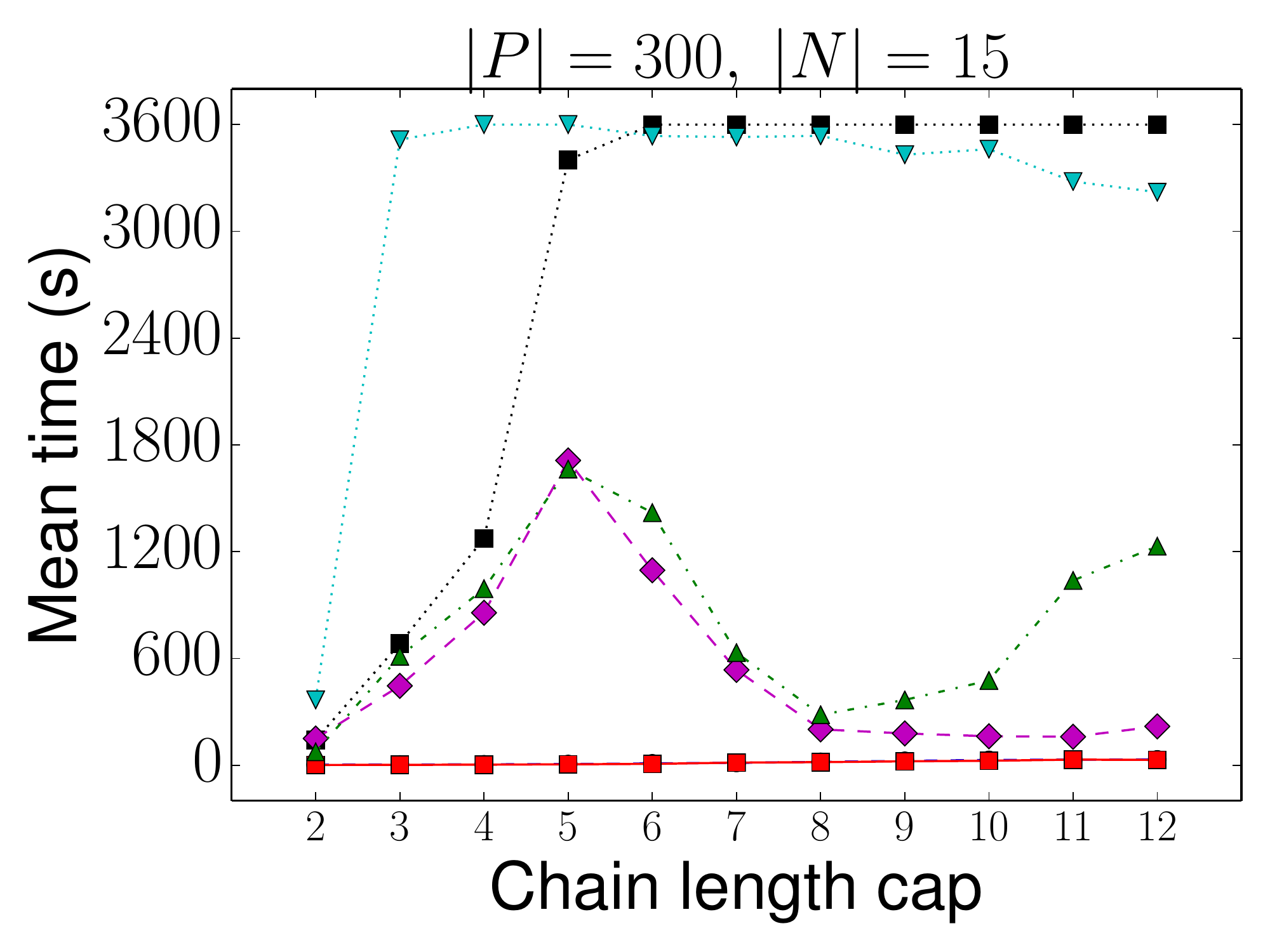}
  \end{minipage}%
  \hspace*{-0.1cm}
  \begin{minipage}{0.33\linewidth}
    \centering
    \includegraphics[width=1.0\linewidth]{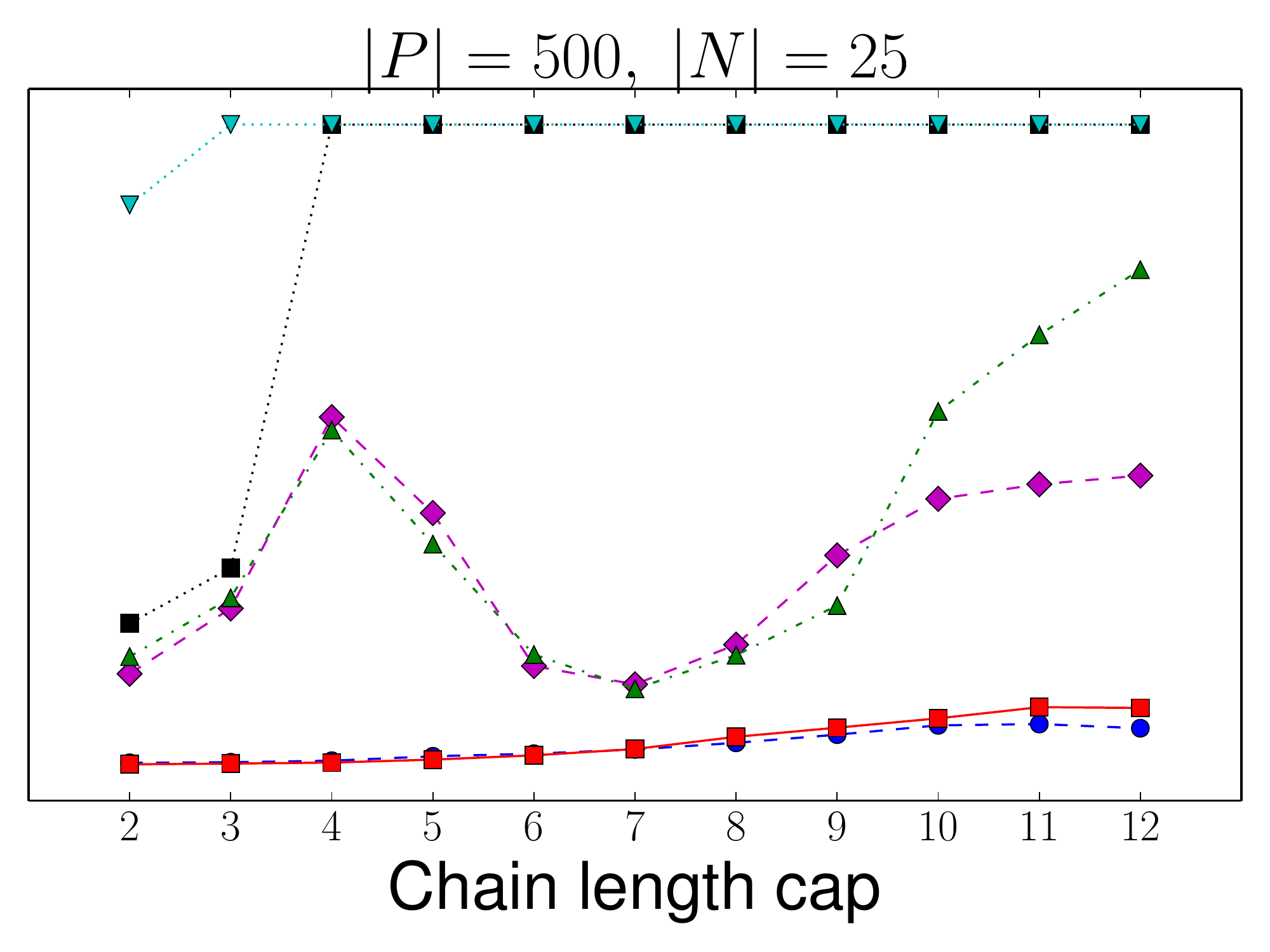}
  \end{minipage}%
  \hspace*{-0.1cm}
  \begin{minipage}{0.33\linewidth}
    \centering
    \includegraphics[width=1.0\linewidth]{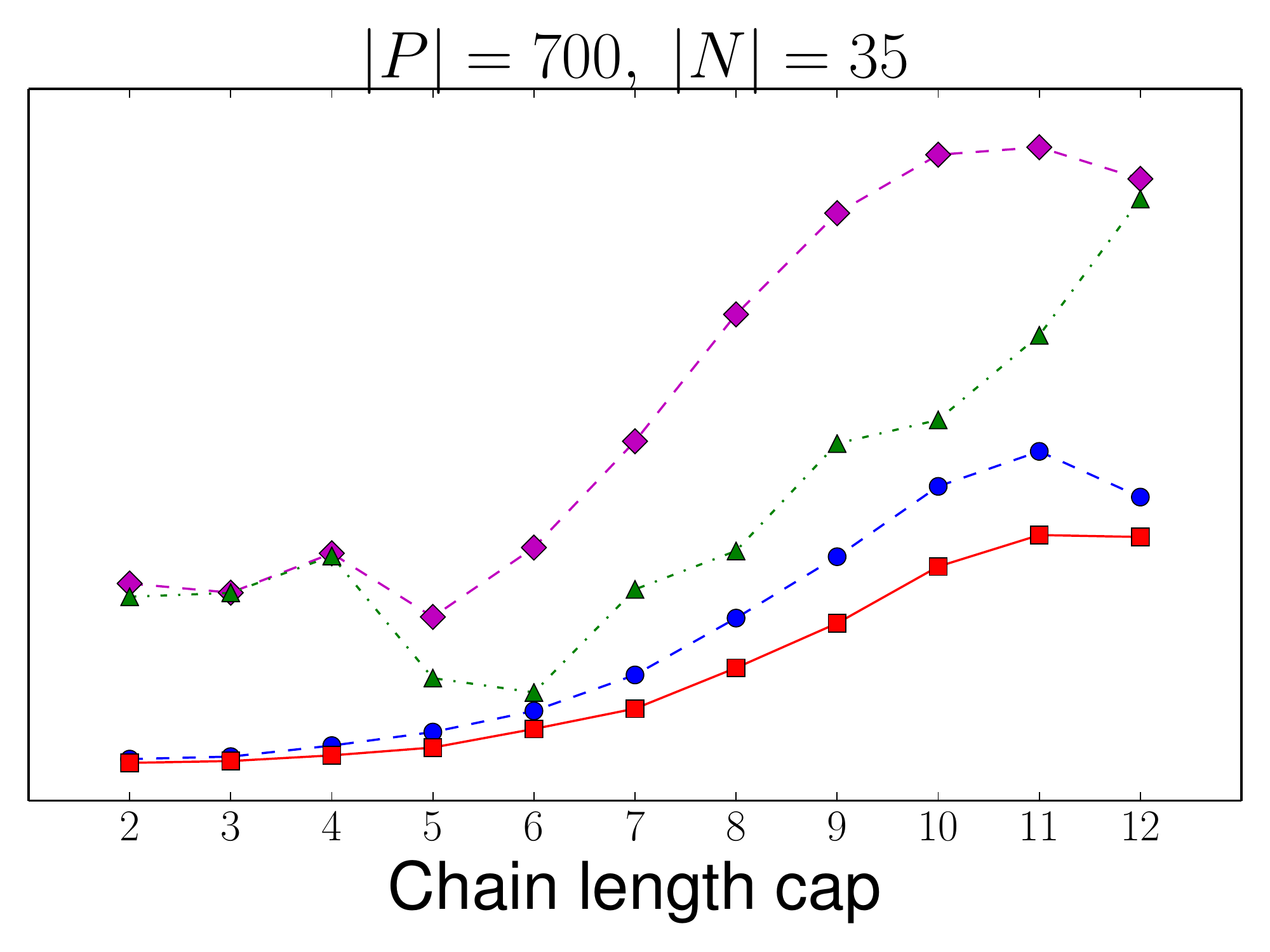}
  \end{minipage}%

  \begin{minipage}{0.33\linewidth}
    \centering
    \includegraphics[width=1.0\linewidth]{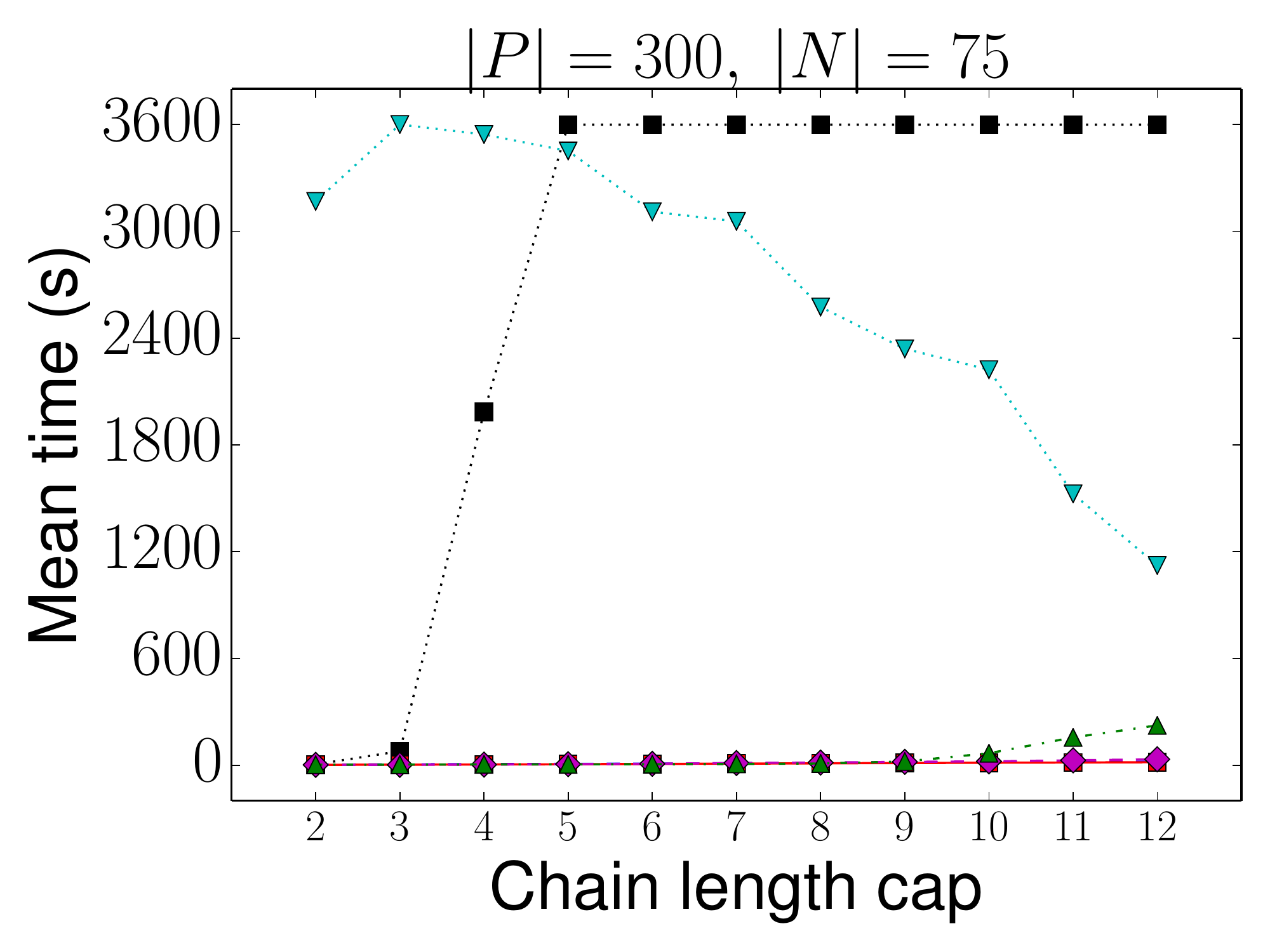}
  \end{minipage}%
  \hspace*{-0.1cm}
  \begin{minipage}{0.33\linewidth}
    \centering
    \includegraphics[width=1.0\linewidth]{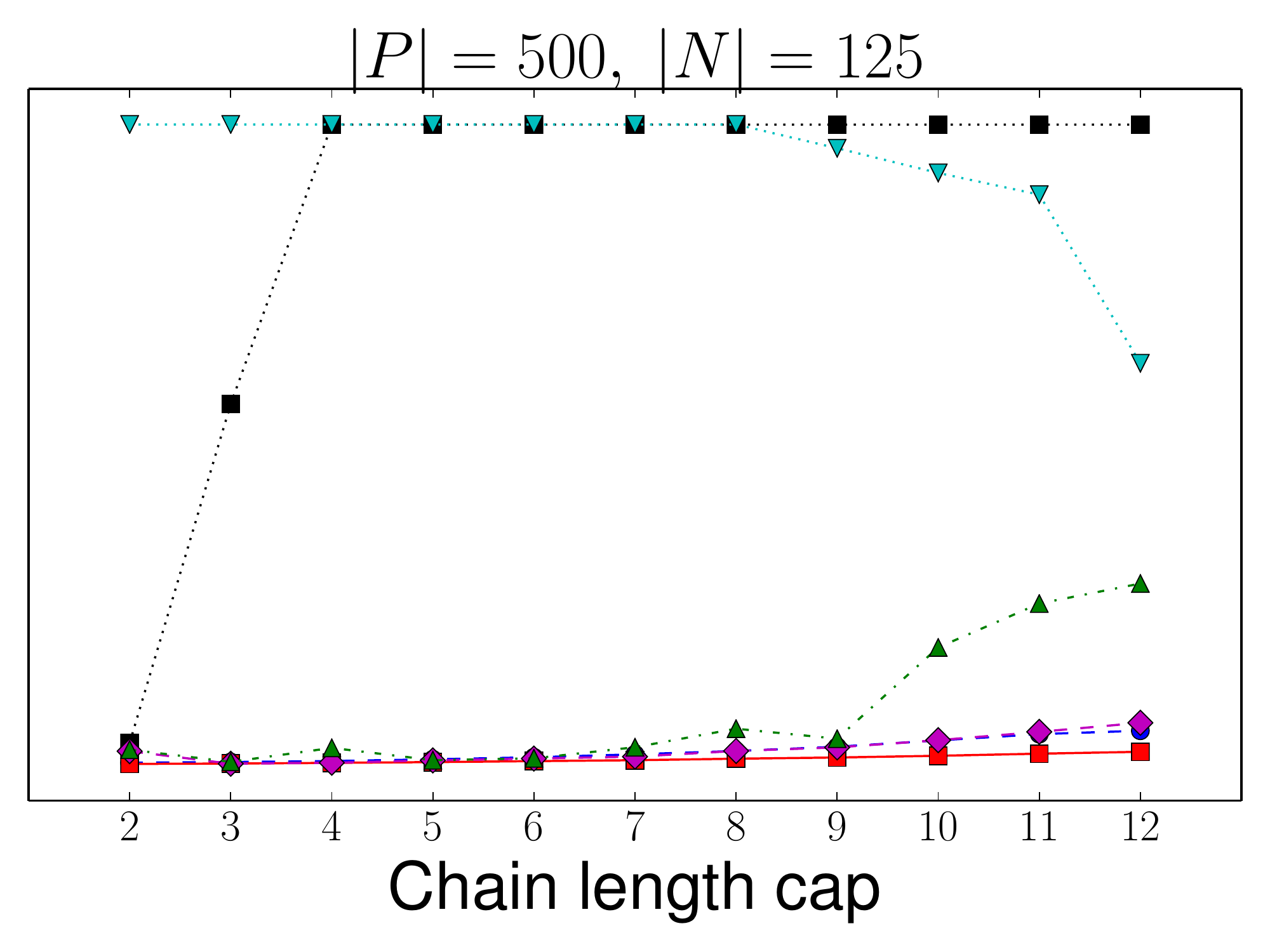}
  \end{minipage}%
  \hspace*{-0.1cm}
  \begin{minipage}{0.33\linewidth}
    \centering
    \includegraphics[width=1.0\linewidth]{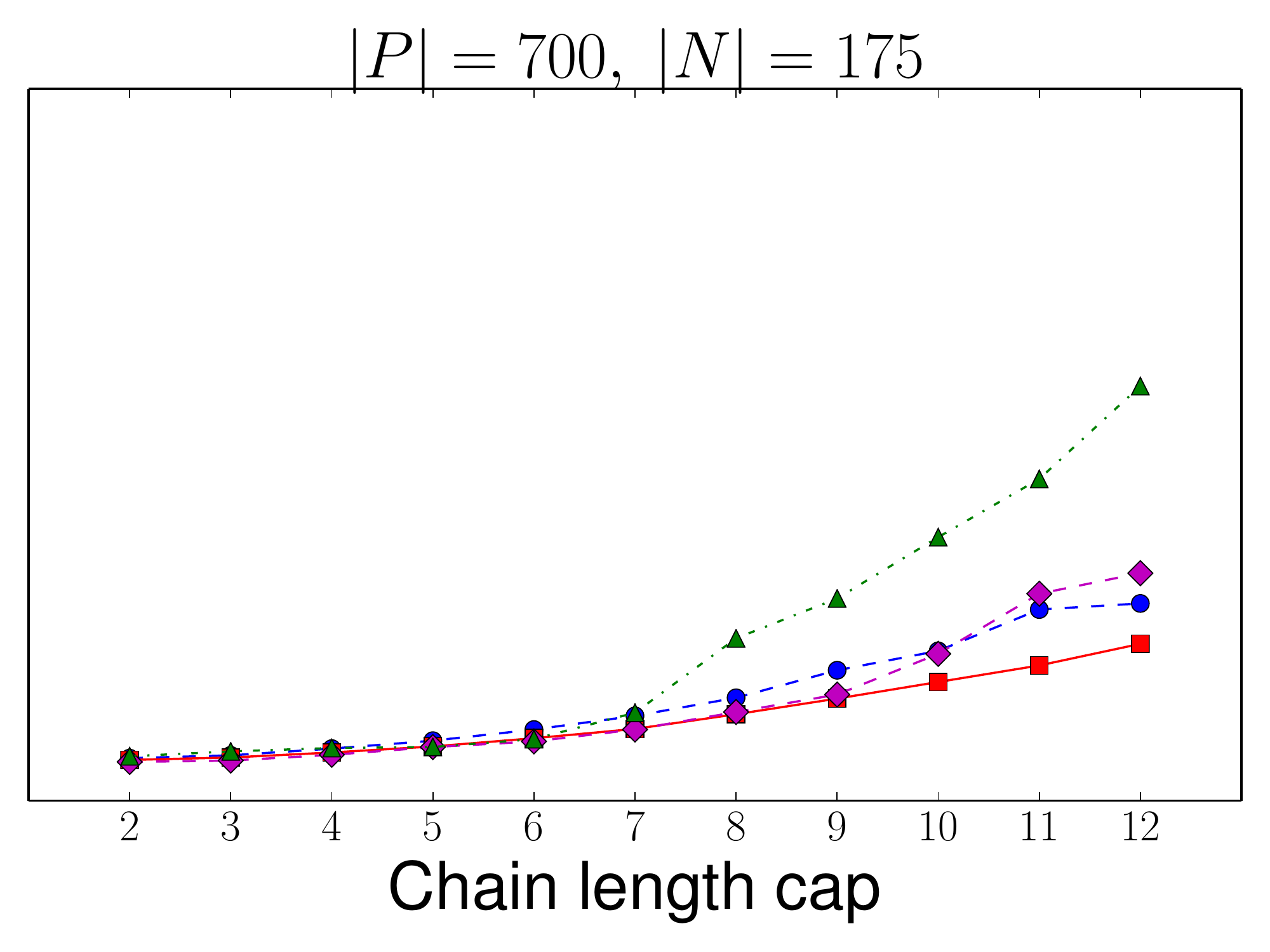}
  \end{minipage}%

  \caption{Mean run time as the number of patient-donor pairs $|P| \in \{300,500,700\}$ increases (left to right), as the percentage of NDDs in the pool increases $|N| = \{1\%, 2\%, 5\%, 25\%\} \text{ of } |P|$ (top to bottom), for varying finite chain caps.}
  \label{fig:increasing-v-increasing-a}
\end{figure*}

In addition to their increased scalability, we note two additional benefits of the PICEF and HPIEF models proposed in this paper: reduced variance in run time, and relative ease of implementation when compared to other state-of-the-art solution techniques.  In both the real and simulated experimental results, we find that the run time of both the PICEF and HPIEF formulations is substantially less variable than the branch-and-price-based and constraint-generation-based IP solvers.  While the underlying problem being solved is NP-hard, and thus will always present worst-case instances that take substantially longer than is typical to solve, the increased predictability of the run time of these models relative to other state-of-the-art solutions---including those that are presently fielded---is attractive.  Second, we note that \emph{significant} engineering effort is involved in the creation of custom branch-and-price and constraint-generation-based codes, while both PICEF and HPIEF are implemented with relative ease, relying on only a single call to a black box IP solver.

\section{Failure-aware kidney exchange}\label{sec:failure}
Real-world exchanges all suffer to varying degrees from ``last-minute'' failures, where an algorithmic match or set of matches fails to move to transplantation.  This can occur for a variety of reasons, including more extensive medical testing performed before a surgery, a patient or donor becoming too sick to participate, or a patient receiving an organ from another exchange or from the deceased donor waiting list.

To address these post-match arc failures,~\citeN{Dickerson13:Failure} augments the standard model of kidney exchange to include
a success probability $p$ for each arc in the graph.  They show how to solve this model using branch and price, where the pricing problem is solved in time exponential in the chain and cycle cap.  Prior compact formulations---and, indeed, prior ``edge formulations'' like those due to~\citeN{Abraham07:Clearing},~\citeN{Constantino13:New}, and~\citeN{Anderson15:Finding}---are not expressive enough to allow for generalization to this model.  Intuitively, while a single arc failure prevents an entire cycle from executing, chains are capable of incremental execution, yielding utility from the NDD to the first arc failure.  Thus, the expected utility gained from an arc in a chain is dependent on where in the chain that arc is located, which is not expressed in those models.

\subsection{PICEF for failure-aware matching}
With only minor modification, PICEF allows for implementation of failure-aware kidney exchange, under the restriction that each arc is assumed to succeed with equal probability $p$.  While this assumption of equal probabilities is likely not true in practice,~\citeN{Dickerson13:Failure} motivate why a fielded implementation of this model would potentially choose to equalize all failure probabilities: namely, so that already-sick patients---who will likely have higher failure rates---are not further marginalized by this model.  Thus, given a single success probability $p$, we can
\iftoggle{isFullVersion}{%
adjust the PICEF objective function to return the maximum expected weight matching as follows: 

\begin{subequations}
  {\small
\begin{align}
\max \qquad \sum_{(i,j)\in A} \sum_{k\in \K'(i,j)}  p^k w_{ij} y_{ijk} 
        + \sum_{c \in \cycles} p^{|c|} w_c z_c & \label{eq:picef-failure-obj}
\end{align}
  }
\end{subequations}

Objective~(\ref{eq:picef-failure-obj})
}{%
  simply rewrite the PICEF objective function as
  $\sum_{(i,j)\in A} \sum_{k\in \K'(i,j)}  p^k w_{ij} y_{ijk} + \sum_{c \in \cycles} p^{|c|} w_c z_c$,
  which returns the maximum expected weight matching.

  This objective
}%
is split into two parts: the utility received from arcs in chains, and the utility received from cycles.  For the latter, a cycle $c$ of size $|c|$ has probability $p^{|c|}$ of executing; otherwise, it yields zero utility.  For the former, if an arc is used at position $k$ in a chain, then it yields a $p^k$ fraction of its original weight---that is, the probability that the underlying chain will execute at least through its first $k$ arcs. 

\subsection{Failure-aware polynomial pricing for cycles}\label{ssec:failure-pricing}

The initial failure-aware branch-and-price work by~\citeN{Dickerson13:Failure} generalized the pricing strategy of~\citeN{Abraham07:Clearing}, and thus suffered from a pricing problem that ran in time exponential in cycle and chain cap. \citeN{Glorie14:Kidney} and~\citeN{Plaut16:Fast} discussed polynomial pricing algorithms for cycles---but not chains~\cite{Plaut16:Hardness}---in the \emph{deterministic} case. Using the algorithm of~\citeN{Plaut16:Fast} as a subroutine, we present an algorithm which solves the failure-aware, or \emph{discounted}, pricing problem for cycles in polynomial time, under the restriction that all arcs have equal success probability $p$.

In the deterministic setting, the price of a cycle $c$ is $\sum_{(i,j) \in c} w_{ij} - \sum_{j\in c}\delta_j$, where $w_{ij}$ is the weight of arc $(i,j)$, and $\delta_j$ is the dual value of vertex $j$ in the LP. \citeN{Glorie14:Kidney} show how the arc weights and dual values can be collapsed into just arc weights, and reduce the deterministic pricing problem to finding negative-weight cycles of length at most $K$ in a directed graph. In this section, we use ``length'' to denote the number of vertices in a cycle, not its weight.

The discounted price of a cycle is $p^{|c|}\sum_{(i, j) \in c} w_{ij} - \sum_{j\in c}\delta_j$. Since the utility of an arc now depends on what cycle it ends up in, we cannot collapse arc weights and dual values without knowing the length of the cycle containing it.

With this motivation, we augment the algorithm to run $O(K)$ iterations for each source vertex: one for each possible final cycle length. On each iteration, we know exactly how much arc weights will be worth in the final cycle, so we can reduce the discounted pricing problem to the deterministic pricing problem.

Pseudocode for the failure-aware cycle pricing algorithm is given by \textsc{GetDiscountedPositivePriceCycles}. Let $w$ and $\delta$ be the arc weights and dual values respectively in the original graph. The function \textsc{GetNegativeCycles} is the deterministic pricing algorithm due to~\citeN{Plaut16:Fast} which returns at least one negative cycle of length at most $K$, or shows that none exist.

The algorithm of~\citeN{Plaut16:Fast} has complexity $O(|V||A|K^2)$. Considering all $K-1$ possible cycle lengths brings the complexity of our algorithm to $O(|V||A|K^3)$.

\begin{algorithm}[tb]
\centering

\begin{algorithmic}[1]
  \Function{GetDiscountedPositivePriceCycles}{$D=(V,A), K, p, w, \delta$}
  \State ${\cal C} \gets \emptyset$ 
  \ForAll{$k = 2...K$} \Comment{Consider all possible cycle lengths}
    \State $w_k(i,j) \gets \delta_j - p^k w_{ij} \ \ \forall (i,j)\in A$ \Comment{Reduction of~\citeN{Glorie14:Kidney}} 
  \State ${\cal C} \gets {\cal C} \cup \Call{GetNegativeCycles}{D, k, w_k}$
  \EndFor  
  \Return ${\cal C}$
\EndFunction

\end{algorithmic}
\caption{Polynomial-time failure-aware pricing for cycles.}
\label{alg:failure-pricing}
\end{algorithm}
 
\begin{restatable}{theorem}{thmfailureaware}%
\label{thm:failure-aware}
If there is a discounted positive price cycle in the graph, Algorithm~\ref{alg:failure-pricing} will return at least one discounted positive price cycle.
\end{restatable}

\section{Conclusions \& Future Research}\label{sec:conclusions}

In this paper, we addressed the tractable clearing of kidney exchanges with short
cycles and long, but bounded, chains.  This is motivated by kidney exchange practice, where chains are often long but bounded in length due to post-match arc failure.  We introduced three IP formulations, two of which are compact, and favorably compared their LPRs to a state-of-the-art formulation with a tight relaxation.  Then, on real data from the UNOS US nationwide exchange and the NLDKSS United Kingdom nationwide exchange, as well as on generated data, we showed that our new models outperform all other solvers on
realistically-parameterized kidney exchange problems--often dramatically.  We also explored practical extensions of our models, such as the use of branch and price for additional scalability, and an extension to the failure-aware kidney exchange case that more accurately mimics reality.

Beyond the immediate importance of more scalable static kidney exchange solvers for use in fielded exchanges, solvers like the ones presented in this paper are of practical importance in more advanced---and as yet unfielded---approaches to clearing kidney exchange.  In reality, patients and donors arrive to and depart from the exchange dynamically over time~\cite{Unver10:Dynamic}.  Approaches to clearing \emph{dynamic kidney exchange} often rely on solving the static problem many times~\cite{Awasthi09:Online,Dickerson12:Dynamic,Anderson14:Stochastic,Dickerson15:FutureMatch,Glorie15:Robust}; thus, faster static solvers result in better dynamic exchange solutions.  Use of the techniques in this paper---or adaptations thereof---as subsolvers is of interest.

From a theoretical point of view, extending the comparison of LPRs to a complete ordering of all LPRs amongst models of kidney exchange---especially for different parameterizations of the underlying model, like the inclusion of chains or arc failures---would give insight as to which solver is best suited for an exchange running under a specific set of business constraints.

\vspace{-4mm}
\begin{acks}
\vspace{-2mm}
The authors would like to thank Ross Anderson, Kristiaan Glorie, Xenia Klimentova, Nicolau Santos, and Ana Viana for valuable discussions regarding this work and for making available their kidney exchange software for the purposes of conducting our experimental evaluation.
\end{acks}
\vspace{-3mm}

\iftoggle{isFullVersion}{}{%
  \let\oldbibliography\thebibliography
  \renewcommand{\thebibliography}[1]{\oldbibliography{#1}
    \setlength{\itemsep}{-0.5pt}} 
}

\iftoggle{isFullVersion}{%
\bibliographystyle{ACM-Reference-Format-Journals}
\bibliography{dairefs}
}{%
\input{hierarchy.ec2016.bbl_FOR_CAMERAREADY}
}


\iftoggle{isFullVersion}{%

\elecappendix
\medskip

\section{Additional Proofs for the PIEF model}

We now provide additional theoretical results pertaining to the position-indexed edge formulation (PIEF) model, and proofs to theoretical results stated in the main paper. 

\subsection{Validity of the PIEF model}
\label{sec:PIEFcorrect}
\begin{lemma}\label{thm:pief-validity-a}
Any assignment of values to the $x_{ijk}^l$ that respects the
PIEF constraints yields a vertex-disjoint set of cycles of length no greater
than $K$.
\end{lemma}

\begin{proof}
We show this by demonstrating that in each graph copy, the set of selected
edges is either empty, or composes a single cycle of length no greater than
$K$.

Let $l \in P$ be given such that at least one edge is selected in graph copy
$D^l$, and let $k_{\text{max}}$ be the highest position $k$ such that
$x_{ijk}^l=1$ for some $i,j$. Any selected edge $(i,j)$ at position
$k_{\text{max}}$ in graph copy $D^l$ must point to $l$, as otherwise the flow
conservation constraint (\ref{eq:pief_b}) would be violated at vertex $j$.
Furthermore, there must be no more than one edge selected at position
$k_{\text{max}}$ in graph copy $D^l$, as otherwise the capacity constraint
(\ref{eq:pief_b}) for vertex $l$ would be violated.

The flow conservation constraints (\ref{eq:pief_a}) ensure that we can follow a
path backwards from the selected edge at position $k_{\text{max}}$ to a selected
edge at position 1, and also that at most one edge is selected at each position.
Since the edge at position 1 must start at vertex $l$ by the construction of $\K(i,j,l)$, we have
shown that graph copy $D^l$ contains a selected cycle beginning and ending
at $l$, and that this graph copy does not contain any other selected edges.

Constraint (\ref{eq:pief_a}) ensures that the vertex-disjointness condition is satisfied.
\end{proof}

\begin{lemma}
For any vertex-disjoint set of cycles of length no greater than $K$, there is an assignment of values to the $x_{ijk}^l$ respecting the PIEF constraints.
\end{lemma}

\begin{proof}
This assignment can be constructed trivially.
\end{proof}

\begin{theorem}
\label{thm:PIEFcorrect}
The PIEF model yields an optimal solution to the kidney exchange problem.
\end{theorem}

\subsection{Proofs for the LPR of PIEF}
\label{sec:PIEFLPRproofs}
The following lemma is used in the proof of Theorem~\ref{thm:pief-vs-cf}, and its proof
is included here for completeness.

\begin{lemma}\label{lem:closed-walk}
Let a sequence of arcs $W=(a_1,\dots,a_{|W|})$ in a directed graph $D$ be given, such that $W$ is a closed walk---that is, the target vertex of each arc is the source vertex of the following arc, and the sequence starts and ends at the same vertex.  (It is permitted for an arc to appear more than once in $W$.)  Let $X$ be the multiset
$\{a_1,\dots,a_{|W|}\}$. Then we can partition $X$ into $C=\{c_1, \dots, c_{|C|}\}$, where each $c_i \in C$ is a set of arcs that
form a cycle in $D$.
\end{lemma}

\begin{proof} The following algorithm can be used to construct the set $C$.
\begin{enumerate}
  \item Let $C=\{\}$.
  \item\label{step:label_vertices}For each $i \in \{1,\dots,|W|\}$,
    let $s_i$ be the source vertex of $a_i$.
  \item \label{step:choose}If the sequence $(s_i)_{1 \leq i \leq |W|}$ contains no repeated
    vertices, then $W$ must be a cycle; go to step~\ref{step:one_cycle}.
  \item \label{step:sub_cycle} Choose $i,j \in \{1,\dots,|W|\}$ with $i < j$,
    such that $s_i = s_j$ and all the $s_k$ are
    distinct for $i \leq k < j$. The arcs $(a_k)_{i \leq k < j}$ form a cycle;
    remove these from $W$ and add the set of removed arcs to $C$.
    (Observe that $W$ remains a closed walk). Re-index the arcs in the new,
    shorter $W$ as $1,\dots,|W|$.
    Return to step~\ref{step:label_vertices}.
  \item \label{step:one_cycle}Add $\{a:a\text{ appears in } W\}$ to $C$ and terminate.
\end{enumerate}
\vspace{-6.3mm}\hspace{79mm}
\end{proof}
\vspace{0.1mm}


\thmpiefvscf*

\begin{proof}
$\zCF{} \preceq \zPIEF{}$. Let $(z_c^*)_{c \in \cycles}$ be an optimal
solution to the LPR of the cycle formulation, with objective value
$\zCF{}$. We will construct a solution to the LPR of PIEF whose objective
value is also $\zCF{}$. We translate the $z_c^*$ into an assignment of
values to the $x_{ijk}^l$ in a natural way, as follows.  For each $c \in
\cycles$, let $l$ be the index of the
lowest-numbered vertex appearing in $c$.  Number the positions of arcs of $c$ in order as
$\{1,\dots,|c|\}$, beginning with the arc leaving $l$.

For each vertex $l \in V$, each arc $(i,j) \in A^l$, and each position $k \in
\K(i,j,l)$, let $\cycles(i,j,k,l)$ be the set of cycles in $\cycles$ whose
lowest-numbered vertex is $l$ and which contain $(i,j)$ at position $k$. Let

\[
x_{ijk}^l = \sum_{c \in \cycles(i,j,k,l)}
    z_c^*.
\]

This construction yields a solution which
satisfies the PIEF constraints and has objective value $\zCF{}$.

$\zPIEF{} \preceq \zCF{}$.  Let an optimal solution $(x_{ijk}^{l})$ to the LPR
of PIEF be given, with objective value $\zPIEF{}$.  Our strategy
is to begin by assigning zero to each cycle formulation variable $z_c (c \in
\cycles)$, and to make a series of decreases in PIEF variables and corresponding increases in cycle formulation variables, ending when all of the PIEF variables are set to zero.
We maintain three invariants after each such step.  First, the sum of the
PIEF and cycle formulation objective values remains $\zPIEF{}$.  Second, the
constraints of the relaxed PIEF are satisfied. Third, the following vertex
capacity constraint---which combines the capacity constraints from the
PIEF and the cycle formulation---is satisfied for each vertex.

\[
\sum_{l\in V} \sum_{i:(i,j)\in A^l} \sum_{k\in \K(i,j,l)} x_{ijk}^l +
   \sum_{c : j \in c} z_c \leq 1 \qquad j \in P \label{eq:cf_pief_capacity}
\]

If any of the PIEF variables takes a non-zero value, then we can select $i_1, i_2, l$ such that $i_1=l$ and $x_{i_1 i_2 1}^l > 0$. By the PIEF flow conservation constraints (\ref{eq:pief_b}), we can select a closed walk $W=((i_1, i_2), (i_2, i_3), \dots, (i_{k'-1},i_{k'}))$ of at most $K$ arcs in $D^l$ such that $x_{i_k i_{k+1} k}^l > 0$ for $1 \leq k < k'$, and such that $i_1=i_{k'}$. Let $x_{\min}$ be the smallest non-zero value taken by any of the $x_{i_k i_{k+1} k}^l$.

By Lemma~\ref{lem:closed-walk}, the set of arc positions $\{1, \dots, k'-1\}$ can be decomposed into a set of sets $C$, such that for each $c \in C$ we have that the arcs $\{a_k : k \in c\}$ can be arranged to form a cycle; we denote this as $\text{cyc}(c)$.  For each $c \in C$, we subtract $x_{\min}$ from $x_{i_k i_{k+1} k}^l$ for each $k \in c$, and we add $x_{\min}$ to $z_{\text{cyc}(c)}$.  This transformation strictly decreases the count
of $x_{ijk}^l$ variables that take a non-zero value. By repeatedly carrying out
this step, we will reach a point where all of the $x_{ijk}^l$ variables take
the value zero, and where the $z_c$ variables respect the cycle formulation
constraints and give objective value $\zPIEF{}$.\end{proof}


\section{Additional Proofs for the PICEF model}
We now provide additional theoretical results pertaining to the position-indexed chain-edge formulation (PICEF) model, and proofs to theoretical results stated in the main paper. 

\subsection{Validity of the PICEF model}
\label{sec:PICEFvalid}
\begin{lemma}\label{lemma:picef_chains}
Any assignment of values to the $y_{ijk}$ that respects constraints
(\ref{eq:picef_b}) and (\ref{eq:picef_c}) and such that

\begin{equation}\label{eq:chain_flow}
\sum_{j:(j,i)\in A} \sum_{k\in \K(j,i)} y_{jik} \leq 1
\end{equation}

for all $i \in P$, yields a vertex-disjoint set chains of length no greater than $L$.
\end{lemma}

\begin{proof}
We say that arc $(i,j)$ is \textit{selected} at position $k$ if and only if
$z_{ijk} = 1$.

Our proof has three parts. We first give a procedure to construct a set $S$ of
chains of length no greater than $L$, where each chain in $S$ consists only of
selected edges. We then show that these chains are vertex-disjoint, and that
any selected edge appears in some chain in $S$.

By constraint~(\ref{eq:picef_b}), each $i \in N$ has at most one selected
outgoing arc. For each $i \in N$ that has an outgoing arc $(i,j_1)$, we begin
to construct chain $c$ by letting $c=((i,j_1))$---a sequence containing one
edge.  Vertex $j_1$ has at most one selected outgoing arc at position 2, by
constraint~(\ref{eq:picef_c}).  If such an arc $(j_1,j_2)$ exists, we add it to
our chain. We continue to add edges $(j_2,j_3), (j_3,j_4), \dots$, until we
reach $k$ such that a selected edge from $j_k$ at position $k+1$ does not
exist.  The chain $c$ will therefore be a path of selected edges at positions
$1, \dots, |c|$, where the length of $c$ can be no greater than $L$ since no
variable in the model has position greater than $L$.  Add the chain $c$ to $S$.

By constraint (\ref{eq:picef_b}), no vertex in $N$ can appear in two chains in
$S$.  By constraint (\ref{eq:picef_a}), the same is true for vertices in $P$.
Hence, the chains in $S$ are vertex disjoint.

To complete the proof, we show that any selected edge must be part of one of
these chains in $S$. Let a variable $z_{ijk}$ taking the value 1 be given. By
applying constraint (\ref{eq:picef_c}) repeatedly, we can see that there must exist
be a path of length $k$ from an NDD $h$ to $j$, containing only selected edges.
Let $c \in S$ be the chain starting at $h$. Since no vertex in $c$ has two
selected outgoing arcs, there must must exist a unique path of length $k$ from $h$, and
$(i,j)$ must therefore be the $k$th edge of $c$.
\end{proof}

\begin{lemma}\label{thm:picef_validity_a}
Any assignment of values to the $y_{ijk}$ and $z_c$ that respects the PICEF constraints yields a vertex-disjoint set cycles of length no greater than $K$ and chains no greater than $L$.
\end{lemma}

\begin{proof}
We call a cycle $c$ such that $z_c=1$ a \textit{selected cycle}, and an arc
$(i,j)$ such that $z_{ijk}=1$ for some $k$ a \textit{selected arc}.

By (\ref{eq:picef_a}), the selected cycles are vertex disjoint.  By
Lemma~\ref{lemma:picef_chains} the selected arcs compose a set of
vertex-disjoint chains, each of which has length bounded by $L$ (The conditions
of the lemma are satisfied since constraint~(\ref{eq:picef_a}) implies
(\ref{eq:chain_flow})).

It remains to show that no selected cycle shares a vertex with a selected arc.
Suppose, to the contrary, that some selected cycle $c$ shares vertex $i \in P$
with a selected arc $a$.  Vertex $i$ cannot be the target of $a$, since constraint
(\ref{eq:picef_a}) would be violated if $i$ appears both in selected cycle $c$
and as the target of selected arc $a$. Hence $a=(i,j)$ for some $j \in P$. By
constraint (\ref{eq:picef_c}), $i$ must be the target of another selected arc,
$a'$. Therefore, $i$ appears in $c$ and is the target of $a'$, violating
constraint (\ref{eq:picef_a}).  \end{proof}

\begin{lemma} For any valid set of vertex-disjoint cycles and chains, there
is an assignment of values to the $y_{ijk}$ and $z_c$ respecting the PICEF
constraints.  \end{lemma}

\begin{proof}
This assignment can be constructed trivially.
\end{proof}

\begin{theorem}\label{thm:PICEFcorrect}
The PICEF model yields an optimal solution to the kidney exchange problem.  \end{theorem}

\subsection{Proofs for the LPR of PICEF}
\label{sec:PICEFLPRproofs}

\thmpicefvscf*

\begin{proof}
  $\zCF \preceq \zPICEF{}$.
  Consider an optimal
  solution to the LPR of the cycle formulation.  We show how to
  construct an equivalent (optimal) solution to the LPR of PICEF.
  For $c \in \cycles$, we transfer the value of $z_c$ directly from the cycle formulation
  solution to the PICEF solution. For each $(i,j) \in A$ and each $k \in \K'(i,j)$, let
\[
  	y_{ijk} = \sum_{(i,j) \text{ appears at position $k$ of } c} z_c.
\]
  
  This solution has the same objective value as the
  cycle formulation solution, and satisfies the constraints of the LPR of
  PICEF.

  $\zCF{} \prec \zPICEF{}$.
  Figure~\ref{fig:two-arms} shows a graph for which $\zPICEF{}$ is strictly greater (i.e., worse) than $\zCF{}$.
  Let $K=2$ and $L=4$.
  In this instance, $N = \{1\}$ and $P=\{2, \dots, 7\}$.

  In the cycle formulation, this instance has no admissible cycles, and the only
  admissible chains are
  $1 \rightarrow 2 \rightarrow 3 \rightarrow 4$,
  $1 \rightarrow 5 \rightarrow 6 \rightarrow 7$,
  and their prefixes. Since the longest chain has length 3 and any the sum of chain-variables
  containing vertex 1 may not exceed 1, we can see that the optimal objective value to the LPR
  of the cycle formulation is 3.

  We can achieve an objective value of $7/2$ to the LPR of PICEF, by letting
  $y_{121} = y_{232} = y_{343} = y_{151} =y_{562} = y_{673} = y_{754} = 1/2$.
\end{proof}

  \begin{figure}[ht!bp]
    \centering
    \begin{tikzpicture}[scale=0.8,transform shape]

  \node (alt1) at (0, 0) [whitesq] {$1$};

  \node (p1) at (1, .55) [whitecirc]    {$2$};
  \node (p2) at (2.2, .55) [whitecirc]  {$3$};
  \node (p3) at (3.4, .55) [whitecirc]    {$4$};
  \node (p4) at (1, -.55) [whitecirc]   {$5$};
  \node (p5) at (2.2, -.55) [whitecirc] {$6$};
  \node (p6) at (3.4, -.55) [whitecirc]   {$7$};

  \draw [-latex] (alt1) -- (p1);
  \draw [-latex] (p1) -- (p2);
  \draw [-latex] (p2) -- (p3);

  \draw [-latex] (alt1) -- (p4);
  \draw [-latex] (p4) -- (p5);
  \draw [-latex] (p5) -- (p6);
  \draw [-latex, bend left] (p6) to (p4);

\end{tikzpicture}
    \caption{A graph where $\zPICEF{}$ is strictly greater than $\zCF{}$.}\label{fig:two-arms}
  \end{figure}
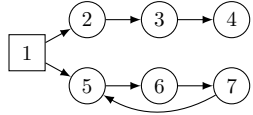


  \thmpicefvscflarge*
  
\begin{proof}
 We give a family of graphs, parameterised by $K$ and $L$, for which $\zPICEF{}$ is strictly greater than $\zCF{}$.  Given a constant cycle cap of $K$ and a chain cap of $L$ (which can effectively be infinite, if $L=|V|$), the graphs are constructed as follows.  For $i \in [L-K-1]$, create a cycle $\langle v^i_1, v^i_2, \ldots, v^i_{K+1} \rangle$ such that $v^{i+1}_1 = v^{i}_2$ for each $i \in [1, L-K-2]$; the cycle is otherwise disjoint from the rest of the graph.  Connect a single altruist $a$ to $v^1_1$; the altruist is otherwise disjoint from the rest of the graph.  Figure~\ref{fig:udders} visualizes the constructed graph.

  \begin{figure}[ht!bp]
    \centering
    \begin{tikzpicture}[scale=0.8,transform shape]

  \node (alt) at (0, 0) [whitesq] {$a$};

  \node (t0) at (1, 0) [whitecirc] {$v^1_1$};
  \node (t1) at (3, 0) [whitecirc] {$v^2_1$};
  \node (t2) at (5, 0) [whitecirc] {$v^3_1$};
  \node (t3) at (7, 0) [whitecirc] {$\ldots$};
  \node (tlm1) at (9, 0) [whitecirc] {};
  \node (tl) at (11, 0) [whitecirc] {};

  \node (u0_k) at ($(1.5,-1)$) [whitecirc] {};
  \node (u0_1) at ($(2.5,-1)$) [whitecirc] {};
  \node (u1_k) at ($(3.5,-1)$) [whitecirc] {};
  \node (u1_1) at ($(4.5,-1)$) [whitecirc] {};
  \node (u2_k) at ($(5.5,-1)$) [whitecirc] {};
  \node (u2_1) at ($(6.5,-1)$) [whitecirc] {};
  \node (ul_k) at ($(9.5,-1)$) [whitecirc] {};
  \node (ul_1) at ($(10.5,-1)$) [whitecirc] {};
  
  \draw [-latex] (alt.east) -- (t0.west);
  \draw [-latex] (t0.east) -- (t1.west);
  \draw [-latex] (t1.east) -- (t2.west);
  \draw [-latex] (t2.east) -- (t3.west);
  \draw [-latex,densely dotted] (t3.east) -- (tlm1.west);
  \draw [-latex] (tlm1.east) -- (tl.west);

  \draw [-latex] (u0_k.north west) -- (t0.south);
  \draw [-latex] (t1.south) -- (u0_1.north east);
  \draw [-latex,densely dotted] (u0_1) to [out=225,in=315] (u0_k);
  
  \draw [-latex] (u1_k.north west) -- (t1.south);
  \draw [-latex] (t2.south) -- (u1_1.north east);
  \draw [-latex,densely dotted] (u1_1) to [out=225,in=315] (u1_k);
    
  \draw [-latex] (u2_k.north west) -- (t2.south);
  \draw [-latex] (t3.south) -- (u2_1.north east);
  \draw [-latex,densely dotted] (u2_1) to [out=225,in=315] (u2_k);

  \draw [-latex] (ul_k.north west) -- (tlm1.south);
  \draw [-latex] (tl.south) -- (ul_1.north east);
  \draw [-latex,densely dotted] (ul_1) to [out=225,in=315] (ul_k);

  \draw [
    thick,
    decoration={
        brace,
        raise=0.35cm
    },
    decorate
  ] (t0.west) -- (tl.east)
  node [pos=0.5,anchor=north,yshift=1.2cm] {$L-K-1 \text{ patient-donor pairs}$};

  \draw [
    thick,
    decoration={
        brace,
        raise=1.35cm,
        mirror
    },
    decorate
  ] (t0.east) -- (t1.west)
  node [pos=0.5,anchor=north,yshift=-1.8cm] {$K+1$};

  \draw [
    thick,
    decoration={
        brace,
        raise=1.35cm,
        mirror
    },
    decorate
  ] (t1.east) -- (t2.west)
  node [pos=0.5,anchor=north,yshift=-1.8cm] {$K+1$};
  
  \draw [
    thick,
    decoration={
        brace,
        raise=1.35cm,
        mirror
    },
    decorate
  ] (t2.east) -- (t3.west)
  node [pos=0.5,anchor=north,yshift=-1.8cm] {$K+1$};
  
  \draw [
    thick,
    decoration={
        brace,
        raise=1.35cm,
        mirror
    },
    decorate
  ] (tlm1.east) -- (tl.west)
  node [pos=0.5,anchor=north,yshift=-1.8cm] {$K+1$};

\end{tikzpicture}
    \caption{Family of graphs where $\zPICEF{}$ is strictly looser than $\zCF{}$.}\label{fig:udders}
  \end{figure}
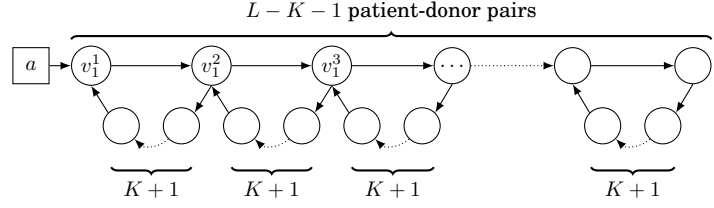

  The maximum cardinality disjoint packing of cycles of length at most $K$ and chains of length at most $L$ is the unique chain $(a, v^1_1, v^2_1, \ldots, v^{L-K-1}_1, v^{L-K-1}_2, \ldots, v^{L-K-1}_{K+1})$.  Thus, $\OPT{} = L$, where $\OPT{}$ is the optimal objective value to the integer program.  Indeed, there are no legal cycles of length at most $K$ in the graph, and at most one chain can be in any feasible solution due to the shared altruist $a$, so the (unique, by construction) longest chain is optimal.
  
  The LPR of the PICEF representation of this instance will assign weight of $1/2$ to each edge in the graph, for a total objective of $\zPICEF{} = 1/2 + (L+K-1)(\frac{K+1}{2})$.

  The LPR of the cycle formulation representation will create variables for each feasible cycle and chain in the graph.  There are no feasible cycles in the graph.  All chains in the graph share the edge $(a,v^1_1)$; thus all chains intersect and all chains contain $a$.   Thus, the sole binding constraint in the LPR of the cycle formulation is that the altruist node $a$ appears in at most one chain.  For chain decision variables $x_c \in [0,1]$, this problem can be rewritten as
  $$
  \max \sum_{c} |c|x_c \ \text{ subject to } \sum_{c} x_c \leq 1
  $$
  This constraint matrix is totally unimodular, and thus the LP optimum is integral (and is the IP optimum, or $\zCF{} = \OPT{} = L$).

The ratio of $\zPICEF$ to $\zCF$ is thus
\[
\frac{1}{2L} + \left(1 + \frac{K-1}{L}\right)\left(\frac{K+1}{2}\right)
\]

which can be made arbitrarily large by increasing $K$.
\end{proof}

\section{Hybrid Formulation}\label{sec:hybrid}
\subsection{Description of the HPIEF model}
In this section, we present a compact generalization of the PIEF model to kidney exchange graphs with non-directed donors.  This stands in contrast to the the PICEF formulation, which has polynomial counts of constraints and edge variables, but an exponential number of cycle variables. By replacing the cycle variables in PICEF with the variables from PIEF and modifying the constraints accordingly, we can create a compact formulation, the \emph{hybrid PIEF (HPIEF)}. Let the variables $x_{ijl}^l$ and the index set $\K(i,j,l)$ be defined as in PIEF. Let the variables $y_{ijk}$ and the index set $\K(i,j)$ be defined as in PICEF. The HPIEF integer program is as follows.

\begin{subequations}
\begin{align}
\max \quad
	\sum_{l\in P} \sum_{(i,j)\in A^l} \sum_{k\in \K(i,j,l)} w_{ij}x_{ijk}^l +
	\sum_{(i,j)\in A} \sum_{k\in \K(i,j)} w_{ij}y_{ijk}
	& \label{eq:hpief_obj}\\
\text{s.t.} \qquad
	\sum_{l\in P} \sum_{j:(j,i)\in A^l} \sum_{k\in \K(j,i,l)} x_{jik}^l +
	\sum_{j:(j,i)\in A} \sum_{k\in \K(j,i)} y_{jik}
	\leq 1 & \qquad i \in P
	\label{eq:hpief_a} \\
\text{Constraints (\ref{eq:pief_b}), (\ref{eq:picef_b}), and (\ref{eq:picef_c})}\nonumber \\
x_{ijk}^l \in \{0,1\} & \qquad
	\begin{aligned}
		l \in P,
		(i,j) \in A^l, \\
		k \in \K(i,j,l)
	\end{aligned}
	\label{eq:hpief_e} \\
y_{ijk} \in \{0,1\} & \qquad (i,j) \in A, k \in \K(i,j) \label{eq:hpief_f}
\end{align}
\end{subequations}

Inequalities~(\ref{eq:hpief_a}) and ~(\ref{eq:picef_b}) are the capacity constraints for
patients and NDDs respectively.

The reductions described in Subsections~\ref{sec:pief-basic-reduction} and \ref{sec:pief-redundancy} can
also be applied to the $x_{ijk}^l$ in HPIEF.


\subsection{Validity of the HPIEF model}

\begin{lemma}
Any assignment of values to the $x_{ijk}^l$ and $y_{ijk}$ that respects the
HPIEF constraints yields a vertex-disjoint set cycles of length no greater than
$K$ and chains no greater than $L$.
\end{lemma}

\begin{proof}
(Sketch.) Clearly, if the HPIEF constraints are satisfied then the PIEF
constraints (\ref{eq:pief_a}-\ref{eq:pief_c}) are satisfied also. Therefore, by
Theorem~\ref{thm:pief-validity-a}, the edges selected by the $x_{ijk}^l$ form a
vertex-disjoint set of cycles of length no greater than $K$.

By Lemma~\ref{lemma:picef_chains}, the selected arcs compose a set of
vertex-disjoint chains, each of which has length bounded by $L$.

It remains to show that the selected cycles and chains are vertex-disjoint.
This can be showed straightforwardly, along similar lines to the proof for
Theorem~\ref{thm:picef_validity_a}.
\end{proof}

\begin{lemma} For any valid set of vertex-disjoint cycles and chains, there
is an assignment of values to the $x_{ijk}^l$ and $y_{ijk}$ respecting the HPIEF
constraints.  \end{lemma}

\begin{proof}
This assignment can be constructed trivially.
\end{proof}

\begin{theorem}
The HPIEF model yields an optimal solution to the kidney exchange problem.
\end{theorem}

\subsection{LPR comparison of HPIEF and PICEF}

\begin{theorem}\label{thm:hpief-vs-picef}
  $\zHPIEF{} = \zPICEF{}$
\end{theorem}

The proof is similar to the proof for Theorem~\ref{thm:pief-vs-cf}, and is
therefore omitted.

\section{Additional Background and Proofs for Failure-Aware the PICEF Model}

In this section, we provide a proof of correctness of Algorithm~\ref{alg:failure-pricing}---which implements polynomial-time pricing of cycles for branch and price in the augmented failure-aware PICEF model---and discuss by way of counterexample why the basic deterministic polynomial-time cycle pricing algorithms of~\citeN{Glorie14:Kidney} and~\citeN{Plaut16:Fast} cannot be directly used for this case.

\subsection{Proof of Theorem~\ref{thm:failure-aware}}

\thmfailureaware*
\begin{proof}
Let $c = \langle v_1, v_2, \ldots, v_n\rangle$ be a discounted positive price cycle. Then $p^n\sum_{(i,j)\in c}  w_{ij}- \sum_{j\in c}\delta_j > 0$. Therefore 
$\sum_{j\in c} \delta_j - p^n\sum_{(i,j)\in c}w_{ij} < 0$. Then by definition of $w_k$, we have $\sum_{(i,j)\in c}(\delta_j - p^n w_{ij}) = \sum_{(i,j)\in c} w_n(i,j) < 0$.

This implies that $c$ is a negative cycle in $D$ on the $k = n$ iteration of Algorithm~\ref{alg:failure-pricing}. By the correctness of \textsc{GetNegativeCycles}, if there is a negative cycle in the graph, \textsc{GetNegativeCycles}$(D, n, w_n)$ will return at least one negative cycle of length at most $n$.

Let $c'$ be a returned cycle. Since $c'$ is negative in $D$ on the $k = n$ iteration, we have $\sum_{j\in c'} \delta_j - p^n\sum_{(i,j)\in c'}w_{ij} < 0$. Therefore $p^n\sum_{(i,j)\in c'}  w_{ij}- \sum_{j\in c'}\delta_j > 0$.

Since $|c'| \leq n$ by the correctness of \textsc{GetNegativeCycles}, we have $p^{|c'|} \geq p^n$. Because all arc weights in the original graph are nonnegative, $\sum_{(i,j)\in c'} w_{ij} \geq 0$. Therefore $p^{|c'|}\sum_{(i,j)\in c'} w_{ij} \geq p^n\sum_{(i,j)\in c'} w_{ij}$. Then $p^{|c'|}\sum_{(i,j)\in c'}  w_{ij}- \sum_{j\in c'}\delta_j \geq p^n\sum_{(i,j)\in c'}  w_{ij}- \sum_{j\in c'}\delta_j > 0$, so $c'$ is indeed discounted positive price.

Therefore Algorithm~\ref{alg:failure-pricing} returns at least one discounted positive price cycle.
\end{proof}

\subsection{Insufficiency of previous algorithms for the failure-aware pricing problem}\label{app:failure}

The pricing problem in the deterministic context, where post-match failures are not considered, is known to be solvable in polynomial time for cycles~\cite{Glorie14:Kidney,Plaut16:Fast} but not chains~\cite{Plaut16:Hardness}.  In Section~\ref{ssec:failure-pricing}, we presented Algorithm~\ref{alg:failure-pricing}, a polynomial-time algorithm for the pricing problem for cycles in the failure-aware context, for uniform success probability. In this section, we show how the basic algorithm for the deterministic pricing problem is not sufficient for the failure-aware context.

The algorithm for the deterministic setting initiates a Bellman-Ford style search to find negative cycles. Bellman-Ford is run $P$ times: on each iteration a different vertex representing a donor-patient pair is the source. After Bellman-Ford has been run from the source $s$ for $K - 1$ steps, suppose there is a path $\rho$ from $s$ to some vertex $v$ with weight $w(\rho)$, and there is an arc from $v$ back to $s$ with weight $w_{vs}$. If $w(\rho) + w_{vs} < 0$, then $\rho \cup (v,s)$ is a negative cycle \cite{Glorie14:Kidney,Plaut16:Fast}.

For consistency, in this section we discuss finding discounted \emph{negative} price cycles, which is trivially equivalent to finding discounted positive price cycles by reversing the signs on all arc weights and dual values. Therefore, we are looking for cycles $c$ satisfying $\sum_{j\in c} \delta_j - p^n\sum_{(i,j)\in c}w_{ij} < 0$.

Consider a straightforward modification to the algorithm from the deterministic setting, where each path now separately remembers its accumulated sum of dual values, sum of arc weights, and length. All of these can be easily recorded during the Bellman-Ford update step without altering the algorithm's complexity.

The issue arises when comparing paths. Figure~\ref{fig:failure1} gives an example of this. Consider running Bellman-Ford with $s$ as the source and $K=3$. The path $(s, v_2, v_3)$ is preferable to $(s, v_1, v_3)$, since we will end with the $3$-cycle $\langle s, v_2, v_3\rangle$ which has weight $p^3(\frac{-\eta}{p^3}) + \eta - 1 = -1$. However, suppose $K = 4$, and we removed the arc $(v_3, s)$. Then $\langle s, v_2, v_3\rangle$ is no longer a cycle, and the path $(s, v_1, v_3, v_4)$ will have weight $p^4(\frac{\eta}{p^3} - 1) + \eta -1 = \eta -p\eta - p^4 - 1 > 0$, assuming $\eta$ is large and $p$ is not close to $1$. However, the path $(s, v_1, v_3, v_4)$ would lead to a discounted negative cycle with weight $-p^4$. The algorithm from the deterministic setting cannot compare two paths without knowing the final cycle length.

\begin{figure}[ht!bp]
\centering
\resizebox{!}{2 in}{ 
\begin{tikzpicture}[->,>=stealth',shorten >=1pt,auto,node distance=2.5cm,
  thick,main node/.style={circle,fill=blue!20,minimum size=12mm,draw,font=\sffamily\Large\bfseries}]

  \node[main node] (1) {$s$};
  \node[main node] (2) [above right of=1] [label=above:{$0$}] {$v_1$};
  \node[main node] (3) [below right of=1] [label=above:{$\eta - 1$}] {$v_2$};
  \node [main node] (4) [below right of=2] [label=above:{$0$}] {$v_3$};
  \node[main node] (5) [above right of=2] [label=above:{$0$}] {$v_4$};

  \path[every node/.style={font=\sffamily\small}]
    (1) edge [left]  node[above] {$0$} (2)
         edge [right] node[below left] {$-\eta/p^3$} (3)
    (2)  edge[left] node[below] {$0$} (4)
    (3) edge [left] node[above] {$0$} (4) 
    (4) edge[left] node[above] {$0$} (1)
         edge[bend right] node[right] {$0$} (5)
    (5) edge [bend right=45] node[above left] {$-1$} (1)
      ;
\end{tikzpicture}
}
\caption{Example demonstrating that multiple possible final path lengths must be considered.}\label{fig:failure1}
\end{figure}
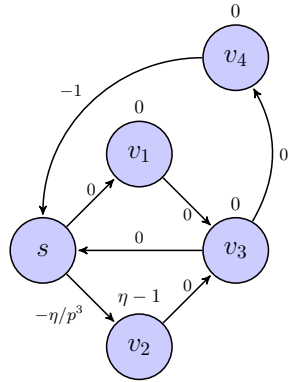

\section{Tabulated experimental results}\label{app:tables}
In this section, we restate the experimental results shown in Figures~\ref{fig:real-data} and~\ref{fig:increasing-v-increasing-a} of Section~\ref{sec:experiments} in the body of the paper, but now including statistics that were not possible to show in that figure.

\begin{landscape}
  {\small%

    \thispagestyle{plain}\begin{table}[ht!bp]
    \centering
    \caption{\\UNOS~match~runs.}
    \begin{tabular}
    {c r | r | r | r | r | r | r | r | r | r | r | r}
Method & & $L = 2$ & $L = 3$ & $L = 4$ & $L = 5$ & $L = 6$ & $L = 7$ & $L = 8$ & $L = 9$ & $L = 10$ & $L = 11$ & $L = 12$\\
\hline\\\\
 \textsc{PICEF} &  Mean & $0.45$& $0.46$& $0.51$& $0.58$& $0.63$& $0.65$& $0.71$& $0.79$& $0.85$& $0.93$& $0.98$\\
  &  Stdev & $0.06$& $0.10$& $0.12$& $0.17$& $0.25$& $0.28$& $0.32$& $0.44$& $0.48$& $0.63$& $0.62$\\
  &  Min & $0.31$& $0.31$& $0.31$& $0.31$& $0.31$& $0.31$& $0.31$& $0.31$& $0.31$& $0.31$& $0.31$\\
  &  Max & $0.62$& $0.97$& $1.02$& $1.42$& $1.72$& $2.12$& $2.22$& $3.12$& $3.52$& $4.82$& $3.62$\\
\hline\\\\
 \textsc{HPIEF} &  Mean & $0.98$& $1.11$& $1.06$& $1.22$& $1.17$& $1.27$& $1.38$& $1.37$& $1.44$& $1.50$& $1.56$\\
  &  Stdev & $0.39$& $0.42$& $0.44$& $0.49$& $0.53$& $0.52$& $0.59$& $0.55$& $0.61$& $0.67$& $0.73$\\
  &  Min & $0.31$& $0.31$& $0.31$& $0.31$& $0.31$& $0.31$& $0.31$& $0.31$& $0.31$& $0.31$& $0.31$\\
  &  Max & $1.87$& $1.97$& $2.37$& $2.57$& $3.27$& $2.62$& $3.58$& $3.07$& $3.42$& $4.07$& $4.32$\\
\hline\\\\
 \textsc{BnP-Poly} &  Mean & $0.18$& $0.19$& $0.19$& $0.20$& $0.21$& $0.25$& $0.47$& $0.73$& $1.37$& $0.35$& $2.99$\\
  &  Stdev & $0.06$& $0.07$& $0.08$& $0.09$& $0.10$& $0.56$& $3.78$& $8.25$& $11.60$& $1.52$& $45.28$\\
  &  Min & $0.03$& $0.03$& $0.03$& $0.03$& $0.03$& $0.03$& $0.03$& $0.03$& $0.03$& $0.03$& $0.03$\\
  &  Max & $0.42$& $0.42$& $0.57$& $0.57$& $0.67$& $9.33$& $63.22$& $138.46$& $154.74$& $24.06$& $759.13$\\
\hline\\\\
 \textsc{CG-TSP} &  Mean & $3.43$& $3.54$& $6.81$& $14.48$& $19.13$& $30.91$& $27.03$& $30.79$& $33.17$& $29.47$& $29.22$\\
  &  Stdev & $0.46$& $1.04$& $30.62$& $116.81$& $215.58$& $303.41$& $249.35$& $303.51$& $306.18$& $302.89$& $302.87$\\
  &  Min & $1.97$& $1.92$& $1.92$& $1.92$& $1.92$& $1.92$& $1.92$& $1.97$& $1.92$& $1.92$& $1.92$\\
  &  Max & $4.68$& $13.34$& $410.07$& $1401.87$& $3600.09$& $3600.10$& $3600.08$& $3600.08$& $3600.13$& $3600.08$& $3600.04$\\
\hline\\\\
 \textsc{BnP-DFS} &  Mean & $0.14$& $0.14$& $0.14$& $0.17$& $0.30$& $1.17$& $26.77$& $32.13$& $58.63$& $78.68$& $120.93$\\
  &  Stdev & $0.04$& $0.04$& $0.06$& $0.16$& $1.01$& $6.17$& $289.50$& $263.92$& $401.73$& $461.29$& $605.34$\\
  &  Min & $0.03$& $0.03$& $0.03$& $0.03$& $0.03$& $0.03$& $0.03$& $0.03$& $0.03$& $0.03$& $0.03$\\
  &  Max & $0.27$& $0.32$& $0.52$& $1.97$& $14.24$& $72.14$& $3600.00$& $3600.00$& $3600.00$& $3600.00$& $3600.00$\\
\hline\\\\
 \textsc{BnP-PICEF} &  Mean & $0.25$& $0.40$& $0.55$& $0.77$& $1.05$& $1.52$& $2.55$& $3.10$& $3.84$& $5.98$& $8.16$\\
  &  Stdev & $0.08$& $0.14$& $0.19$& $0.32$& $0.43$& $1.40$& $10.99$& $11.09$& $8.52$& $21.45$& $36.19$\\
  &  Min & $0.03$& $0.06$& $0.07$& $0.06$& $0.07$& $0.12$& $0.12$& $0.11$& $0.11$& $0.11$& $0.17$\\
  &  Max & $0.42$& $0.77$& $0.97$& $2.77$& $2.42$& $18.20$& $185.24$& $186.95$& $135.21$& $283.75$& $550.37$\\
\hline\\\\
\end{tabular}
    \label{tbl:unos-data}
    \end{table}

    \clearpage\thispagestyle{plain}\begin{table}[ht!bp]
    \centering
    \caption{\\NLDKSS~match~runs.}
    \begin{tabular}
    {c r | r | r | r | r | r | r | r | r | r | r | r}
Method & & $L = 2$ & $L = 3$ & $L = 4$ & $L = 5$ & $L = 6$ & $L = 7$ & $L = 8$ & $L = 9$ & $L = 10$ & $L = 11$ & $L = 12$\\
\hline\\\\
 \textsc{PICEF} &  Mean & $0.12$& $0.15$& $0.18$& $0.21$& $0.27$& $0.32$& $0.33$& $0.36$& $0.43$& $0.45$& $0.48$\\
  &  Stdev & $0.02$& $0.03$& $0.04$& $0.05$& $0.10$& $0.21$& $0.12$& $0.16$& $0.22$& $0.19$& $0.18$\\
  &  Min & $0.09$& $0.09$& $0.10$& $0.12$& $0.13$& $0.14$& $0.16$& $0.17$& $0.18$& $0.17$& $0.19$\\
  &  Max & $0.17$& $0.24$& $0.28$& $0.37$& $0.54$& $1.13$& $0.69$& $0.92$& $1.21$& $1.05$& $1.00$\\
\hline\\\\
 \textsc{HPIEF} &  Mean & $0.26$& $0.27$& $0.28$& $0.31$& $0.34$& $0.38$& $0.43$& $0.45$& $0.49$& $0.54$& $0.57$\\
  &  Stdev & $0.04$& $0.05$& $0.06$& $0.07$& $0.09$& $0.15$& $0.20$& $0.17$& $0.20$& $0.26$& $0.25$\\
  &  Min & $0.19$& $0.19$& $0.20$& $0.20$& $0.21$& $0.21$& $0.22$& $0.22$& $0.24$& $0.24$& $0.26$\\
  &  Max & $0.35$& $0.40$& $0.43$& $0.49$& $0.57$& $0.93$& $1.13$& $0.97$& $1.12$& $1.29$& $1.40$\\
\hline\\\\
 \textsc{BnP-Poly} &  Mean & $0.16$& $0.23$& $0.18$& $92.63$& $0.82$& $0.24$& $0.64$& $0.84$& $0.85$& $212.58$& $216.34$\\
  &  Stdev & $0.06$& $0.27$& $0.07$& $369.70$& $2.00$& $0.15$& $1.11$& $1.20$& $1.45$& $846.86$& $846.07$\\
  &  Min & $0.07$& $0.07$& $0.06$& $0.05$& $0.05$& $0.04$& $0.10$& $0.09$& $0.10$& $0.11$& $0.10$\\
  &  Max & $0.28$& $1.27$& $0.34$& $1571.45$& $8.57$& $0.72$& $4.85$& $4.25$& $5.10$& $3600.00$& $3600.00$\\
\hline\\\\
 \textsc{CG-TSP} &  Mean & $0.95$& $41.01$& $753.96$& $650.07$& $255.26$& $256.39$& $260.80$& $129.00$& $221.10$& $229.37$& $4.69$\\
  &  Stdev & $0.42$& $141.39$& $1349.06$& $1167.96$& $846.22$& $845.67$& $844.40$& $418.79$& $845.20$& $843.49$& $9.84$\\
  &  Min & $0.47$& $0.58$& $0.56$& $0.54$& $0.58$& $0.51$& $0.52$& $0.47$& $0.54$& $0.46$& $0.50$\\
  &  Max & $2.12$& $604.14$& $3600.00$& $3600.00$& $3600.00$& $3600.00$& $3600.00$& $1780.46$& $3600.00$& $3600.00$& $42.72$\\
\hline\\\\
 \textsc{BnP-DFS} &  Mean & $0.13$& $0.18$& $0.78$& $170.12$& $238.76$& $329.46$& $1090.72$& $1473.75$& $1888.86$& $2386.44$& $2506.19$\\
  &  Stdev & $0.08$& $0.16$& $1.37$& $653.91$& $841.62$& $842.37$& $1430.36$& $1540.36$& $1649.63$& $1619.09$& $1540.93$\\
  &  Min & $0.03$& $0.03$& $0.03$& $0.03$& $0.03$& $0.04$& $0.04$& $0.04$& $0.04$& $0.04$& $0.04$\\
  &  Max & $0.36$& $0.71$& $5.83$& $2784.96$& $3600.00$& $3600.00$& $3600.00$& $3600.00$& $3600.00$& $3600.00$& $3600.00$\\
\hline\\\\
 \textsc{BnP-PICEF} &  Mean & $0.19$& $0.23$& $0.30$& $233.49$& $53.32$& $0.83$& $215.66$& $5.25$& $3.18$& $4.89$& $19.68$\\
  &  Stdev & $0.12$& $0.08$& $0.14$& $845.94$& $210.51$& $0.76$& $846.14$& $9.95$& $5.49$& $10.34$& $57.00$\\
  &  Min & $0.09$& $0.09$& $0.11$& $0.23$& $0.27$& $0.33$& $0.41$& $0.47$& $0.41$& $0.53$& $0.65$\\
  &  Max & $0.61$& $0.42$& $0.71$& $3600.00$& $895.37$& $3.65$& $3600.00$& $40.09$& $23.65$& $43.03$& $245.17$\\
\hline\\\\
\end{tabular}
    \label{tbl:uk-data}
    \end{table}

    \clearpage\thispagestyle{plain}\begin{table}[ht!bp]
\centering
\caption{$|P|~=~300,~|A|~=~3$}
\begin{tabular}
{ c r | r  | r  | r  | r  | r  | r  | r  | r  | r  | r  | r }
Method & & $L = 2$ & $L = 3$ & $L = 4$ & $L = 5$ & $L = 6$ & $L = 7$ & $L = 8$ & $L = 9$ & $L = 10$ & $L = 11$ & $L = 12$\\
\hline\\\\
$\textsc{PICEF}$ & Mean & $1.00$& $1.40$& $2.04$& $2.89$& $4.06$& $5.63$& $7.73$& $10.99$& $14.52$& $17.87$& $20.18$\\
& Stdev & $0.12$& $0.32$& $0.51$& $0.77$& $0.92$& $1.28$& $2.23$& $3.13$& $4.54$& $5.90$& $6.89$\\
& Min& $0.77$& $0.94$& $1.27$& $1.87$& $2.69$& $3.60$& $5.05$& $6.30$& $6.78$& $9.33$& $10.28$\\
 & Max& $1.37$& $2.42$& $3.80$& $5.42$& $6.76$& $8.24$& $13.00$& $18.33$& $25.50$& $33.84$& $36.46$\\
\hline\\\\
$\textsc{HPIEF}$ & Mean & $3.04$& $3.47$& $4.16$& $5.14$& $6.52$& $8.12$& $10.40$& $13.67$& $16.74$& $20.59$& $22.29$\\
& Stdev & $0.51$& $0.79$& $0.85$& $1.27$& $1.98$& $2.26$& $2.88$& $3.67$& $3.44$& $5.33$& $6.07$\\
& Min& $2.44$& $2.52$& $2.87$& $3.40$& $4.35$& $5.17$& $6.96$& $8.83$& $12.16$& $9.93$& $10.11$\\
 & Max& $4.10$& $6.35$& $5.85$& $8.35$& $12.59$& $14.49$& $17.92$& $22.78$& $24.24$& $31.85$& $37.55$\\
\hline\\\\
$\textsc{BnP-Poly}$ & Mean & $2.08$& $81.96$& $238.82$& $306.05$& $631.69$& $1074.37$& $1000.78$& $2300.22$& $1328.40$& $2076.98$& $2515.60$\\
& Stdev & $5.30$& $353.34$& $581.46$& $654.90$& $928.03$& $931.64$& $1080.38$& $1060.62$& $1174.75$& $1072.71$& $1347.17$\\
& Min& $0.54$& $0.57$& $0.62$& $0.72$& $0.89$& $1.04$& $0.87$& $4.68$& $1.64$& $2.50$& $1.87$\\
 & Max& $26.56$& $1800.68$& $1801.42$& $1801.40$& $2890.70$& $3600.61$& $3600.72$& $3601.89$& $3600.86$& $3601.85$& $3602.00$\\
\hline\\\\
$\textsc{CG-TSP}$ & Mean & $8.44$& $184.96$& $918.97$& $1221.19$& $1709.03$& $1682.64$& $1632.00$& $1908.47$& $1839.85$& $1993.84$& $1903.64$\\
& Stdev & $1.04$& $403.94$& $913.85$& $1031.16$& $1316.14$& $1132.89$& $1075.43$& $1206.85$& $1128.78$& $1173.08$& $1062.16$\\
& Min& $6.48$& $9.07$& $12.71$& $12.07$& $18.49$& $22.77$& $35.26$& $131.48$& $86.87$& $21.13$& $28.67$\\
 & Max& $10.62$& $1804.94$& $2732.26$& $3600.07$& $3600.10$& $3600.09$& $3600.08$& $3600.11$& $3600.10$& $3600.10$& $3600.16$\\
\hline\\\\
$\textsc{BnP-DFS}$ & Mean & $14.83$& $17.13$& $383.40$& $694.57$& $2644.81$& $3494.87$& $3600.00$& $3600.00$& $3600.00$& $3600.00$& $3600.00$\\
& Stdev & $69.43$& $71.79$& $714.88$& $793.47$& $1026.37$& $285.67$& $0.00$& $0.00$& $0.00$& $0.00$& $0.00$\\
& Min& $0.49$& $0.59$& $1.34$& $14.21$& $284.94$& $2632.71$& $3600.00$& $3600.00$& $3600.00$& $3600.00$& $3600.00$\\
 & Max& $354.99$& $367.77$& $1816.74$& $2166.63$& $3600.00$& $3600.00$& $3600.00$& $3600.00$& $3600.00$& $3600.00$& $3600.00$\\
\hline\\\\
$\textsc{BnP-PICEF}$ & Mean & $2.39$& $14.77$& $110.11$& $210.89$& $635.14$& $1088.77$& $1561.96$& $2420.79$& $1693.39$& $2514.56$& $2791.20$\\
& Stdev & $5.18$& $63.60$& $330.96$& $444.70$& $827.57$& $869.08$& $1061.65$& $877.50$& $1312.65$& $1029.71$& $993.59$\\
& Min& $0.72$& $1.17$& $1.82$& $3.10$& $5.65$& $12.17$& $14.63$& $538.04$& $37.81$& $48.04$& $67.38$\\
 & Max& $24.74$& $326.27$& $1451.10$& $1639.32$& $3043.01$& $3088.56$& $3601.15$& $3600.78$& $3600.65$& $3600.52$& $3600.42$\\
\hline\\\\
\end{tabular}

\label{tbl:data}
\end{table}
    
    \clearpage\thispagestyle{plain}\begin{table}[ht!bp]
\centering
\caption{$|P|~=~300,~|A|~=~6$}
\begin{tabular}
{ c r | r  | r  | r  | r  | r  | r  | r  | r  | r  | r  | r }
Method & & $L = 2$ & $L = 3$ & $L = 4$ & $L = 5$ & $L = 6$ & $L = 7$ & $L = 8$ & $L = 9$ & $L = 10$ & $L = 11$ & $L = 12$\\
\hline\\\\
$\textsc{PICEF}$ & Mean & $1.45$& $2.14$& $2.96$& $4.36$& $5.86$& $8.03$& $10.02$& $12.69$& $17.82$& $23.35$& $26.19$\\
& Stdev & $0.53$& $0.92$& $1.04$& $1.88$& $2.60$& $3.31$& $4.48$& $6.06$& $10.81$& $14.57$& $17.75$\\
& Min& $0.92$& $1.24$& $1.89$& $2.49$& $3.45$& $4.47$& $4.95$& $6.45$& $8.45$& $9.83$& $12.49$\\
 & Max& $3.13$& $5.76$& $5.93$& $10.29$& $12.67$& $17.20$& $23.42$& $29.72$& $58.28$& $66.11$& $96.83$\\
\hline\\\\
$\textsc{HPIEF}$ & Mean & $3.22$& $3.78$& $4.84$& $6.44$& $8.78$& $13.51$& $19.45$& $25.49$& $53.11$& $58.12$& $55.49$\\
& Stdev & $0.56$& $0.67$& $1.17$& $2.13$& $3.38$& $8.70$& $14.74$& $12.80$& $43.89$& $44.44$& $29.12$\\
& Min& $2.49$& $2.92$& $3.42$& $4.00$& $4.80$& $5.22$& $6.60$& $6.90$& $11.34$& $11.09$& $16.52$\\
 & Max& $4.55$& $5.63$& $7.88$& $11.71$& $16.27$& $42.34$& $81.69$& $55.29$& $236.82$& $197.45$& $110.75$\\
\hline\\\\
$\textsc{BnP-Poly}$ & Mean & $73.01$& $222.84$& $815.13$& $1168.37$& $1902.71$& $2611.16$& $2513.36$& $2346.77$& $2470.76$& $2264.23$& $2140.75$\\
& Stdev & $352.77$& $583.10$& $1171.55$& $1205.56$& $1263.29$& $1166.98$& $1282.07$& $1173.08$& $1073.70$& $846.97$& $1011.70$\\
& Min& $0.57$& $0.64$& $0.74$& $0.89$& $1.09$& $2.14$& $2.44$& $12.42$& $25.30$& $264.09$& $94.33$\\
 & Max& $1801.19$& $1801.23$& $3602.15$& $3601.76$& $3601.45$& $3601.79$& $3601.75$& $3601.90$& $3601.85$& $3601.72$& $3600.67$\\
\hline\\\\
$\textsc{CG-TSP}$ & Mean & $63.67$& $1783.81$& $2825.67$& $3274.28$& $3469.97$& $3528.59$& $3507.68$& $3493.12$& $3332.71$& $3320.82$& $3103.94$\\
& Stdev & $249.48$& $1021.00$& $835.84$& $599.43$& $443.61$& $350.18$& $338.77$& $344.12$& $613.11$& $577.68$& $820.61$\\
& Min& $8.01$& $32.68$& $1437.28$& $1811.15$& $1810.83$& $1813.06$& $1989.48$& $2012.35$& $1823.48$& $1821.29$& $975.77$\\
 & Max& $1285.66$& $3600.08$& $3600.09$& $3600.10$& $3600.11$& $3600.09$& $3600.08$& $3600.09$& $3600.10$& $3600.09$& $3600.11$\\
\hline\\\\
$\textsc{BnP-DFS}$ & Mean & $2.34$& $440.67$& $974.91$& $2324.52$& $3540.79$& $3600.00$& $3600.00$& $3600.00$& $3600.00$& $3600.00$& $3600.00$\\
& Stdev & $4.92$& $765.51$& $1247.14$& $1097.18$& $290.09$& $0.00$& $0.00$& $0.00$& $0.00$& $0.00$& $0.00$\\
& Min& $0.54$& $1.09$& $12.73$& $349.24$& $2119.63$& $3600.00$& $3600.00$& $3600.00$& $3600.00$& $3600.00$& $3600.00$\\
 & Max& $22.03$& $1805.45$& $3600.25$& $3600.02$& $3600.01$& $3600.00$& $3600.00$& $3600.00$& $3600.00$& $3600.00$& $3600.00$\\
\hline\\\\
$\textsc{BnP-PICEF}$ & Mean & $2.09$& $124.04$& $871.35$& $1533.14$& $1902.03$& $2567.05$& $2786.84$& $2403.03$& $2456.11$& $2041.20$& $1273.52$\\
& Stdev & $2.80$& $380.96$& $1139.42$& $1238.55$& $1272.68$& $1220.26$& $1066.82$& $1193.46$& $1212.64$& $952.01$& $1046.65$\\
& Min& $0.79$& $1.19$& $1.97$& $3.65$& $7.76$& $14.43$& $21.98$& $43.21$& $78.77$& $112.14$& $168.46$\\
 & Max& $12.41$& $1600.58$& $3262.90$& $3580.66$& $3602.58$& $3601.76$& $3600.80$& $3600.62$& $3600.44$& $3600.37$& $3600.24$\\
\hline\\\\
\end{tabular}

\label{tbl:data}
\end{table}

    \clearpage\thispagestyle{plain}\begin{table}[ht!bp]
\centering
\caption{$|P|~=~300,~|A|~=~15$}
\begin{tabular}
{ c r | r  | r  | r  | r  | r  | r  | r  | r  | r  | r  | r }
Method & & $L = 2$ & $L = 3$ & $L = 4$ & $L = 5$ & $L = 6$ & $L = 7$ & $L = 8$ & $L = 9$ & $L = 10$ & $L = 11$ & $L = 12$\\
\hline\\\\
$\textsc{PICEF}$ & Mean & $1.36$& $2.00$& $2.93$& $5.06$& $7.44$& $14.52$& $17.20$& $21.79$& $24.90$& $31.46$& $29.79$\\
& Stdev & $0.44$& $0.64$& $0.72$& $2.27$& $3.94$& $20.99$& $13.30$& $7.76$& $8.87$& $13.39$& $6.64$\\
& Min& $1.04$& $1.49$& $2.17$& $2.79$& $3.87$& $4.50$& $5.75$& $9.18$& $11.01$& $15.21$& $18.42$\\
 & Max& $3.26$& $4.15$& $5.32$& $12.68$& $21.23$& $114.52$& $77.98$& $38.47$& $58.61$& $85.77$& $47.16$\\
\hline\\\\
$\textsc{HPIEF}$ & Mean & $3.11$& $3.85$& $4.72$& $6.33$& $10.18$& $13.13$& $18.62$& $24.36$& $29.43$& $31.45$& $32.13$\\
& Stdev & $1.02$& $1.57$& $2.09$& $3.26$& $5.96$& $8.86$& $13.97$& $16.56$& $16.73$& $17.77$& $17.75$\\
& Min& $1.32$& $1.32$& $1.32$& $1.32$& $1.32$& $1.32$& $1.32$& $1.32$& $1.32$& $1.32$& $1.32$\\
 & Max& $4.85$& $7.20$& $9.52$& $13.59$& $20.95$& $35.23$& $68.89$& $73.93$& $67.31$& $58.72$& $66.14$\\
\hline\\\\
$\textsc{BnP-Poly}$ & Mean & $74.83$& $611.66$& $991.62$& $1663.90$& $1419.18$& $632.51$& $283.39$& $366.94$& $475.84$& $1038.21$& $1230.37$\\
& Stdev & $353.01$& $867.89$& $872.00$& $1390.32$& $1210.69$& $837.82$& $770.75$& $616.38$& $623.12$& $1099.15$& $1025.71$\\
& Min& $0.74$& $0.89$& $1.07$& $1.57$& $0.39$& $0.37$& $0.37$& $0.37$& $0.37$& $0.37$& $0.37$\\
 & Max& $1804.07$& $2318.25$& $1805.02$& $3601.10$& $3600.77$& $1817.17$& $3600.05$& $2021.28$& $2109.74$& $3600.14$& $3600.01$\\
\hline\\\\
$\textsc{CG-TSP}$ & Mean & $367.16$& $3515.05$& $3600.09$& $3600.09$& $3535.69$& $3529.94$& $3536.72$& $3430.70$& $3461.73$& $3279.96$& $3220.04$\\
& Stdev & $640.14$& $346.51$& $0.04$& $0.03$& $315.51$& $343.66$& $310.45$& $462.86$& $469.31$& $642.91$& $613.91$\\
& Min& $16.67$& $1857.78$& $3600.06$& $3600.06$& $1990.03$& $1846.38$& $2015.84$& $1948.32$& $1833.44$& $1838.52$& $1968.78$\\
 & Max& $1822.65$& $3600.17$& $3600.20$& $3600.16$& $3600.25$& $3600.21$& $3600.20$& $3600.18$& $3600.15$& $3600.24$& $3600.16$\\
\hline\\\\
$\textsc{BnP-DFS}$ & Mean & $141.93$& $683.49$& $1274.19$& $3401.45$& $3600.00$& $3600.00$& $3600.00$& $3600.00$& $3600.00$& $3600.00$& $3600.00$\\
& Stdev & $476.11$& $850.91$& $966.39$& $410.86$& $0.00$& $0.00$& $0.00$& $0.00$& $0.00$& $0.00$& $0.00$\\
& Min& $0.74$& $3.95$& $81.43$& $2091.66$& $3600.00$& $3600.00$& $3600.00$& $3600.00$& $3600.00$& $3600.00$& $3600.00$\\
 & Max& $1801.39$& $1842.76$& $3600.22$& $3600.04$& $3600.00$& $3600.00$& $3600.00$& $3600.00$& $3600.00$& $3600.00$& $3600.00$\\
\hline\\\\
$\textsc{BnP-PICEF}$ & Mean & $149.79$& $445.86$& $856.61$& $1712.61$& $1095.22$& $535.43$& $201.36$& $178.40$& $162.61$& $159.80$& $218.11$\\
& Stdev & $487.53$& $754.76$& $854.16$& $1192.86$& $1126.06$& $803.12$& $479.10$& $343.89$& $83.78$& $70.62$& $111.01$\\
& Min& $0.87$& $1.42$& $2.62$& $4.75$& $9.64$& $18.01$& $28.09$& $48.20$& $70.51$& $53.85$& $79.26$\\
 & Max& $1802.03$& $1803.15$& $1830.94$& $3602.67$& $3601.98$& $1867.16$& $1832.55$& $1843.38$& $404.82$& $339.64$& $523.09$\\
\hline\\\\
\end{tabular}

\label{tbl:data}
\end{table}

    \clearpage\thispagestyle{plain}\begin{table}[ht!bp]
\centering
\caption{$|P|~=~300,~|A|~=~75$}
\begin{tabular}
{ c r | r  | r  | r  | r  | r  | r  | r  | r  | r  | r  | r }
Method & & $L = 2$ & $L = 3$ & $L = 4$ & $L = 5$ & $L = 6$ & $L = 7$ & $L = 8$ & $L = 9$ & $L = 10$ & $L = 11$ & $L = 12$\\
\hline\\\\
$\textsc{PICEF}$ & Mean & $2.06$& $2.97$& $4.03$& $5.42$& $6.96$& $8.90$& $10.38$& $12.41$& $13.75$& $15.14$& $16.95$\\
& Stdev & $0.70$& $1.13$& $1.18$& $1.91$& $2.41$& $3.78$& $2.70$& $2.53$& $2.86$& $2.31$& $3.40$\\
& Min& $1.37$& $1.97$& $2.64$& $3.54$& $4.32$& $5.30$& $6.93$& $7.75$& $8.55$& $10.36$& $10.34$\\
 & Max& $4.47$& $7.43$& $8.08$& $13.19$& $16.37$& $25.68$& $17.05$& $18.69$& $23.80$& $19.35$& $26.39$\\
\hline\\\\
$\textsc{HPIEF}$ & Mean & $3.85$& $4.70$& $5.88$& $7.38$& $8.81$& $11.55$& $13.47$& $16.08$& $18.44$& $19.86$& $20.05$\\
& Stdev & $1.32$& $1.38$& $1.74$& $2.46$& $2.75$& $3.77$& $3.80$& $4.08$& $4.46$& $4.44$& $4.78$\\
& Min& $2.72$& $3.17$& $3.80$& $4.50$& $5.15$& $6.00$& $7.18$& $7.78$& $9.63$& $11.01$& $12.66$\\
 & Max& $9.64$& $9.98$& $12.42$& $16.14$& $17.15$& $21.58$& $20.98$& $22.90$& $26.86$& $28.19$& $31.60$\\
\hline\\\\
$\textsc{BnP-Poly}$ & Mean & $2.61$& $3.36$& $4.41$& $4.91$& $5.42$& $6.97$& $9.74$& $20.70$& $67.07$& $155.99$& $223.71$\\
& Stdev & $0.19$& $0.34$& $0.67$& $0.94$& $1.03$& $2.52$& $6.76$& $31.39$& $94.52$& $379.20$& $426.11$\\
& Min& $2.17$& $2.80$& $3.42$& $3.63$& $3.87$& $3.55$& $3.80$& $4.48$& $6.88$& $5.25$& $8.64$\\
 & Max& $3.02$& $4.08$& $5.53$& $7.36$& $8.21$& $14.15$& $33.77$& $166.95$& $398.11$& $1812.22$& $1806.62$\\
\hline\\\\
$\textsc{CG-TSP}$ & Mean & $3168.00$& $3600.18$& $3544.38$& $3451.82$& $3110.08$& $3056.40$& $2575.13$& $2339.98$& $2222.51$& $1524.77$& $1122.56$\\
& Stdev & $828.44$& $0.07$& $191.05$& $409.89$& $539.91$& $560.83$& $671.91$& $799.95$& $929.42$& $1030.21$& $881.64$\\
& Min& $523.02$& $3600.11$& $2812.72$& $2097.94$& $2237.32$& $2001.36$& $1339.74$& $238.59$& $410.76$& $22.79$& $20.75$\\
 & Max& $3600.25$& $3600.36$& $3600.38$& $3600.43$& $3600.41$& $3600.33$& $3600.61$& $3600.32$& $3600.41$& $3600.16$& $2753.27$\\
\hline\\\\
$\textsc{BnP-DFS}$ & Mean & $4.66$& $78.58$& $1984.70$& $3600.01$& $3600.00$& $3600.00$& $3600.00$& $3600.00$& $3600.00$& $3600.00$& $3600.00$\\
& Stdev & $0.99$& $22.89$& $748.50$& $0.00$& $0.00$& $0.00$& $0.00$& $0.00$& $0.00$& $0.00$& $0.00$\\
& Min& $3.15$& $34.57$& $776.95$& $3600.00$& $3600.00$& $3600.00$& $3600.00$& $3600.00$& $3600.00$& $3600.00$& $3600.00$\\
 & Max& $6.70$& $122.86$& $3143.17$& $3600.02$& $3600.01$& $3600.00$& $3600.01$& $3600.01$& $3600.01$& $3600.01$& $3600.00$\\
\hline\\\\
$\textsc{BnP-PICEF}$ & Mean & $1.76$& $2.12$& $3.51$& $5.61$& $7.96$& $12.27$& $14.82$& $18.70$& $21.42$& $27.41$& $33.25$\\
& Stdev & $0.11$& $0.27$& $1.19$& $2.12$& $2.73$& $3.70$& $4.64$& $4.29$& $5.82$& $8.47$& $10.15$\\
& Min& $1.52$& $1.69$& $1.82$& $3.05$& $2.97$& $6.81$& $7.46$& $10.26$& $11.22$& $13.93$& $17.36$\\
 & Max& $2.05$& $3.02$& $5.96$& $10.74$& $16.08$& $19.74$& $27.74$& $28.28$& $37.68$& $43.10$& $55.48$\\
\hline\\\\
\end{tabular}

\label{tbl:data}
\end{table}

    \clearpage\thispagestyle{plain}\begin{table}[ht!bp]
\centering
\caption{$|P|~=~500,~|A|~=~5$}
\begin{tabular}
{ c r | r  | r  | r  | r  | r  | r  | r  | r  | r  | r  | r }
Method & & $L = 2$ & $L = 3$ & $L = 4$ & $L = 5$ & $L = 6$ & $L = 7$ & $L = 8$ & $L = 9$ & $L = 10$ & $L = 11$ & $L = 12$\\
\hline\\\\
$\textsc{PICEF}$ & Mean & $3.33$& $5.69$& $9.68$& $16.06$& $29.55$& $49.75$& $79.13$& $101.60$& $161.25$& $324.80$& $378.28$\\
& Stdev & $0.60$& $1.25$& $2.73$& $4.38$& $10.00$& $16.02$& $23.53$& $32.33$& $78.74$& $309.10$& $405.84$\\
& Min& $2.67$& $4.17$& $6.35$& $9.61$& $12.19$& $17.20$& $21.20$& $24.36$& $36.23$& $41.63$& $50.26$\\
 & Max& $4.92$& $8.58$& $17.21$& $24.63$& $50.08$& $70.96$& $126.06$& $160.64$& $345.49$& $1333.83$& $1864.47$\\
\hline\\\\
$\textsc{HPIEF}$ & Mean & $13.14$& $15.33$& $17.49$& $20.30$& $24.34$& $30.93$& $40.59$& $50.92$& $102.03$& $187.22$& $246.18$\\
& Stdev & $2.89$& $3.46$& $3.24$& $3.55$& $3.62$& $6.46$& $10.33$& $19.93$& $88.83$& $175.34$& $310.59$\\
& Min& $9.78$& $11.01$& $12.53$& $14.49$& $17.72$& $19.15$& $22.78$& $30.32$& $42.70$& $34.58$& $36.96$\\
 & Max& $23.73$& $24.60$& $24.29$& $26.37$& $30.25$& $49.24$& $73.30$& $135.55$& $475.85$& $831.95$& $1494.91$\\
\hline\\\\
$\textsc{BnP-Poly}$ & Mean & $476.19$& $398.72$& $987.34$& $1995.51$& $1882.08$& $2376.94$& $2595.14$& $3027.64$& $3313.34$& $3529.60$& $3600.32$\\
& Stdev & $924.35$& $875.52$& $1015.36$& $1191.54$& $1290.49$& $1216.20$& $1025.58$& $974.20$& $657.83$& $346.79$& $0.08$\\
& Min& $2.07$& $2.24$& $2.52$& $3.82$& $4.10$& $5.00$& $10.44$& $53.25$& $1804.95$& $1830.66$& $3600.12$\\
 & Max& $3600.11$& $3600.10$& $3600.21$& $3600.42$& $3600.36$& $3600.61$& $3600.61$& $3600.67$& $3600.57$& $3600.50$& $3600.45$\\
\hline\\\\
$\textsc{CG-TSP}$ & Mean & $86.08$& $1519.67$& $2888.52$& $3307.18$& $3500.92$& $3542.08$& $3600.09$& $3600.09$& $3600.09$& $3600.09$& $3600.10$\\
& Stdev & $163.39$& $1194.67$& $827.72$& $675.52$& $356.92$& $267.64$& $0.01$& $0.01$& $0.01$& $0.01$& $0.02$\\
& Min& $17.80$& $43.13$& $1528.27$& $1445.40$& $1937.79$& $2233.38$& $3600.07$& $3600.07$& $3600.07$& $3600.07$& $3600.07$\\
 & Max& $870.78$& $3600.22$& $3600.18$& $3600.18$& $3600.15$& $3600.13$& $3600.13$& $3600.13$& $3600.14$& $3600.14$& $3600.16$\\
\hline\\\\
$\textsc{BnP-DFS}$ & Mean & $584.00$& $609.09$& $1578.12$& $3595.49$& $3600.00$& $3600.00$& $3600.00$& $3600.00$& $3600.00$& $3600.00$& $3600.00$\\
& Stdev & $978.95$& $969.92$& $962.98$& $22.13$& $0.00$& $0.00$& $0.00$& $0.00$& $0.00$& $0.00$& $0.00$\\
& Min& $3.02$& $8.48$& $220.63$& $3487.09$& $3600.00$& $3600.00$& $3600.00$& $3600.00$& $3600.00$& $3600.00$& $3600.00$\\
 & Max& $3600.25$& $3600.31$& $3600.07$& $3600.01$& $3600.00$& $3600.00$& $3600.00$& $3600.00$& $3600.00$& $3600.00$& $3600.00$\\
\hline\\\\
$\textsc{BnP-PICEF}$ & Mean & $489.87$& $408.98$& $1033.60$& $1748.72$& $1900.07$& $2549.32$& $2856.84$& $3275.68$& $3360.62$& $3546.67$& $3600.09$\\
& Stdev & $915.90$& $712.01$& $1148.29$& $1281.58$& $1060.35$& $1105.73$& $841.48$& $652.43$& $561.54$& $261.92$& $0.04$\\
& Min& $3.05$& $4.65$& $8.98$& $22.14$& $47.29$& $98.88$& $1832.39$& $1868.01$& $1892.93$& $2263.54$& $3600.02$\\
 & Max& $3601.12$& $1805.88$& $3601.60$& $3602.22$& $3600.93$& $3600.99$& $3600.45$& $3600.39$& $3600.36$& $3600.21$& $3600.16$\\
\hline\\\\
\end{tabular}

\label{tbl:data}
\end{table}

    \clearpage\thispagestyle{plain}\begin{table}[ht!bp]
\centering
\caption{$|P|~=~500,~|A|~=~10$}
\begin{tabular}
{ c r | r  | r  | r  | r  | r  | r  | r  | r  | r  | r  | r }
Method & & $L = 2$ & $L = 3$ & $L = 4$ & $L = 5$ & $L = 6$ & $L = 7$ & $L = 8$ & $L = 9$ & $L = 10$ & $L = 11$ & $L = 12$\\
\hline\\\\
$\textsc{PICEF}$ & Mean & $3.87$& $6.41$& $10.22$& $16.54$& $26.81$& $54.10$& $111.83$& $321.99$& $286.20$& $392.19$& $323.20$\\
& Stdev & $0.92$& $2.03$& $3.55$& $6.36$& $10.77$& $39.64$& $124.63$& $297.04$& $224.94$& $280.57$& $198.14$\\
& Min& $2.72$& $4.12$& $6.35$& $9.33$& $12.19$& $15.67$& $19.18$& $26.11$& $33.20$& $67.86$& $70.61$\\
 & Max& $6.20$& $13.39$& $21.53$& $33.38$& $55.39$& $190.56$& $518.68$& $1142.68$& $879.51$& $1182.30$& $927.75$\\
\hline\\\\
$\textsc{HPIEF}$ & Mean & $12.79$& $15.37$& $19.41$& $29.69$& $45.92$& $94.99$& $242.75$& $434.08$& $484.99$& $884.31$& $568.17$\\
& Stdev & $1.42$& $2.33$& $3.64$& $7.00$& $18.54$& $72.01$& $268.11$& $309.07$& $327.06$& $695.32$& $469.58$\\
& Min& $11.04$& $12.04$& $14.09$& $16.02$& $18.48$& $21.61$& $28.04$& $31.14$& $36.64$& $127.14$& $66.01$\\
 & Max& $16.11$& $20.78$& $26.40$& $41.07$& $113.75$& $374.04$& $1250.39$& $1096.22$& $1167.39$& $2230.86$& $2155.83$\\
\hline\\\\
$\textsc{BnP-Poly}$ & Mean & $600.79$& $726.93$& $1514.72$& $2378.45$& $2816.66$& $3528.67$& $3034.58$& $3051.69$& $2965.69$& $3229.29$& $3437.41$\\
& Stdev & $969.67$& $877.81$& $1099.85$& $1215.75$& $1132.17$& $351.68$& $1092.85$& $943.16$& $868.24$& $611.16$& $443.14$\\
& Min& $2.27$& $2.90$& $3.67$& $4.28$& $9.79$& $1805.82$& $15.00$& $76.72$& $398.26$& $1948.58$& $2135.48$\\
 & Max& $3600.17$& $1802.41$& $3600.26$& $3600.46$& $3600.59$& $3600.71$& $3600.59$& $3600.52$& $3600.37$& $3600.34$& $3600.25$\\
\hline\\\\
$\textsc{CG-TSP}$ & Mean & $406.15$& $3116.63$& $3600.10$& $3600.11$& $3600.11$& $3600.11$& $3600.10$& $3600.11$& $3600.11$& $3600.11$& $3600.12$\\
& Stdev & $630.53$& $844.15$& $0.02$& $0.02$& $0.02$& $0.04$& $0.01$& $0.03$& $0.02$& $0.04$& $0.07$\\
& Min& $53.28$& $517.29$& $3600.07$& $3600.07$& $3600.07$& $3600.08$& $3600.08$& $3600.08$& $3600.07$& $3600.08$& $3600.07$\\
 & Max& $1839.96$& $3600.18$& $3600.17$& $3600.16$& $3600.17$& $3600.28$& $3600.13$& $3600.24$& $3600.16$& $3600.29$& $3600.40$\\
\hline\\\\
$\textsc{BnP-DFS}$ & Mean & $727.04$& $1249.56$& $2754.81$& $3600.00$& $3600.00$& $3600.00$& $3600.00$& $3600.00$& $3600.00$& $3600.00$& $3600.00$\\
& Stdev & $992.61$& $827.81$& $682.38$& $0.00$& $0.00$& $0.00$& $0.00$& $0.00$& $0.00$& $0.00$& $0.00$\\
& Min& $2.90$& $19.27$& $1571.31$& $3600.00$& $3600.00$& $3600.00$& $3600.00$& $3600.00$& $3600.00$& $3600.00$& $3600.00$\\
 & Max& $3438.77$& $1826.38$& $3600.07$& $3600.00$& $3600.00$& $3600.00$& $3600.00$& $3600.00$& $3600.00$& $3600.00$& $3600.00$\\
\hline\\\\
$\textsc{BnP-PICEF}$ & Mean & $829.37$& $873.28$& $1593.51$& $2389.64$& $2910.67$& $3534.89$& $3056.63$& $2928.76$& $2817.08$& $2619.23$& $2645.57$\\
& Stdev & $1340.08$& $1029.25$& $1054.00$& $1206.53$& $1075.62$& $320.66$& $1044.62$& $932.67$& $988.09$& $763.68$& $660.52$\\
& Min& $3.20$& $6.28$& $11.97$& $26.78$& $90.19$& $1963.99$& $177.12$& $520.29$& $820.70$& $1037.74$& $1193.26$\\
 & Max& $3601.57$& $3601.87$& $3601.30$& $3601.08$& $3600.76$& $3600.50$& $3600.32$& $3600.23$& $3600.18$& $3600.05$& $3600.04$\\
\hline\\\\
\end{tabular}

\label{tbl:data}
\end{table}

    \clearpage\thispagestyle{plain}\begin{table}[ht!bp]
\centering
\caption{$|P|~=~500,~|A|~=~25$}
\begin{tabular}
{ c r | r  | r  | r  | r  | r  | r  | r  | r  | r  | r  | r }
Method & & $L = 2$ & $L = 3$ & $L = 4$ & $L = 5$ & $L = 6$ & $L = 7$ & $L = 8$ & $L = 9$ & $L = 10$ & $L = 11$ & $L = 12$\\
\hline\\\\
$\textsc{PICEF}$ & Mean & $5.13$& $8.81$& $14.79$& $31.28$& $55.40$& $91.01$& $159.51$& $210.90$& $263.30$& $326.07$& $321.93$\\
& Stdev & $3.10$& $4.88$& $7.90$& $16.29$& $31.14$& $40.02$& $74.17$& $106.60$& $110.04$& $128.59$& $115.14$\\
& Min& $3.22$& $5.65$& $8.10$& $14.42$& $18.82$& $28.81$& $48.51$& $51.74$& $74.37$& $81.94$& $90.87$\\
 & Max& $18.24$& $23.86$& $40.36$& $69.45$& $130.15$& $172.46$& $319.40$& $494.03$& $520.89$& $520.57$& $530.59$\\
\hline\\\\
$\textsc{HPIEF}$ & Mean & $13.64$& $16.70$& $25.11$& $50.19$& $64.58$& $88.26$& $125.61$& $171.61$& $223.86$& $231.21$& $207.82$\\
& Stdev & $2.96$& $3.36$& $7.83$& $46.74$& $25.11$& $33.30$& $57.05$& $64.98$& $75.50$& $76.25$& $67.91$\\
& Min& $9.86$& $11.91$& $14.81$& $19.35$& $29.36$& $42.41$& $47.14$& $56.01$& $67.62$& $90.86$& $105.80$\\
 & Max& $21.23$& $24.67$& $44.65$& $229.70$& $122.98$& $155.02$& $248.62$& $286.95$& $357.62$& $369.33$& $312.53$\\
\hline\\\\
$\textsc{BnP-Poly}$ & Mean & $610.17$& $939.96$& $1883.09$& $1242.21$& $621.96$& $428.89$& $618.50$& $896.95$& $1988.17$& $2418.15$& $2784.83$\\
& Stdev & $823.88$& $1032.01$& $1470.65$& $1100.86$& $957.44$& $721.69$& $879.42$& $957.15$& $1234.27$& $1160.63$& $942.70$\\
& Min& $3.12$& $4.00$& $6.65$& $9.96$& $13.62$& $14.65$& $21.21$& $31.33$& $39.06$& $68.33$& $151.56$\\
 & Max& $1802.71$& $3600.11$& $3600.38$& $3600.47$& $3600.23$& $2307.41$& $3600.21$& $3600.13$& $3600.09$& $3600.27$& $3600.16$\\
\hline\\\\
$\textsc{CG-TSP}$ & Mean & $3148.56$& $3600.21$& $3600.23$& $3600.20$& $3600.21$& $3600.22$& $3600.21$& $3600.21$& $3600.20$& $3600.24$& $3600.31$\\
& Stdev & $728.91$& $0.14$& $0.19$& $0.14$& $0.13$& $0.16$& $0.15$& $0.16$& $0.13$& $0.17$& $0.18$\\
& Min& $1009.67$& $3600.10$& $3600.11$& $3600.10$& $3600.12$& $3600.12$& $3600.11$& $3600.12$& $3600.11$& $3600.09$& $3600.04$\\
 & Max& $3600.17$& $3600.64$& $3600.85$& $3600.73$& $3600.63$& $3600.66$& $3600.66$& $3600.68$& $3600.59$& $3600.75$& $3600.67$\\
\hline\\\\
$\textsc{BnP-DFS}$ & Mean & $796.98$& $1108.07$& $3599.33$& $3600.00$& $3600.00$& $3600.00$& $3600.00$& $3600.00$& $3600.00$& $3600.00$& $3600.00$\\
& Stdev & $890.74$& $975.82$& $3.35$& $0.00$& $0.00$& $0.00$& $0.00$& $0.00$& $0.00$& $0.00$& $0.00$\\
& Min& $4.40$& $78.78$& $3582.91$& $3600.00$& $3600.00$& $3600.00$& $3600.00$& $3600.00$& $3600.00$& $3600.00$& $3600.00$\\
 & Max& $1811.33$& $3600.28$& $3600.03$& $3600.01$& $3600.00$& $3600.00$& $3600.00$& $3600.00$& $3600.00$& $3600.00$& $3600.00$\\
\hline\\\\
$\textsc{BnP-PICEF}$ & Mean & $513.65$& $881.36$& $1955.77$& $1417.19$& $557.22$& $454.23$& $677.44$& $1179.01$& $1496.40$& $1578.70$& $1627.19$\\
& Stdev & $804.07$& $1023.76$& $1423.30$& $1143.07$& $686.23$& $201.31$& $347.69$& $493.79$& $579.61$& $534.03$& $515.77$\\
& Min& $3.15$& $6.13$& $20.00$& $47.80$& $105.26$& $207.65$& $292.64$& $654.54$& $583.22$& $720.50$& $798.70$\\
 & Max& $1803.61$& $3601.50$& $3601.39$& $3600.55$& $2011.26$& $1138.30$& $2126.49$& $2554.51$& $2904.98$& $2730.44$& $2797.31$\\
\hline\\\\
\end{tabular}

\label{tbl:data}
\end{table}

    \clearpage\thispagestyle{plain}\begin{table}[ht!bp]
\centering
\caption{$|P|~=~500,~|A|~=~125$}
\begin{tabular}
{ c r | r  | r  | r  | r  | r  | r  | r  | r  | r  | r  | r }
Method & & $L = 2$ & $L = 3$ & $L = 4$ & $L = 5$ & $L = 6$ & $L = 7$ & $L = 8$ & $L = 9$ & $L = 10$ & $L = 11$ & $L = 12$\\
\hline\\\\
$\textsc{PICEF}$ & Mean & $7.08$& $9.33$& $13.04$& $17.54$& $22.90$& $28.08$& $36.50$& $43.25$& $52.92$& $64.75$& $75.29$\\
& Stdev & $1.56$& $1.61$& $2.52$& $3.43$& $5.08$& $5.70$& $5.79$& $7.13$& $10.70$& $11.86$& $20.55$\\
& Min& $4.53$& $6.40$& $9.09$& $11.87$& $16.00$& $19.19$& $25.22$& $27.78$& $34.55$& $49.41$& $48.34$\\
 & Max& $11.31$& $13.87$& $18.28$& $24.26$& $38.67$& $38.91$& $44.74$& $59.57$& $76.93$& $91.35$& $107.55$\\
\hline\\\\
$\textsc{HPIEF}$ & Mean & $14.00$& $17.78$& $23.48$& $32.89$& $44.44$& $61.14$& $79.91$& $103.76$& $139.35$& $174.60$& $193.66$\\
& Stdev & $3.08$& $5.18$& $5.83$& $11.61$& $13.89$& $23.56$& $27.10$& $37.86$& $53.76$& $64.40$& $72.27$\\
& Min& $10.83$& $13.29$& $16.17$& $19.70$& $22.02$& $25.65$& $30.00$& $35.09$& $39.42$& $46.47$& $59.89$\\
 & Max& $26.49$& $39.65$& $42.38$& $73.79$& $68.14$& $126.82$& $117.08$& $163.35$& $230.51$& $255.93$& $291.70$\\
\hline\\\\
$\textsc{BnP-Poly}$ & Mean & $87.19$& $18.76$& $97.20$& $28.25$& $38.46$& $100.95$& $205.02$& $148.89$& $661.64$& $908.36$& $1020.85$\\
& Stdev & $350.96$& $1.91$& $350.66$& $3.94$& $27.81$& $188.11$& $491.05$& $201.06$& $767.87$& $852.81$& $752.28$\\
& Min& $11.04$& $14.57$& $19.19$& $22.15$& $24.80$& $25.13$& $26.05$& $32.74$& $35.45$& $38.06$& $133.57$\\
 & Max& $1805.96$& $22.00$& $1815.00$& $35.22$& $171.97$& $928.79$& $1897.63$& $847.83$& $1892.30$& $2474.09$& $2765.64$\\
\hline\\\\
$\textsc{CG-TSP}$ & Mean & $3600.55$& $3600.76$& $3600.71$& $3600.81$& $3600.80$& $3600.69$& $3600.86$& $3467.07$& $3328.12$& $3206.21$& $2258.82$\\
& Stdev & $0.34$& $0.46$& $0.41$& $0.36$& $0.45$& $0.38$& $0.54$& $454.31$& $625.70$& $708.71$& $1165.52$\\
& Min& $3600.15$& $3600.26$& $3600.27$& $3600.28$& $3600.21$& $3600.21$& $3600.24$& $1838.14$& $1825.68$& $1847.26$& $102.74$\\
 & Max& $3601.35$& $3601.98$& $3601.80$& $3601.61$& $3601.98$& $3601.54$& $3602.45$& $3602.05$& $3601.67$& $3601.78$& $3601.04$\\
\hline\\\\
$\textsc{BnP-DFS}$ & Mean & $125.75$& $2029.37$& $3600.01$& $3600.00$& $3600.00$& $3600.00$& $3600.00$& $3600.00$& $3600.00$& $3600.00$& $3600.00$\\
& Stdev & $346.25$& $466.77$& $0.00$& $0.00$& $0.00$& $0.00$& $0.00$& $0.00$& $0.00$& $0.00$& $0.00$\\
& Min& $32.45$& $1048.41$& $3600.00$& $3600.00$& $3600.00$& $3600.00$& $3600.00$& $3600.00$& $3600.00$& $3600.00$& $3600.00$\\
 & Max& $1821.02$& $2831.76$& $3600.02$& $3600.00$& $3600.00$& $3600.00$& $3600.00$& $3600.00$& $3600.00$& $3600.00$& $3600.00$\\
\hline\\\\
$\textsc{BnP-PICEF}$ & Mean & $78.89$& $8.18$& $13.75$& $26.74$& $37.15$& $48.15$& $80.05$& $100.20$& $140.66$& $187.31$& $238.27$\\
& Stdev & $351.95$& $1.51$& $5.48$& $7.94$& $9.18$& $11.02$& $19.44$& $22.66$& $34.57$& $37.81$& $46.79$\\
& Min& $5.68$& $6.31$& $6.76$& $10.52$& $20.57$& $22.27$& $53.48$& $56.36$& $76.52$& $122.53$& $145.24$\\
 & Max& $1803.09$& $11.65$& $26.37$& $39.92$& $54.66$& $70.94$& $126.77$& $140.54$& $219.19$& $254.26$& $332.33$\\
\hline\\\\
\end{tabular}

\label{tbl:data}
\end{table}

    \clearpage\thispagestyle{plain}\begin{table}[ht!bp]
\centering
\caption{$|P|~=~700,~|A|~=~7$}
\begin{tabular}
{ c r | r  | r  | r  | r  | r  | r  | r  | r  | r  | r  | r }
Method & & $L = 2$ & $L = 3$ & $L = 4$ & $L = 5$ & $L = 6$ & $L = 7$ & $L = 8$ & $L = 9$ & $L = 10$ & $L = 11$ & $L = 12$\\
\hline\\\\
$\textsc{PICEF}$ & Mean & $9.81$& $16.64$& $27.03$& $45.15$& $85.77$& $132.08$& $221.64$& $476.49$& $746.67$& $1060.93$& $1391.81$\\
& Stdev & $3.96$& $8.86$& $14.48$& $15.45$& $36.51$& $54.10$& $143.49$& $473.81$& $674.82$& $624.99$& $953.63$\\
& Min& $6.10$& $10.23$& $16.92$& $28.10$& $46.56$& $64.96$& $83.36$& $123.93$& $153.49$& $229.89$& $310.20$\\
 & Max& $22.29$& $48.46$& $79.24$& $96.61$& $170.04$& $236.84$& $752.40$& $2045.09$& $2417.46$& $2157.87$& $3600.01$\\
\hline\\\\
$\textsc{HPIEF}$ & Mean & $47.46$& $60.30$& $71.00$& $93.18$& $141.93$& $194.25$& $301.22$& $851.62$& $874.49$& $1214.84$& $1688.33$\\
& Stdev & $20.40$& $31.39$& $35.59$& $51.30$& $125.99$& $128.11$& $286.29$& $784.70$& $735.08$& $978.75$& $1196.99$\\
& Min& $24.81$& $28.72$& $34.29$& $42.88$& $52.14$& $77.45$& $86.79$& $89.34$& $101.09$& $175.17$& $149.92$\\
 & Max& $89.99$& $130.90$& $142.10$& $189.76$& $649.25$& $583.77$& $1232.76$& $2557.22$& $3057.21$& $3600.02$& $3600.02$\\
\hline\\\\
$\textsc{BnP-Poly}$ & Mean & $1521.97$& $1159.34$& $1734.26$& $2239.58$& $2454.59$& $2979.44$& $3269.40$& $3274.71$& $3518.71$& $3459.88$& $3420.45$\\
& Stdev & $976.24$& $1121.69$& $1074.63$& $1364.66$& $1328.86$& $834.51$& $809.68$& $656.30$& $316.35$& $475.76$& $496.20$\\
& Min& $6.95$& $7.16$& $8.84$& $11.51$& $12.04$& $1807.77$& $323.14$& $1812.67$& $2037.16$& $1826.98$& $1819.21$\\
 & Max& $3600.07$& $3600.06$& $3600.08$& $3600.23$& $3600.37$& $3600.37$& $3600.38$& $3600.46$& $3600.32$& $3600.55$& $3600.81$\\
\hline\\\\
$\textsc{BnP-PICEF}$ & Mean & $1597.19$& $1239.85$& $1970.94$& $2557.71$& $2661.05$& $3094.01$& $3426.76$& $3436.00$& $3551.63$& $3566.88$& $3600.04$\\
& Stdev & $1054.23$& $1099.56$& $1112.42$& $1099.13$& $1262.80$& $738.47$& $469.54$& $444.97$& $237.36$& $162.52$& $0.01$\\
& Min& $10.47$& $18.77$& $41.81$& $114.79$& $183.42$& $1959.23$& $2143.36$& $2172.23$& $2388.82$& $2770.72$& $3600.01$\\
 & Max& $3601.19$& $3601.02$& $3600.93$& $3601.05$& $3600.62$& $3600.34$& $3600.21$& $3600.14$& $3600.10$& $3600.10$& $3600.06$\\
\hline\\\\
\end{tabular}

\label{tbl:data}
\end{table}

    \clearpage\thispagestyle{plain}\begin{table}[ht!bp]
\centering
\caption{$|P|~=~700,~|A|~=~14$}
\begin{tabular}
{ c r | r  | r  | r  | r  | r  | r  | r  | r  | r  | r  | r }
Method & & $L = 2$ & $L = 3$ & $L = 4$ & $L = 5$ & $L = 6$ & $L = 7$ & $L = 8$ & $L = 9$ & $L = 10$ & $L = 11$ & $L = 12$\\
\hline\\\\
$\textsc{PICEF}$ & Mean & $14.57$& $23.52$& $38.46$& $65.49$& $143.29$& $340.24$& $694.99$& $1085.51$& $1036.88$& $1016.48$& $1027.84$\\
& Stdev & $5.69$& $9.43$& $15.64$& $28.21$& $119.45$& $267.07$& $540.55$& $756.88$& $521.12$& $390.71$& $412.54$\\
& Min& $7.45$& $12.76$& $20.75$& $28.59$& $46.11$& $53.92$& $65.29$& $170.95$& $354.81$& $290.44$& $385.99$\\
 & Max& $30.89$& $55.16$& $82.82$& $128.23$& $656.89$& $1078.52$& $2192.32$& $2548.76$& $2269.19$& $1769.58$& $1887.18$\\
\hline\\\\
$\textsc{HPIEF}$ & Mean & $39.05$& $51.21$& $62.70$& $93.00$& $240.20$& $564.15$& $1038.49$& $1413.02$& $1444.70$& $1473.19$& $1283.89$\\
& Stdev & $9.05$& $23.86$& $24.46$& $43.04$& $249.84$& $484.44$& $719.21$& $952.48$& $816.46$& $716.48$& $527.02$\\
& Min& $30.31$& $32.43$& $40.51$& $45.95$& $62.73$& $102.38$& $105.83$& $232.39$& $281.55$& $335.61$& $516.76$\\
 & Max& $63.12$& $134.94$& $129.85$& $208.47$& $1241.43$& $1786.57$& $2751.71$& $3600.02$& $2877.35$& $2757.60$& $2371.12$\\
\hline\\\\
$\textsc{BnP-Poly}$ & Mean & $1017.01$& $1519.12$& $2381.49$& $2884.76$& $3315.11$& $3266.22$& $3348.78$& $3119.47$& $3418.67$& $3600.08$& $3600.08$\\
& Stdev & $1250.38$& $1406.13$& $977.20$& $1011.44$& $653.27$& $670.13$& $578.91$& $723.61$& $447.02$& $0.05$& $0.04$\\
& Min& $6.83$& $8.59$& $22.14$& $28.53$& $1814.72$& $1822.70$& $1894.07$& $1842.47$& $2084.25$& $3600.02$& $3600.02$\\
 & Max& $3600.06$& $3600.13$& $3600.21$& $3600.34$& $3600.43$& $3600.38$& $3600.39$& $3600.30$& $3600.26$& $3600.21$& $3600.17$\\
\hline\\\\
$\textsc{BnP-PICEF}$ & Mean & $1264.28$& $1612.64$& $2407.02$& $2921.08$& $3412.63$& $3390.81$& $3346.30$& $3378.98$& $3478.75$& $3586.62$& $3600.02$\\
& Stdev & $1295.03$& $1272.61$& $1189.48$& $960.60$& $509.38$& $436.66$& $396.00$& $313.42$& $305.45$& $58.48$& $0.01$\\
& Min& $9.71$& $18.89$& $76.04$& $208.48$& $1957.17$& $2259.70$& $2220.26$& $2502.96$& $2166.24$& $3302.37$& $3600.01$\\
 & Max& $3600.80$& $3600.83$& $3600.68$& $3600.62$& $3600.37$& $3600.24$& $3600.15$& $3600.08$& $3600.07$& $3600.05$& $3600.05$\\
\hline\\\\
\end{tabular}

\label{tbl:data}
\end{table}

    \clearpage\thispagestyle{plain}\begin{table}[ht!bp]
\centering
\caption{$|P|~=~700,~|A|~=~35$}
\begin{tabular}
{ c r | r  | r  | r  | r  | r  | r  | r  | r  | r  | r  | r }
Method & & $L = 2$ & $L = 3$ & $L = 4$ & $L = 5$ & $L = 6$ & $L = 7$ & $L = 8$ & $L = 9$ & $L = 10$ & $L = 11$ & $L = 12$\\
\hline\\\\
$\textsc{PICEF}$ & Mean & $13.13$& $23.38$& $55.18$& $99.27$& $203.74$& $317.91$& $547.05$& $797.67$& $1116.34$& $1293.25$& $1282.57$\\
& Stdev & $7.45$& $11.99$& $34.73$& $34.95$& $49.29$& $70.34$& $181.52$& $242.31$& $516.93$& $432.45$& $501.84$\\
& Min& $8.25$& $14.21$& $24.31$& $42.27$& $72.62$& $113.05$& $131.92$& $179.43$& $239.43$& $264.00$& $315.97$\\
 & Max& $47.76$& $75.74$& $194.49$& $208.56$& $276.76$& $462.41$& $940.13$& $1179.01$& $2736.66$& $1896.02$& $2673.03$\\
\hline\\\\
$\textsc{HPIEF}$ & Mean & $34.37$& $48.26$& $110.16$& $186.29$& $305.20$& $507.13$& $826.60$& $1171.59$& $1566.82$& $1763.52$& $1506.63$\\
& Stdev & $3.00$& $11.15$& $77.52$& $45.19$& $60.17$& $108.26$& $238.30$& $248.72$& $313.97$& $509.08$& $605.57$\\
& Min& $29.09$& $36.88$& $54.65$& $95.82$& $220.89$& $352.13$& $449.15$& $825.67$& $1012.74$& $1000.12$& $502.45$\\
 & Max& $39.79$& $87.52$& $420.49$& $353.66$& $499.14$& $770.13$& $1652.82$& $1810.44$& $2481.68$& $2796.51$& $2489.86$\\
\hline\\\\
$\textsc{BnP-Poly}$ & Mean & $946.39$& $968.40$& $1175.60$& $489.31$& $408.89$& $987.92$& $1204.20$& $1807.54$& $1940.38$& $2415.57$& $3181.40$\\
& Stdev & $1147.60$& $1134.31$& $1114.13$& $758.29$& $656.11$& $843.59$& $947.76$& $931.08$& $1101.74$& $903.48$& $679.19$\\
& Min& $10.57$& $14.50$& $28.51$& $33.90$& $37.49$& $43.77$& $86.42$& $101.46$& $139.30$& $244.96$& $1887.27$\\
 & Max& $3600.14$& $3600.09$& $3600.35$& $1894.51$& $1866.21$& $1977.77$& $3600.26$& $3600.20$& $3600.18$& $3600.13$& $3600.15$\\
\hline\\\\
$\textsc{BnP-PICEF}$ & Mean & $1020.98$& $969.54$& $1190.59$& $833.25$& $1222.95$& $1820.27$& $2533.52$& $3101.86$& $3430.50$& $3472.62$& $3294.55$\\
& Stdev & $1249.82$& $1137.59$& $975.36$& $712.50$& $476.76$& $701.44$& $481.35$& $412.11$& $251.68$& $241.88$& $380.11$\\
& Min& $13.18$& $25.06$& $98.13$& $260.07$& $587.42$& $815.88$& $1277.61$& $2278.87$& $2662.25$& $2886.00$& $2357.03$\\
 & Max& $3600.62$& $3600.69$& $3600.35$& $2321.96$& $2240.81$& $3600.02$& $3600.02$& $3600.03$& $3600.07$& $3600.04$& $3600.05$\\
\hline\\\\
\end{tabular}

\label{tbl:data}
\end{table}
    
    \clearpage\thispagestyle{plain}\begin{table}[ht!bp]
\centering
\caption{$|P|~=~700,~|A|~=~175$}
\begin{tabular}
{ c r | r  | r  | r  | r  | r  | r  | r  | r  | r  | r  | r }
Method & & $L = 2$ & $L = 3$ & $L = 4$ & $L = 5$ & $L = 6$ & $L = 7$ & $L = 8$ & $L = 9$ & $L = 10$ & $L = 11$ & $L = 12$\\
\hline\\\\
$\textsc{PICEF}$ & Mean & $29.54$& $43.30$& $71.41$& $105.27$& $151.55$& $203.83$& $286.35$& $374.32$& $467.57$& $560.50$& $681.95$\\
& Stdev & $13.36$& $16.27$& $29.08$& $44.51$& $66.12$& $87.99$& $125.80$& $166.32$& $215.48$& $271.71$& $370.30$\\
& Min& $14.44$& $19.95$& $28.55$& $33.10$& $44.67$& $58.77$& $82.04$& $104.68$& $119.57$& $156.31$& $137.36$\\
 & Max& $63.37$& $72.10$& $112.54$& $161.79$& $261.19$& $336.25$& $493.77$& $644.54$& $807.41$& $990.81$& $1278.76$\\
\hline\\\\
$\textsc{HPIEF}$ & Mean & $39.67$& $55.28$& $93.24$& $137.82$& $200.39$& $277.33$& $378.82$& $533.55$& $642.43$& $874.87$& $908.95$\\
& Stdev & $5.51$& $13.00$& $34.40$& $43.98$& $73.27$& $105.09$& $146.86$& $233.83$& $266.91$& $385.84$& $363.31$\\
& Min& $31.07$& $37.56$& $45.29$& $55.53$& $69.20$& $84.06$& $103.58$& $114.62$& $130.30$& $165.84$& $169.59$\\
 & Max& $57.05$& $94.99$& $182.11$& $211.11$& $331.98$& $512.84$& $602.28$& $870.66$& $1023.77$& $1297.35$& $1350.16$\\
\hline\\\\
$\textsc{BnP-Poly}$ & Mean & $50.03$& $77.46$& $96.51$& $102.51$& $146.29$& $294.77$& $712.95$& $936.95$& $1281.41$& $1608.44$& $2130.64$\\
& Stdev & $6.49$& $12.00$& $13.69$& $16.35$& $95.95$& $464.74$& $769.68$& $807.90$& $871.01$& $1016.13$& $1106.75$\\
& Min& $36.48$& $55.54$& $67.21$& $80.60$& $92.13$& $102.89$& $126.45$& $127.95$& $153.58$& $293.59$& $257.74$\\
 & Max& $61.19$& $100.42$& $124.90$& $130.75$& $603.92$& $1880.82$& $2189.93$& $2001.10$& $3600.04$& $3600.05$& $3600.04$\\
\hline\\\\
$\textsc{BnP-PICEF}$ & Mean & $18.37$& $25.95$& $58.41$& $101.92$& $133.68$& $200.45$& $298.62$& $396.98$& $625.76$& $963.66$& $1079.16$\\
& Stdev & $2.09$& $6.89$& $19.69$& $25.51$& $28.12$& $46.29$& $57.63$& $99.96$& $274.92$& $353.71$& $359.13$\\
& Min& $14.32$& $17.54$& $25.16$& $43.12$& $69.15$& $121.50$& $199.20$& $255.27$& $360.26$& $513.35$& $569.38$\\
 & Max& $23.72$& $38.25$& $100.01$& $146.01$& $205.94$& $278.30$& $477.24$& $656.46$& $1539.05$& $2064.30$& $2096.84$\\
\hline\\\\
\end{tabular}

\label{tbl:data}
\end{table}
    
  }
\end{landscape}

}{}

\end{document}